\documentclass[a4paper,UKenglish,cleveref,autoref,thm-restate,a4paper]{lipics-v2021}
\hypersetup{hidelinks}

\usepackage{ifthen}
\newboolean{arxiv} 
\setboolean{arxiv}{true}
\newcommand{\arxiv}[2]{\ifthenelse{\boolean{arxiv}}{#1}{#2}}

\usepackage{tikz}
	\usetikzlibrary{
		arrows,
		positioning,
	}

\tikzset{state/.style={circle,fill=black!45,draw=black!85,inner sep=0pt,minimum size=2mm}}
\tikzset{tran/.style={>=stealth',shorten >=0.5pt,thick, draw=black!45}}
\tikzset{signal/.style={>=stealth',shorten >=0.5pt, dotted}}
\tikzset{bst/.style={>=stealth', dotted, purple}}
\tikzset{btr/.style={>=stealth', dotted, cyan}}
\tikzset{every label/.style={black}}
\tikzset{every node/.style={black}}

\tikzset{place/.style={circle,thick,draw=black!75,fill=black!10,minimum size=4mm}}
\tikzset{red place/.style={place,draw=red!75,fill=red!20}}
\tikzset{transition/.style={rectangle,thick,draw=black!75,fill=black!20,minimum size=4mm}}

\usepackage{wrapfig}
\usepackage{graphicx}
\usepackage{listings}
\usepackage{xspace}
\usepackage{xcolor}
\usepackage{amsmath,amssymb,latexsym}
\usepackage{mathtools}
\newcommand{\pip}{\rho}
\newcommand{\hatpip}{\hat\rho}
\newcommand{\sR}{O}
\usepackage{notations} 
\newcommand{\transition}{transition}
\newcommand{\ABCdE}{ABCdE\xspace}
\newcommand{\premise}{direct subderivation}
\renewcommand{\nabla}{\Tr}
\newcommand{\nablasub}[1]{\en(#1)}
\newcommand{\ie}{i.e.,\xspace}
\IfFileExists{fontawesome5.sty}{}{
  \def\orcidsymbol{\includegraphics[height=9px]{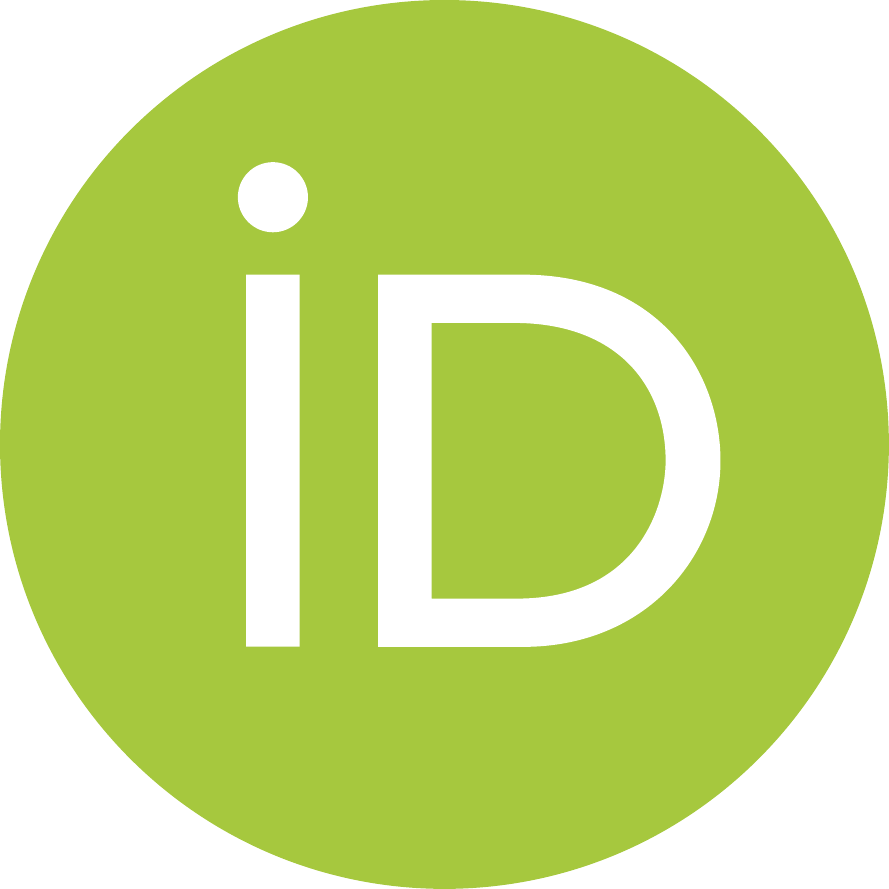}}
  \IfFileExists{fontawesome.sty}{
    \RequirePackage{fontawesome}
    \def\mailsymbol{\textcolor{lipicsGray}{\fontsize{9}{12}\sffamily\bfseries \faEnvelopeO}}
  }{}
}

\title{Enabling Preserving Bisimulation Equivalence}

\author{Rob van Glabbeek}{Data61, CSIRO, Sydney, Australia \and Computer Science and Engineering, University of New South Wales, Sydney, Australia}{rvg@cs.stanford.edu}{}{}
\author{Peter H\"ofner}{School of Computing, Australian National University, Canberra, Australia \and Data61, CSIRO, Sydney, Australia}{peter.hoefner@anu.edu.au}{https://orcid.org/0000-0002-2141-5868}{}
\author{Weiyou Wang}{School of Computing, Australian National University, Canberra, Australia}{weiyou.wang@anu.edu.au}{}{}

\authorrunning{R.J. van Glabbeek, P. H\"ofner and W. Wang} 

\Copyright{Rob van Glabbeek, Peter H\"ofner and Weiyou Wang} 



\relatedversion{\arxiv{A version of this paper without Appendix~\ref{Synchrons} appears in Proc.\ CONCUR'21.}
                           {A full version of this paper is available as \cite{GHW21}.}}

\acknowledgements{We are grateful to Filippo De Bortoli.}

\nolinenumbers

\arxiv{\hideLIPIcs}{}

\EventEditors{Serge Haddad and Daniele Varacca}
\EventNoEds{2}
\EventLongTitle{32nd International Conference on Concurrency Theory (CONCUR 2021)}
\EventShortTitle{CONCUR 2021}
\EventAcronym{CONCUR}
\EventYear{2021}
\EventDate{August 23--27, 2021}
\EventLocation{Virtual, The Earth}
\EventLogo{}
\SeriesVolume{42}
\ArticleNo{33}

\begin{document}

\maketitle
\begin{abstract}
  Most fairness assumptions used for verifying liveness properties are criticised for being too strong or unrealistic.
  On the other hand, justness, arguably the minimal fairness assumption required for the verification of liveness properties,
  is not preserved by classical semantic equivalences, such as strong bisimilarity.
  To overcome this deficiency, we introduce a finer alternative to strong bisimilarity, called enabling preserving bisimilarity.
  We prove that this equivalence is justness-preserving and a congruence for all standard operators, including parallel composition.
\end{abstract}

\section{Introduction}\label{sec:introduction}
Formal verification of concurrent systems becomes more and more standard practice, in~par\-ticular in safety-critical environments.
Progress and fairness assumptions have to be used~when verifying liveness properties, which
guarantee that `something good will eventually happen'. Without assumptions of this kind,
no meaningful liveness property can formally be~proven.%

\begin{example}
Consider the program \lstinline{while(true)  do x:=x+1 od} with \lstinline{x} initialised to \lstinline{0}.
Intuitively, any liveness property of the form `eventually \lstinline{x=n}' should be satisfied by the program.
However, these properties are valid only when assuming \emph{progress}, stating that a system will make progress when it can;
otherwise the program could just stop after some computation.~~\hfill$\lrcorner$
\end{example}

Progress itself is not a strong enough assumption when concurrent systems are verified, for a system of multiple
completely independent components makes progress as long as one of its components makes progress, even
when others do not. For decades, researchers have developed notions of fairness and used them in both
system specification and verification;
the most common ones are surveyed in~\cite{GH19}. Two of the most popular fairness assumptions are weak and strong fairness of 
instructions \cite{Fr86}.\footnote{Often these notions are referred to as weak and strong fairness without mentioning instructions;
 here, we follow the terminology of \cite{GH19}, which is more precise.} They apply to systems whose
behaviour is specified by some kind of code, composed out of \emph{instructions}. A \emph{task} is
any activity of the system that stems from a particular instruction; it is \emph{enabled} when the system
is ready to do some part of that task, and \emph{executed} when the system performs some part of it.
Now weak and strong fairness of instructions state that whenever a task is
enabled persistently (for weak fairness) or infinitely often (for strong fairness),
then it will be executed by the system. 
These fairness assumptions, as well as all others surveyed in \cite{GH19},\footnote{Many other
  notions of fairness are obtained by varying the definition of task. In \emph{fairness of components}
  a task refers to all activity stemming from a component of a system that is a parallel composition.}
imply progress.

Despite being commonly used, it has been argued that most fairness assumptions, including weak and strong fairness of instructions, are often too strong or unrealistic, ``in the sense that run-of-the-mill implementations tend not to be fair''~\cite{GH19}.

(Reactive) systems are often described by \emph{labelled transition systems}, which model all activities as transitions going from state to state, labelled with \emph{actions}. Some actions require synchronisation of the modelled system with its environment; they can occur only when both the system and the environment are ready to engage in this action.
Such actions, and the transitions labelled with them, are called \emph{blocking}.
\begin{example}\label{ex:breakfast} 
Assume that every morning Alice has a choice between a slice of bread with jam or a bacon and egg roll.
A corresponding transition system consists of one state with two transitions, each standing for one
 kind of breakfast. 
Both weak and strong fairness (of instructions) will force Alice to eventually have both types of breakfast, ruling out 
the possibility that Alice picks up jam every day as she is a vegetarian.\hfill$\lrcorner$
\end{example}
To address this issue, a weaker assumption, called justness, has been proposed.
It has been formulated for reactive systems, modelled as labelled transition systems.
\emph{Justness} is the assumption that
\vspace{-1mm}
\begin{quote}\it
{Once a non-blocking transition is enabled that stems from a set of parallel components, one (or more) of these components will eventually partake in a transition.~\cite{GH19}}
\end{quote}

\noindent
\ex{breakfast} features only one component, Alice. Assuming justness, as expected, 
she now has the option to eat jam for the rest of her life. Let us now look at a more technical example.

\begin{example}\label{ex:filippo}
We consider the following two programs, and assume that all variables are initialised by \lstinline{0}.
\vspace{-5mm}
\begin{center}\hfill
	\begin{minipage}[t]{0.4\textwidth}
\begin{lstlisting}
while(true) do
  choose
    if true then y := y+1;
    if x = 0 then x := 1;
  end
od
\end{lstlisting}
	\end{minipage}\hfill\hfill
	\begin{minipage}[t]{0.22\textwidth}
\begin{lstlisting}
while(true) do
  y := y+1;
od
\end{lstlisting}
	\end{minipage}\;
	\raisebox{-30pt}{\scalebox{1.7}[4]{$\|$}}
	\begin{minipage}[t]{0.15\textwidth} 
\begin{lstlisting}
x := 1;
\end{lstlisting}
	\end{minipage}
	\hfill\mbox{}
\end{center}
\vspace{-8mm}

\noindent The example on the left presents an infinite loop containing an internal nondeterministic choice.
The conditional write \lstinline{if x = 0 then x := 1} describes an atomic read-modify-write (RMW)
operation\footnote{\url{https://en.wikipedia.org/wiki/Read-modify-write}}.
Such operators, supported by modern hardware, read a memory location and simultaneously write a new value into it.
This example is similar to \ex{breakfast} in the sense that the liveness property 
`eventually \lstinline{x=1}' should not be satisfied as the program has a choice every time the loop body 
is executed.

The example on the right-hand side is similar, but the handling of variables \lstinline{x} and 
\lstinline{y} are managed by different components. As the two programs are independent from each other -- they could be executed on different machines -- the property `eventually \lstinline{x=1}' should hold.

Justness differentiates these behaviour, whereas weak and strong fairness fail to do so.~~\hfill$\lrcorner$
\end{example}

\noindent
The above example illustrates that standard notions of fairness are regularly too strong,\arxiv{\pagebreak[2]}{} and the
notion of justness may be a good replacement. When it comes to verification tasks, semantic
equivalences, such as strong bisimilarity~\cite{Mi90ccs}, are a standard tool to
reduce the state space of the systems under consideration.
Unfortunately, these semantic equivalences do not accord well with justness.
The problem is that they are based on labelled transition systems, 
which do not capture the concept of components.
The different behaviour of the two programs 
\begin{wrapfigure}[3]{r}{0.25\textwidth}
	\vspace{-0.5cm}
	\centering
	\begin{tikzpicture}[
		on grid=true,
		node distance=1.5cm,
		>=stealth',
		bend angle=45,
		auto]

	\begin{scope}
		\node[state] (r1) {};
		\node[state] (r2) [right of=r1] {};

		\path
		(r1)
			edge[->, tran, loop below, in=230, out=310, looseness=10] node [below] {\scriptsize $y := y + 1$} (r1)
			edge[->, tran] node [above] {\scriptsize $x := 1$} (r2)
		(r2)
			edge[->, tran, loop below, in=230, out=310, looseness=10] node[below] {\scriptsize $y := y + 1$} (r2);
	\end{scope}
	\end{tikzpicture}
\end{wrapfigure}
 in \ex{filippo} stems from the
components involved.
In fact, both programs give rise to the same transition system, depicted on the right.
Systems featuring the same transition system cannot be distinguished by any semantic equivalence 
found in the literature. Consequently, the verification of the stated liveness property will fail for one of the two 
programs of \ex{filippo}. 

To overcome this deficiency, we introduce \emph{enabling preserving bisimilarity}, a finer alternative to strong bisimilarity, which respects justness. It is based on extended labelled transition systems that take components involved in particular transitions into account.

\section{Labelled Transition Systems with Successors}\label{sec:ltss}

As discussed in the introduction, one reason why strong bisimilarity does not preserve liveness
properties under justness is that necessary information is missing, namely components. 

The definition of \emph{(parallel) components} was based on the parallel composition operator in process algebras
when justness was first introduced in~\cite{TR13,GH15a},
and has been generalised in later work to allow the use of justness in different contexts.

Here we define a justness-preserving  semantic equivalence on an extension of labelled transition systems.
Using labelled transition systems rather than process algebra as underlying concept makes our
approach more general, for other models of concurrency, such as Petri nets or higher-dimensional
automata, induce a semantics based on transition systems as~well.

The essence of justness is that when a non-blocking transition $t$ is enabled in a state $s$,
eventually the system must perform an action $u$ that \emph{interferes} with it \cite{GH19},
notation $t \naconc u$, in the sense that
a component affected by $u$ is necessary for the execution of $t$ -- or, to be more precise, for the
variant $t'$ of $t$ that is enabled after the system has executed some actions that do not interfere with $t$.
The present paper abstracts from the notion of component, but formalises justness, as well as our
enabling preserving bisimilarity, in terms of a \emph{successor relation} $t \pleadsto_{\pi} t'$,
marking $t'$ as a successor of $t$,
parametrised with the noninterfering actions $\pi$ happening in between. This relation also encodes
the above relation $\naconc$. The advantage of this approach over one that uses components
explicitly, is that it also applies to models like higher-dimensional automata \cite{Pr91a,vG91,GJ92,vG06}
in which the notion of a component is more fluid, and changes during execution.

  A \emph{labelled transition system} (LTS) is a tuple $(S,\Tr,\source,\target,\ell)$ with $S$ and $\Tr$ sets
  (of \emph{states} and \emph{transitions}), $\source,\target:\Tr\to S$ and
  $\ell:\Tr\to\Lab$, for some set $\Lab$ of transition labels.
  A transition $t\in\Tr$ of an LTS is \emph{enabled} in a state $p\in S$ if $\source(t)=p$.
  Let $\en(p)$ be the set of transitions that are enabled in $p$.

  A \emph{path} in an LTS is an alternating sequence
  $p_0\,u_0\,p_1\,u_1\,p_2\dots$ of states and transitions,
  starting with a state and either being infinite or ending with a state,
  such that $\source(u_i)=p_i$ and $\target(u_i)=p_{i+1}$ for all relevant $i$.
  The \emph{length} $l(\pi)\in\IN\cup\{\infty\}$ of a path $\pi$ is the number of transitions in it.
  If $\pi$ is a path, then $\hat\pi$ is the sequence of transitions occurring in $\pi$.

\begin{definition}[LTSS]\label{df:LTSS}\rm
  A \emph{labelled transition system with successors} (LTSS) is a tuple
  $(S,\Tr,\source,\target,\ell,\leadsto)$ with
  $(S,\Tr,\source,\target,\ell)$ an LTS, and
  ${\leadsto} \subseteq \Tr \times \Tr \times \Tr$, the \emph{successor relation},\arxiv{}{\pagebreak[3]}
  such that if $(t,u,v)\in{\leadsto}$ then $\source(t)=\source(u)$ and $\source(v)=\target(u)$. We write $t \leadsto_u v$ for $(t,u,v)\in{\leadsto}$.
\end{definition}
We use this successor relation to define 
the concept of (un)affected transitions. Let two transitions $t$ and $u$ be enabled in a state $p$, 
\ie $t,u\in\en(p)$ for some $p\in S$; the \emph{concurrency relation} $\aconc$ is defined as
$
t \aconc u :\Leftrightarrow \exists\, v.~ t \leadsto_u v.
$
Its negation $t \naconc u$ says that the possible occurrence of $t$ is \emph{affected} by the occurrence of $u$.
In case $t$ is unaffected by $u$ (\ie $t \aconc u$), each $v$ with $t \leadsto_u v$ denotes a variant of $t$ that can occur after $u$. 
Note that the concurrency relation can be asymmetric. Examples are traffic lights -- a car passing
traffic lights should be affected by them, but the lights do not care whether the car is there; and read-write operations -- 
reading shared memory can be affected by a write action, but, depending on how the memory is implemented, 
the opposite might not hold. 
In case $t$ and $u$ are mutually unaffected we write $t\conc u$, \ie $t\conc u :\Leftrightarrow t\aconc u \land u\aconc t$.

It is possible to have $t \aconc t$, namely when executing transition $t$ does not disable (a
  future variant of) $t$ to occur again. This can happen when $t$ is a signal emission, say of a
  traffic light shining red, for even after shining for some time it keeps on shining; or when
  $t$ is a broadcast receive action, for receiving a broadcast does not invalidate a system's perpetual
  readiness to again accept a broadcast, either by receiving or ignoring it.

\begin{example}
Consider the 
labelled transition system of \ex{filippo}, let $t_1$ and $t_2$ be the two transitions 
corresponding to \lstinline{y:=y+1} in the first and second state, respectively, and let $u$ be the transition 
for assignment \lstinline{x:=1}. The assignments of \lstinline{x} and \lstinline{y} in the right-hand program are independent, hence $t_1 \leadsto_u t_2$, $u\leadsto_{t_1} u$ and $t_1\conc u$.

For the program on the left-hand side, the situation is different. As the
instructions stem from the same component (program), all transitions affect each other, \ie $\leadsto\;=\; \aconc\; =\; \emptyset$. \hfill$\lrcorner$
\end{example}

\noindent
The successor relation relates transitions one step apart.
We lift it to sequences of transitions.

\begin{definition}[Successor along Path]\label{df:after}\rm
  The relation $\leadsto$ is extended to ${\pleadsto} \subseteq \Tr \times \Tr^* \times \Tr$ by
  (i) $t \pleadsto_\varepsilon w$ iff $w=t$, and
  (ii) $t \pleadsto_{\pi u} w$ iff there is a $v$ with $t \pleadsto_{\pi} v$ and $v\leadsto_u w$.
\end{definition}
  Here, $\varepsilon$ denotes the empty sequence, and $\pi u$ the sequence $\pi$ followed by
  transition $u$. We define a concurrency relation ${\paconc} \subseteq \Tr\times\Tr^*$
  considering sequences of transitions by
  $t \paconc \pi :\Leftrightarrow \exists\,v.~ t \pleadsto_\pi v$.
  Intuitively, $t \paconc \pi$ means that there exists a successor of $t$ after $\pi$ has been executed. Thus,
  $t$ is unaffected by all transitions of $\pi$.

  We are ready to define justness, which is parametrised by a set $B$ of blocking actions.

\begin{definition}[Justness]\label{df:justness}\rm
  Given an LTSS $\C = (S,\Tr,\source,\target,\ell,\leadsto)$ labelled over $\Lab$, and $B\subseteq\Lab$,
  a path $\pi$ in $\C$ is \emph{$B$-just} if for each suffix $\pi_0$ of $\pi$ and for each transition
  $t\in\Tr^\bullet$ with $\ell(t)\notin B$ and $\source(t)$ the first state of $\pi_0$,
  the path $\pi_0$ has a finite prefix $\rho$ such that $t \pnaconc \hat\rho$.
  Here $\Tr^\bullet \coloneqq \{t\in\Tr \mid t \naconc t\}$.
\end{definition}

\section{Enabling Preserving Bisimulation Equivalence}\label{sec:ep-bimilarity}

In this section we introduce enabling preserving bisimulation equivalence, and show how it preserves justness.
In contrast to classical  bisimulations, which are relations of type $S\times S$,
the new equivalence is based on triples. 
The essence of justness is that a transition $t$ enabled in a state $s$ must
eventually be affected by the sequence $\pi$ of transitions the system performs.
As long as $\pi$ does not interfere with $t$, we obtain a transition $t'$ with $t \pleadsto_\pi t'$.
This transition $t'$ represents the interests of $t$, and must eventually be affected by an
extension of $\pi$. Here, executing $t$ or $t'$ as part of this extension is a valid way of interfering.
To create a bisimulation that respects such considerations, for each related
pair of states $p$ and $q$ we also match each enabled transition of $p$ with one of $q$, and
vice versa. These relations are maintained during the evolution of the two related systems,
so that when one system finally interferes with a descendant of $t$, the related system interferes
with the related descendant.

\begin{definition}[Ep-bisimilarity]\label{df:ep-bisimilarity}\rm
  Given an LTSS $(S,\Tr,\source,\target,\ell,\leadsto)$,
  an \emph{enabling preserving bisimulation} (\emph{ep-bisimulation}) is a relation
  $\R \subseteq S\times S\times\Pow(\Tr\times\Tr)$ satisfying
  \begin{enumerate}
    \item if $(p,q,R)\in\R$ then $R \subseteq \en(p)\times\en(q)$ such that\label{ep1}
      \begin{enumerate}
        \item $\forall t\in\en(p).\ \exists\, u\in\en(q).\ t \mathrel R u$,\label{ep1a}
        \item $\forall u\in\en(q).\ \exists\, t\in\en(p).\ t \mathrel R u$,\label{ep1b}
        \item if $t \mathrel R u$ then $\ell(t)=\ell(u)$, and\label{ep1c}
      \end{enumerate}
    \item if $(p,q,R)\in\R$ and $v \mathrel R w$, then $(\target(v),\target(w),R')\in\R$ for some $R'$ such that\label{ep2}
      \begin{enumerate}
        \item if $t \mathrel R u$ and $t \leadsto_v t'$,
          then $\exists\, u'.~ u \leadsto_w u' \land t' \mathrel{R'} u'$, and\label{ep2a}
        \item if $t \mathrel R u$ and $u \leadsto_w u'$,
          then $\exists\, t'.~ t \leadsto_v t' \land t' \mathrel{R'} u'$.\label{ep2b}
      \end{enumerate}
  \end{enumerate}
  Two states $p$ and $q$ in an LTSS are \emph{enabling preserving bisimilar} (ep-bisimilar),
  $p \bisep q$, if there exists an enabling preserving bisimulation $\R$ such that $(p,q,R)\in\R$ for some $R$.
\end{definition}
\df{ep-bisimilarity} without Items \ref{ep2a} and \ref{ep2b} is nothing
else than a reformulation of the classical definition of strong bisimilarity. 
An ep-bisimulation additionally maintains for each pair of related states $p$ and $q$ a relation $R$ between the
transitions enabled in $p$ and $q$. 
Items \ref{ep2a} and \ref{ep2b} strengthen the condition on related target states
by requiring that the successors of related transitions are again related relative to these target states.
It is this requirement (and in particular its implication stated in the following observation)
which distinguishes the transition systems for \ex{filippo}. 

\begin{observation}\label{obs:conc}\rm
  Ep-bisimilarity respects the concurrency relation.\\
For a given ep-bisimulation $\R$, if $(p,q,R)\in\R$, $t \mathrel R u$ and $v \mathrel R w$ then $t \aconc v$ iff $u \aconc w$.
\end{observation}

\begin{proposition}\label{pr:equivalence}\rm
$\bis{\it ep}$ is an equivalence relation.
\end{proposition}
\vspace{-4ex}
\begin{proof}

  \vspace{1ex}
  \noindent
  \emph{Reflexivity}:
  Let $(S,\Tr,\source,\target,\ell,\leadsto)$ be an LTSS.
  The relation
  \vspace{-1ex}
  \begin{center}
    $\R_\textit{Id} \coloneqq \{(s,s,\textit{Id}_s) \mid s\in S\}$
  \end{center}
  \vspace{-1ex}
  is an ep-bisimulation.
  Here $\textit{Id}_s \coloneqq \{(t,t) \mid t\in\en(s)\}$.

  \vspace{1ex}
  \noindent
  \emph{Symmetry}:
  For a given ep-bisimulation $\R$, the relation
  \vspace{-1ex}
  \begin{center}
  $\R^{-1} \coloneqq \{(q,p,R^{-1}) \mid (p,q,R)\in\R\}$
  \end{center}
  \vspace{-1ex}
  is also an ep-bisimulation.
  Here $R^{-1} \coloneqq \{(u,t) \mid (t,u)\in R\}$.

  \vspace{1ex}
  \noindent
  \emph{Transitivity}:
  For given ep-bisimulations $\R_1$ and $\R_2$, the relation
  \vspace{-1ex}
  \begin{center}
    $\R_1;\R_2 \coloneqq \{(p,r,R_1;R_2) \mid \exists\, q.~ (p,q,R_1)\in\R_1 \land (q,r,R_2)\in\R_2\}$
  \end{center}
  \vspace{-1ex}
  is also an ep-bisimulation.
  Here $R_1;R_2 \coloneqq \{(t,v) \mid \exists\, u.~ (t,u)\in R_1 \land (u,v)\in R_2\}$.
\end{proof}

\begin{observation}\label{obs:union}\rm
  The union of any collection of ep-bisimulations is itself an ep-bisimulation.
\end{observation}
Consequently there exists a largest ep-bisimulation.

Before proving that ep-bisimilarity preserves justness, we lift this relation to paths.
\begin{definition}[Ep-bisimilarity of Paths]\label{df:bisimilarity of paths}\rm
  Given an ep-bisimulation $\R$ and two paths $\pi = p_0\,u_0\,p_1\,u_1\,p_2\dots$ and $\pi' = p'_0\,u'_0\,p'_1\,u'_1\,p'_2\dots$,
  we write $\pi \mathrel\R \pi'$ iff $l(\pi)=l(\pi')$, and there exists $R_i \subseteq \Tr\times\Tr$ for all $i\in\IN$ with $i\leq l(\pi)$, such that
  \begin{enumerate}
    \item $(p_i,p'_i,R_i)\in\R$ for each $i\in\IN$ with $i\leq l(\pi)$,
    \item $u_i \mathrel{R_i} u'_i$ for each $i<l(\pi)$,
    \item if $t \mathrel{R_i} t'$ and $t \leadsto_{u_i} v$ with $i<l(\pi)$,
      then $\exists\, v'.~ t' \leadsto_{u'_i} v' \land v \mathrel{R_{i+1}} v'$, and
    \item if $t \mathrel{R_i} t'$ and $t' \leadsto_{u'_i} v'$ with $i < l(\pi)$,
      then $\exists\, v.~ t \leadsto_{u_i} v \land v \mathrel{R_{i+1}} v'$.
  \end{enumerate}
  Paths $\pi$ and $\pi'$ are \emph{enabling preserving bisimilar}, notation $\pi \bisep \pi'$,
  if there exists an ep-bisimulation $\R$ with $\pi \mathrel\R \pi'$.
  If $\pi \bisep \pi'$, we also write $\pi \mathrel{\vec{R}} \pi'$ if $\vec{R} \coloneqq (R_0,R_1,\dots)$ are the $R_i$ required above.
\end{definition}
Note that if $p \bisep q$ and $\pi$ is any path starting from $p$,
then, by \df{ep-bisimilarity}, there is a path $\pi'$ starting from $q$ with $\pi \bisep \pi'$.
The following lemma lifts \obs{conc}.

\begin{lemma}\label{lem:justness-preserving}\rm
If $\pi \mathrel{\vec{R}} \pip $ with $\pi$ finite and $t \mathrel{R_0} t'$ then $t \paconc \hat\pi$ iff $t' \paconc \hatpip$.
\end{lemma}
\begin{proof}
  We have to show that $\exists v.~ t \leadsto_{\hat\pi} v$ iff $\exists v'.~ t' \leadsto_{\hatpip} v'$.
  Using symmetry, we may restrict attention to the `only if' direction.
  We prove a slightly stronger statement, namely that for every transition $v$ with $t \leadsto_{\hat\pi} v$
  there exists a $v'$ such that $t' \leadsto_{\hatpip} v'$ and $v \mathbin{R_{l(\pi)}} v'$.

  We proceed by induction on the length of $\pi$.

  The base case, where $l(\pi)=0$, $\hat \pi=\varepsilon$ and thus $v=t$, holds trivially, taking $v' \coloneqq t'$.

  So assume $\pi u p \mathrel{\vec{R}} \rho u' p'$ and $t \leadsto_{\hat\pi u} w$.
  Then there is a $v$ with $t \leadsto_{\hat\pi} v$ and $v \leadsto_u w$.
  By induction there is a transition $v'$ such that $t' \leadsto_{\hatpip} v'$ and $v \mathrel{R_{l(\pi)}} v'$.
  By \df{bisimilarity of paths}(3), there is a $w'$ with $v' \leadsto_{u'} w'$ and $w \mathrel{R_{l(\pi)+1}} w'$.
  Thus $t' \leadsto_{\hatpip u'} w'$ by \df{after}.
\end{proof}

\begin{theorem}\label{thm:justness-preserving}\rm Ep-bisimilarity preserves justness:
  Given two paths $\pi$ and $\pi'$ in an LTSS with $\pi\mathbin{\bisep}\pi'$, and a set $B$ of blocking actions, then $\pi$ is $B$-just iff $\pi'$ is $B$-just.
\end{theorem}
\begin{proof}
Let $\pi = p_0\,u_0\,p_1\,u_1\,p_2\dots$ and $\pi' = p'_0\,u'_0\,p'_1\,u'_1\,p'_2\dots$\,.
Suppose $\pi$ is $B$-unjust, so there exist an $i\in \IN$ with $i\leq l(\pi)$ and a transition
$t \in \Tr^\bullet$ with $\ell(t)\notin B$ and $\source(t)=p_i$ such that $t \aconc \hat\rho$
for each finite prefix $\rho$ of the suffix $p_i\,u_i\,p_{i+1}\,u_{i+1}\,p_{i+2}\dots$ of $\pi$.
It suffices to show that also $\pi'$ is $B$-unjust.

Take an ep-bisimulation $\R$ such that $\pi \mathrel\R \pi'$. Choose $R_i\subseteq \Tr \times \Tr$ for all
$i \in \IN$ with $i\leq l(\pi)$, satisfying the four conditions of \df{bisimilarity of paths}.
Pick any $t'\in \Tr$ with $t \mathrel{R_i} t'$  --  such a $t'$ must exist by \df{ep-bisimilarity}.
Now $\source(t')=p'_i$. By \df{ep-bisimilarity}, $t' \in \Tr^\bullet$ and $\ell(t')=\ell(t)\notin B$.
It remains to show that $t' \aconc \hat\rho'$ for each finite
prefix $\rho'$ of the suffix $p_i'\,u_i'\,p_{i+1}'\,u_{i+1}'\,p_{i+2}'\dots$ of $\pi'$.
This follows by \lem{justness-preserving}.
\end{proof}

\section{Stating and Verifying Liveness Properties}

The main purpose of ep-bisimilarity is as a vehicle for proving
liveness properties. A \emph{liveness property} is any property
saying that eventually something good will happen \cite{Lam77}. Liveness properties are
\emph{linear-time properties}, in the sense that they are interpreted primarily on the (complete)
\emph{runs} of a system. When a distributed system is formalised as a state in an (extended) LTS,
a run of the distributed system is modelled as a path in the transition system, starting from that state. 
However, not every such path models a realistic system run.
A \emph{completeness criterion} \cite{synchrons} selects some of the paths of a system as
\emph{complete paths}, namely those that model runs of the represented system.

A state $s$ in an (extended) LTS is said to satisfy a linear time property $\varphi$ when employing
the completeness criterion $\it CC$, notation $s \models^{\it CC} \varphi$, if each complete run of
$s$ satisfies $\varphi$ \cite{EPTCS322.6}. Writing $\pi \models \varphi$ when property $\varphi$
holds for path $\pi$, we thus have $s \models^{\it CC} \varphi$ iff $\pi \models \varphi$
for all complete paths $\pi$ starting from $s$.
When simplifying a system $s$ into an equivalent system $s'$, so that $s \sim s'$ for
some equivalence relation $\sim$, it is important that judgements $\models^{\it CC} \varphi$ are preserved:\\[-1.5mm]
\centerline{$s \sim s' ~~\Rightarrow~~ (s \models^{\it CC} \varphi \Leftrightarrow s' \models^{\it CC} \varphi).$}\\[1.5mm]
This is guaranteed when for each path $\pi$ of $s$ there exists a path $\pi'$ of $s'$, such that
(a) $\pi \models \varphi$ iff $\pi' \models \varphi$, and (b) $\pi'$ is complete iff $\pi$ is complete.
Here (a) is already guaranteed when $\pi$ and $\pi'$ are related by
strong bisimilarity. Taking ${\it CC}$ to be  $B$-justness  for any classification of a set of actions
$B$ as blocking, and $\sim$ to be $\bis{\it ep}$, \thm{justness-preserving} now ensures (b) as well.

\section{Interpreting Justness in Process Algebras}\label{sec:ABCdE}

Rather than using LTSs directly to model distributed systems,
one usually employs other formalisms such as process algebras or Petri nets, for they are often easier to use for 
system modelling. Their formal semantics maps their syntax into states of LTSs. In this section we introduce the
 \emph{Algebra of Broadcast Communication with discards and Emissions} (\ABCdE),
an extension of Milner's \emph{Calculus of Communication Systems} (CCS)~\cite{Mi90ccs} with
broadcast communication and signal emissions. In particular, we give a structural operational
semantics~\cite{Pl04} that interprets process expressions as states in an LTS\@.
Subsequently, we define the successor relation $\leadsto$ for \ABCdE,
thereby enriching the LTS into an LTSS\@.

We use \ABCdE here, as for many realistic applications CCS is not expressive enough
\cite{vG05d,GH15b}.
The presented approach can be applied to a wide range of process algebras.
\ABCdE is largely designed to be a starting point for 
transferring the presented theory to algebras used for `real' applications.
For example, broadcast communication is needed for verifying routing protocols (e.g.\ \cite{DIST16});
signals are employed to correctly render and verify protocols for
mutual exclusion~\cite{GH15b,EPTCS255.2}.
Another reason is that broadcasts as well as signals, in different ways, give rise
to asymmetric concurrency relations, and we want to show that our approach is flexible enough to handle~this.

\subsection{Algebra of Broadcast Communication with Discards and Emissions}\label{sec:abcde}

\ABCdE is parametrised with sets
$\A$ of \emph{agent identifiers},
$\Ch$ of \emph{handshake communication names},
$\B$ of $\emph{broadcast communication names}$, and
$\Sig$ of $\emph{signals}$;
each $A\in\A$ comes with a defining equation \plat{$A \defis P$} with $P$ being a guarded \ABCdE expression as defined below.
$\bar{\Ch} \coloneqq \{\bar{c} \mid c\in\Ch\}$ is the set of \emph{handshake communication co-names},
and $\bar{\Sig} \coloneqq \{\bar{s} \mid s\in\Sig\}$ is the set of signal emissions.
The collections $\B!$, $\B?$, and $\B{:}$ of \emph{broadcast}, \emph{receive}, and \emph{discard} actions are given by
$\B\sharp \coloneqq \{b\sharp \mid b\in \B\}$ for $\sharp \in \{!,?,{:}\}$.
$\textit{Act} \coloneqq \Ch \djcup \bar{\Ch} \djcup \{\tau\} \djcup \B! \djcup \B? \djcup \Sig$
is the set of {\em actions}, where $\tau$ is a special \emph{internal action}.
$\Lab \coloneqq \textit{Act} \djcup \B{:} \djcup \bar{\Sig}$ is the set of \emph{transition labels}.
Complementation extends to $\Ch \djcup \bar{\Ch} \djcup \Sig \djcup \bar{\Sig}$ by $\bar{\bar{c}}\mathbin{\coloneqq}c$.

Below,
$c$ ranges over $\Ch \djcup \bar{\Ch} \djcup \Sig \djcup \bar{\Sig}$,
$\eta$  over $\Ch \djcup \bar{\Ch} \djcup \{\tau\} \djcup \Sig \djcup \bar{\Sig}$,
$\alpha$ over $\textit{Act}$,
$\ell$ over $\Lab$,
$b$ over $\B$,
$\sharp,\sharp_1,\sharp_2$  over $\{{!,?,:}\}$ and
$s,r$ over $\Sig$. A \emph{relabelling} is a function $f:(\Ch\to\Ch)\djcup(\B\to\B)\linebreak[2]\djcup(\Sig\to\Sig)$;
it extends to $\Lab$ by $f(\bar{c})=\overline{f(c)}$, $f(\tau)\coloneqq\tau$, and $f(b\sharp)=f(b)\sharp$.

The set $\cT$ of \ABCdE expressions or \emph{processes} is the smallest set including:
\begin{center}\small
  \vspace{-1ex}
  \begin{tabular}{@{}l@{\ }l@{\ }l@{\ }|@{\ }l@{\ }l@{\ }l@{}}
    ${\bf 0}$       &                                                            & \emph{inaction} &
    $\alpha.P$      & for $\alpha\mathop\in\textit{Act}$ and $P\mathop\in\cT$    & \emph{action prefixing} \\
    $P+Q$           & for $P,Q\mathop\in\cT$                                     & \emph{choice} &
    $P|Q$           & for $P,Q\mathop\in\cT$                                     & \emph{parallel composition} \\
    $P\backslash L$ & for $L\mathop\subseteq\Ch\djcup\Sig$, $P\mathop\in\cT$     & \emph{restriction} &
    $P[f]$          & for $f$ a relabelling, $P\mathop\in\cT$                    & \emph{relabelling} \\
    $A$             & for $A\mathop\in\A$                                        & \emph{agent identifier} &
    $P\signals s$   & for $s\mathop\in\Sig$                                      & \emph{signalling}
  \end{tabular}
\end{center}
We abbreviate $\alpha.{\bf 0}$ by $\alpha$, and $P\backslash\{c\}$ by $P\backslash c$.
An expression is guarded if each agent identifier occurs within the scope of a prefixing operator.

The semantics of \ABCdE is given by the labelled transition relation
${\rightarrow} \subseteq \cT \times \Lab \times\cT$, where transitions 
\plat{$P \goto{\ell} Q$} are derived from the rules of Tables~\ref{tab:ABCdE Basic}--\ref{tab:ABCdE Signal}.
Here $\overline{L} \coloneqq \{\bar{c} \mid c\in L\}$.

\begin{table*}[t]
  \caption{Structural operational semantics of {\ABCdE} -- Basic\label{tab:ABCdE Basic}}
  \vspace{-1ex}
  \footnotesize
  \centering
    \framebox{$\begin{array}{ccc}
      \alpha.P \goto{\alpha} P \mylabel{Act} &
      \displaystyle\frac{P \goto{\alpha} P'}{P+Q \goto{\alpha} P'} \mylabel{Sum-l} &
      \displaystyle\frac{Q \goto{\alpha} Q'}{P+Q \goto{\alpha} Q'} \mylabel{Sum-r} \\[4ex]
      \displaystyle\frac{P \goto{\eta} P'}{P|Q \goto{\eta} P'|Q} \mylabel{Par-l} &
      \displaystyle\frac{P \goto{c} P',~ Q \goto{\bar{c}} Q'}{P|Q \goto{\tau} P'|Q'} \mylabel{Comm} &
      \displaystyle\frac{Q \goto{\eta} Q'}{P|Q \goto{\eta} P|Q'} \mylabel{Par-r} \\[4ex]
      \displaystyle\frac{P \goto{\ell} P'}{P\backslash L \goto{\ell} P'\backslash L}
        ~(\ell\notin L \djcup \overline{L}) \mylabel{Res} &
      \displaystyle\frac{P \goto{\ell} P'}{P[f] \goto{f(\ell)} P'[f]} \mylabel{Rel} &
      \displaystyle\frac{P \goto{\alpha} P'}{A \goto{\alpha} P'}
        ~(A \defis P) \mylabel{Rec}
    \end{array}$}
\end{table*}

Table~\ref{tab:ABCdE Basic} shows the basic operational semantics rules, identical to the ones of CCS~\cite{Mi90ccs}.
  The process $\alpha.P$ performs the action $\alpha$ first and subsequently acts as $P$.
The choice operator $P+Q$ may act as either $P$ or $Q$, depending on which of the processes is able to act at all.
The parallel composition $P|Q$ executes an action $\eta$ from $P$, an action $\eta$ from $Q$, or in the case where
$P$ and $Q$ can perform complementary actions $c$ and $\bar{c}$, the process can perform a synchronisation, resulting in an internal action $\tau$.
The restriction operator $P \backslash L$ inhibits execution of the actions from $L$ and their complements. 
The relabelling $P[f]$ acts like process $P$ with all labels $\ell$ replaced by $f(\ell)$.
Finally, an agent $A$ can do the same actions as the body $P$ of its defining equation.
When we take $\B=\Sig \coloneqq \emptyset$, only the rules of \tab{ABCdE Basic} matter, and \ABCdE simplifies to CCS\@.

\begin{table*}[b]
  \caption{Structural operational semantics of {\ABCdE} -- Broadcast\label{tab:ABCdE Broadcast}}
  \vspace{-1ex}
  \footnotesize
  \centering
    \framebox{$\begin{array}{c@{\qquad}c@{\qquad}c}
      {\bf 0} \goto{b{:}} {\bf 0}  \mylabel{Dis-nil}&
      \alpha.P \goto{b{:}} \alpha.P
        ~(\alpha\neq b?) \mylabel{Dis-act}&
      \displaystyle\frac{P \goto{b{:}} P',~ Q \goto{b{:}} Q'}{P+Q \goto{b{:}} P'+Q'}  \mylabel{Dis-sum}\\[4ex]
      \multicolumn{2}{c@{\qquad}}{
        \displaystyle\frac{P \goto{b\sharp_1} P',~ Q \goto{b\sharp_2} Q'}{P|Q \goto{b\sharp} P'|Q'}
          \scriptstyle (\sharp_1\circ\sharp_2 = \sharp \neq \_)
          \textstyle ~\text{with}~
          \begin{array}{c@{\ }|@{\ }c@{\ \ }c@{\ \ }c}
            \scriptstyle \circ & \scriptstyle ! & \scriptstyle ? & \scriptstyle : \\
            \hline
            \scriptstyle ! & \scriptstyle \_ & \scriptstyle ! & \scriptstyle ! \\
            \scriptstyle ? & \scriptstyle !  & \scriptstyle ? & \scriptstyle ? \\
            \scriptstyle : & \scriptstyle !  & \scriptstyle ? & \scriptstyle : \\
          \end{array} \mylabel{Bro}} &
      \displaystyle\frac{P \goto{b{:}} P'}{A \goto{b{:}} A}
       ~(A \defis P)\mylabel{Dis-rec}
    \end{array}$}
    \vspace{-4mm}
\end{table*}

Table~\ref{tab:ABCdE Broadcast} augments CCS with a mechanism for broadcast communication.
The rules are similar to the ones for the Calculus of Broadcasting Systems (CBS)~\cite{CBS91}; 
they also appear in the process algebra ABC~\cite{GH15a}, a strict subalgebra of \ABCdE.
The Rule \myref{Bro} presents the core of broadcast communication, where any broadcast-action $b!$
performed by a component in a parallel composition is guaranteed to be received by any other
component that is ready to do so, \ie in a state that admits a $b?$-transition.
Since it is vital that the sender of a broadcast can always proceed with it, regardless of the state
of other processes running in parallel, the process algebra features discard actions $b{:}$, in such a way
that each process in any state can either receive a particular broadcast $b$, by performing the
action $b?$, or discard it, by means of a $b{:}$, but not both. A broadcast transmission $b!$ can
synchronise with either $b?$ or $b{:}$, and thus is never blocked by lack of a listening party. 
In order to ensure associativity of the parallel composition, one requires rule \myref{Bro}
to consider receipt at the same time ($\sharp_1 = \sharp_2 = \mathord{?}$).
The remaining four rules of Table~\ref{tab:ABCdE Broadcast} generate the discard-actions.
The Rule \myref{Dis-nil} allows the nil process (inaction) to discard any incoming message; 
in the same spirit \myref{Dis-act} allows a message to be discarded by a process 
that cannot receive it.  A process offering a choice can only perform a discard-action 
if neither choice-option can handle it (Rule \myref{Dis-sum}).
Finally, an agent A can discard a broadcast iff the body $P$ of its defining equation can discard it.
Note that in all these cases a process does not change state by discarding a broadcast.

There exists a variant of CBS, ABC and \ABCdE without discard actions, see~\cite{GH15a,synchrons}. 
This approach, however, features negative premises in the operational rules. 
As a consequence, the semantics are not in De Simone format~\cite{dS85}.
Making use of discard actions and staying within the De Simone format allows us 
to use meta-theory about this particular format. For example 
we know, without producing our own proof, that the operators $+$ and $|$ of ABC and \ABCdE are associative and commutative, up to strong bisimilarity~\cite{CMR08}. Moreover, strong bisimilarity \cite{Mi90ccs}
is a congruence for all operators of \ABCdE.

Next to the standard operators of CCS and a broadcast mechanism, \ABCdE features also signal emission.
Informally, the signalling operator $P\signals s$ emits the signal $s$ to be read by another process. 
Signal emissions cannot block other actions of $P$. Classical examples are the modelling of read-write processes 
or traffic lights (see \Sec{ltss}).

Formally, our process algebra features a set $\Sig$ of signals.
The semantics of signals is presented in Table~\ref{tab:ABCdE Signal}.
\begin{table*}[t]
  \caption{Structural operational semantics of {\ABCdE} -- Signals\label{tab:ABCdE Signal}}
  \vspace{-1ex}
  \footnotesize
  \centering
    \framebox{$\begin{array}{@{}c@{\quad}c@{\quad}c@{\quad}c@{}}
      P\signals s \goto{\bar{s}} P\signals s \mylabel{Sig}&
      \displaystyle\frac{P \goto{\bar{s}} P'}{P+Q \goto{\bar{s}} P'+Q} \mylabel{Sig-sum-l}&
      \multicolumn{2}{l}
      {\displaystyle\frac{Q \goto{\bar{s}} Q'}{P+Q \goto{\bar{s}} P+Q'} \mylabel{Sig-sum-r}}\\[4ex]
      \displaystyle\frac{P \goto{\bar{s}} P'}{P\signals r \goto{\bar{s}} P'\signals r} \mylabel{Sig-sig}&
      \displaystyle\frac{P \goto{\bar{s}} P'}{A \goto{\bar{s}} A}
        ~(A \defis P) \mylabel{Sig-rec} &
        \displaystyle\frac{P \goto{\alpha} P'}{P\signals r \goto{\alpha} P'} \mylabel{Act-sig}&
        \displaystyle\frac{P \goto{b{:}} P'}{P\signals r \goto{b{:}} P'\signals r} \mylabel{Dis-sig}
    \end{array}$}
    \vspace{-3mm}
\end{table*}
The first rule \myref{Sig} models the emission $\bar s$ of signal $s$ to the environment. 
The environment (processes running in parallel) can read the signal by performing a read action $s$.
This action synchronises with the emission $\bar s$, via the rules of Table~\ref{tab:ABCdE Basic}.
Reading does not change the state of the emitter. The next four rules describe the interaction between signal 
emission and other operators,
namely choice, signal emission and recursion. In short, these operators 
do not prevent the emission of a signal, and emitting signals never changes the state of the
  emitting process. Other operators, such as relabelling and restriction do not 
need special attention as they are already handled by the corresponding rules in Table~\ref{tab:ABCdE Basic}.
This is achieved by carefully selecting the types of the labels: while \myref{Sum-l} features a label $\alpha$ 
of type $\textit{Act}$, the rules for restriction and relabelling use a label $\ell\in\Lab$.
In case a process performs a `proper' action, the signal emission ceases (Rule \myref{Act-sig}),
but if the process performs a broadcast discard transition, it does not (Rule \myref{Dis-sig}).

The presented semantics stays within the realm of the De Simone format \cite{dS85}, which brings many advantages. 
However, there exists an alternative, equivalent semantics, which is based on predicates. Rather than explicitly modelling
$P$ emitting $s$ by the transition $P \goto{\bar s} P$, one can introduce the predicate
$P^{\curvearrowright{s}}$. The full semantics can be found in \cite{EPTCS255.2}. Some readers might
find this notation more intuitive as signal emitting processes do not perform an actual action when a component reads the emitted signal.

\subsection{Naming Transitions}

The operational semantics of \ABCdE presented in \Sec{abcde} interprets the language as an LTS.
In \Sec{successors}, we aim to extend this LTS into an LTSS by defining a
successor relation $\leadsto$ on the transitions, and thereby also a concurrency relation
$\aconc$. However, the standard interpretation of CCS-like languages, which takes as transitions the
triples $P \goto\alpha Q$ that are derivable from the operational rules, does not work for our purpose.
\begin{example}
Let $P=A|B$ with \plat{$A \defis \tau.A + a.A$} and \plat{$B \defis \bar a.B$}. Now the transition $P \goto\tau P$
arises in two ways; either as a transition solely of the left component, or as a synchronisation
between both components. The first form of that transition is concurrent with the transition
$P \goto{\bar{a}} P$, but the second is not. In fact, an infinite path that would only perform the
$\tau$-transition stemming from the left component would not be just, whereas a path that schedules
both $\tau$-transitions infinitely often is. This shows that we want to distinguish these
two $\tau$-transitions, and hence not see a transition as a triple $P \goto\alpha Q$.\hfill$\lrcorner$
\end{example}
Instead, we take as the set $\Tr$ of transitions in our LTSS the \emph{derivations} of the
\emph{transition triples} $P \goto\alpha Q$ from our operational rules.
This is our reason to start with a definition of an LTS that features transitions as a primitive rather
than a derived concept.

\begin{definition}[Derivation]\label{df:derivation}\rm
  A \emph{derivation} of a transition triple $\varphi$ is
  a \emph{well-founded} (without infinite paths that keep going up), ordered, upwardly branching tree
  with the nodes labelled by transition triples,
  such that
  \begin{enumerate}
    \item the root is labelled by $\varphi$, and
    \item if $\mu$ is the label of a node and $K$ is the sequence of labels of this node's children
      then \plat{$\frac{K}{\mu}$} is a substitution instance of a rule from Tables~\ref{tab:ABCdE Basic}--\ref{tab:ABCdE Signal}.
  \end{enumerate}
  Given a derivation, we refer to the subtrees obtained by deleting the root node as its \emph{direct subderivations}. 
  Furthermore, by definition, \plat{$\frac{K_\varphi}{\varphi}$} is a substitution instance of a rule,
    where $\varphi$ is the label of the derivation's root and $K_\varphi$ is the sequence of labels of the root's children;
    the derivation is said to be \emph{obtained} by this rule.
\end{definition}

We interpret \ABCdE as an LTS $(S,\Tr,\source,\target,\ell)$ by taking as states $S$ the \ABCdE expressions $\cT$
and as transitions $\Tr$ the derivations of transition triples $P \goto{\alpha} Q$.
  Given a derivation $t$ of a triple $P \goto{\alpha} Q$, we define its label $\ell(t) \coloneqq \alpha$,
  its source $\source(t) \coloneqq P$,~and its target $\target(t) \coloneqq Q$.

\begin{definition}[Name of Derivation]\label{df:name of derivation}\rm
  The derivation obtained by application of \myref{Act} is called $\actsyn{\alpha}P$.
  The derivation obtained by application of \myref{Comm} or \myref{Bro} is called $t|u$, where $t,u$
  are the names of its direct subderivations.%
  \footnote{The order of a rule's premises should be maintained in the names of derivations obtained by it.
    Here $t$ should be the derivation corresponding to the first premise and $u$ to the second.
    As a result, $t \neq u \implies t|u \neq u|t$.}
  The derivation obtained by application of \myref{Par-l} is called $t|Q$
  where $t$ is the \premise's name and $Q$ is the process on the right hand side of $|$ in the derivation's source.
      In the same way, the derivation obtained by application of \myref{Par-r} is called $P|t$,
  while \myref{Sum-l}, \myref{Sum-r}, \myref{Res}, \myref{Rel}, and \myref{Rec} yield $t{+}Q$, $P{+}t$, $t\backslash L$, $t[f]$ and $A{:}t$,
  where $t$ is the \premise's name.
  The remaining four rules of Table~\ref{tab:ABCdE Broadcast} yield $b{:}{\bf 0}$, $b{:}\alpha.P$, $t{+}u$ and $A{:}t$,
  where $t,u$ are {\premise}s' names.
  The derivation of $P\signals s \goto{\bar{s}} P\signals s$ obtained by \myref{Sig} is called $P\sigsyn{s}$.
  Rules \myref{Act-sig}, \myref{Dis-sig} and~\myref{Sig-sig} yield $t\signals r$, and rules \myref{Sig-sum-l},
  \myref{Sig-sum-r} and  \myref{Sig-rec} yield $t{+}Q$, $P{+}t$ and $A{:}t$,  where $t$ is the \premise's name.
\end{definition}

\noindent
A derivation's name reflects the syntactic structure of that derivation.
The derivations' names not only provide a convenient way to identify derivations but also highlight the compositionality of derivations.
For example, given a derivation $t$ of $P \goto{c} P'$ and a derivation $u$ of $Q \goto{\bar{c}} Q'$
with $c \in \Ch \djcup \bar{\Ch} \djcup \Sig \djcup \bar{\Sig}$, $t|u$ will be a derivation of $P|Q \goto{\tau} P'|Q'$.

Hereafter, we refer to a derivation of a transition triple as a `transition'.

\subsection{Successors}\label{sec:successors}

In this section we extend the LTS of \ABCdE into an LTSS, by defining the successor relation $\leadsto$.
For didactic reasons, we do so first for CCS, and then extend our work to \ABCdE.

Note that $\chi \leadsto_{\zeta} \chi'$ can hold only when $\source(\chi)=\source(\zeta)$, \ie
transitions $\chi$ and $\zeta$ are both enabled in the state $\sR \coloneqq \source(\chi)=\source(\zeta)$.
It can thus be defined by structural induction on $\sR$.
The meaning of $\chi \leadsto_{\zeta} \chi'$ is (a) that $\chi$ is unaffected by $\zeta$ -- denoted $\chi \aconc \zeta$ --
and (b) that when doing $\zeta$ instead of $\chi$, afterwards a variant $\chi'$ of $\chi$ is still enabled.
Restricted to CCS, the relation $\aconc$ is moreover symmetric, and we can write $\chi \conc \zeta$.

In the special case that $\sR={\bf 0}$ or $\sR=\alpha.Q$, there are no two concurrent transitions enabled in
$\sR$, so this yields no triples $\chi \leadsto_{\zeta} \chi'$.
When $\sR = P+Q$, any two concurrent transitions $\chi \conc \zeta$  enabled in $\sR$ must either stem
both from $P$ or both from $Q$. In the former case, these transitions have the form $\chi=t+Q$ and
$\zeta=v+Q$, and we must have $t \conc v$, in the sense that $t$ and $v$ stem from different
parallel components within $P$. So $t \leadsto_v t'$ for some transition $t'$.
As the execution of $\zeta$ discards the summand $Q$, we also obtain $\chi \leadsto_{\zeta} t'$.
This motivates Item \ref{sum1CCS} in \df{leadstoCCS} below.
Item  \ref{sum2CCS} follows by symmetry.

Let $\sR=P|Q$. One possibility for $\chi \leadsto_{\zeta} \chi'$
is that $\chi$ comes from the left component and $\zeta$ from the right. So $\chi$ has the form $t|Q$ and $\zeta=P|w$.
In that case $\chi$ and $\zeta$ must be concurrent: we always have $\chi\conc\zeta$.
When doing $w$ on the right, the left component does not change, and afterwards $t$ is still possible.
Hence $\chi \leadsto_{\zeta} t|\target(w)$. This explains Item~\ref{par1CCS} in \df{leadstoCCS}.

Another possibility is that $\chi$ and $\zeta$ both stem from the left component.
In that case $\chi=t|Q$ and $\zeta=v|Q$, and it must be that $t \conc u$ within the left component.
Thus $t \leadsto_v t'$ for some transition $t'$, and we obtain $\chi \leadsto_{\zeta} t'|Q$. 
This motivates the first part of Item~\ref{par2CCS}.

It can also happen that $\chi$ stems form the left component, whereas $\zeta$ is a synchronisation,
involving both components. Thus $\chi=t|Q$ and $\zeta=v|w$. For $\chi\conc\zeta$ to hold, it must be
that $t\conc v$, whereas the $w$-part of $\zeta$ cannot interfere with $t$. This yields the second
part of Item~\ref{par2CCS}.

The last part of Item~\ref{par2CCS} is explained in a similar vain from the possibility that $\zeta$
stems from the left while $\chi$ is a synchronisation of both components.
Item~\ref{par3CCS} follows by symmetry.

In case both $\chi$ and $\zeta$ are synchronisations involving both components, \ie $\chi=t|u$ and
$\zeta=v|w$, it must be that $t \conc v$ and $u \conc w$. Now the resulting variant $\chi'$ of
$\chi$ after $\zeta$ is simply $t'|v'$, where $t \leadsto_v t'$ and $u \leadsto_v u'$.
This underpins Item~\ref{par4CCS}.

If $\sR$ has the form $P[f]$, $\chi$ and $\zeta$ must have the form $t[f]$ and $v[f]$, respectively.
Whether $t$ and $v$ are concurrent is not influenced by the renaming operator. So $t\conc v$.
The variant of $t$ that remains after doing $v$ is also not affected by the renaming,
so if $t \leadsto_v t'$ then $\chi \leadsto_{\zeta} t'[f]$. The case that $\sR$ has the form
$P{\setminus}L$ is equally trivial. This yields the first two parts of Item~\ref{othersCCS}.

In case $\sR=A$ with \plat{$A \defis P$}, then $\chi$ and $\zeta$ must have the forms $A{:}t$ and  $A{:}v$,
respectively, where $t$ and $v$ are enabled in $P$. Now $\chi\conc \zeta$ only if $t \conc v$,
so $t \leadsto_v t'$ for some transition $t'$. As the recursion around $P$ disappears upon executing
$\zeta$, we obtain $\chi \leadsto_{\zeta} t'$. This yields the last part of Item~\ref{othersCCS}.
Together, this motivates the following definition.

\begin{definition}[Successor Relation for CCS]\label{df:leadstoCCS}\rm
  The relation ${\leadsto} \subseteq \nabla\times\nabla\times\nabla$ is the smallest relation satisfying
  \begin{enumerate}
    \item $t \leadsto_v t'$ implies $t+Q \leadsto_{v+Q} t'$,\label{sum1CCS}
    \item $u \leadsto_w u'$ implies $P+u \leadsto_{P+w} u'$,\label{sum2CCS}\vspace{1ex}

    \item $t|Q \leadsto_{P|w} (t|\target(w))$ and $P|u \leadsto_{v|Q} (\target(v)|u)$,\label{par1CCS}
    \item $t \leadsto_v t'$ implies $t|Q \leadsto_{v|Q} t'|Q$, $t|Q \leadsto_{v|w} (t'|\target(w))$, and $t|u \leadsto_{v|Q} t'|u$,\label{par2CCS}
    \item $u \leadsto_w u'$ implies $P|u \leadsto_{P|w} P|u'$, $P|u \leadsto_{v|w} (\target(v)|u')$, and $t|u \leadsto_{P|w} t|u'$,\label{par3CCS}
    \item $t \leadsto_v t' \land u \leadsto_w u'$ implies $t|u \leadsto_{v|w} t'|u'$,\label{par4CCS}\vspace{1ex}

    \item $t \leadsto_v t'$ implies $t\backslash L \leadsto_{v\backslash L} t'\backslash L$, $t[f] \leadsto_{v[f]} t'[f]$
       and $A{:}t \leadsto_{A{:}v} t'$.\label{othersCCS}\vspace{1ex}
  \end{enumerate}
  for all $t,t',u,u',v,w\in\nabla$, $P,Q\mathbin\in\cT$ and $L,f,A$ with $\source(t)\mathbin=\source(v)\mathbin=P$,
  $\source(u)\mathbin=\source(w)\mathbin=Q$, $\source(t')\mathbin=\target(v)$, $\source(u')\mathbin=\target(w)$,
  $L\subseteq\Ch$, $f$ a relabelling and $A\in\A$ -- provided that the composed transitions exist.
\end{definition}
By projecting the ternary relation $\leadsto$ on its first two components, we obtain a
characterisation of the concurrency relation $\conc$ between CCS transitions:

\begin{observation}[Concurrency Relation for CCS]\label{obs:aconcCCS}\rm
  The relation ${\conc} \subseteq \nabla\times\nabla$ is the smallest relation satisfying
  \begin{enumerate}
    \item $t \conc v$ implies $t+Q \conc v+Q$,
    \item $u \conc w$ implies $P+u \conc P+w$,\vspace{1ex}

    \item $t|Q \conc P|w$ and $P|u \conc v|Q$,
    \item $t \conc v$ implies $t|Q \conc v|Q$, $t|Q \conc v|w$, and $t|u \conc v|Q$,
    \item $u \conc w$ implies $P|u \conc P|w$, $P|u \conc v|w$, and $t|u \conc P|w$,
    \item $t \conc v \land u \conc w$ implies $t|u \conc v|w$,\vspace{1ex}

    \item $t \conc v$ implies $t\backslash L \conc v\backslash L$, $t[f] \conc v[f]$ and $A{:}t \conc A{:}v$,\vspace{1ex}
  \end{enumerate}
  for all $t,u,v,w\in\nabla$, $P,Q\mathbin\in\cT$ and $L,f,A$ with $\source(t)\mathbin=\source(v)\mathbin=P$, $\source(u)\mathbin=\source(w)\mathbin=Q$,
  $L\subseteq\Ch$, $f$ a relabelling and $A\in\A$ -- provided that the composed transitions exist.
\end{observation}

The same concurrency relation appeared earlier in \cite{GH15a}.
\df{leadstoCCS} and \obs{aconcCCS} implicitly provide SOS rules for $\leadsto$ and $\conc$, such as 
\plat{${\frac{t \conc v'}{t+Q \conc v+Q}}$}. It is part of future work to investigate a rule format for ep-bisimilarity.

\df{leadsto} below generalises \df{leadstoCCS} to all of \ABCdE.
In the special case that $\zeta$ is a broadcast discard or signal emission, \ie $\ell(\zeta) \in \B{:} \djcup \bar\Sig$,
the transition $\zeta$ does not change state -- we have $\source(\zeta)=\target(\zeta)=\sR$ -- and is
supposed not to interfere with any other transition $\chi$ enabled in $\sR$. Hence $\chi \aconc \zeta$ and
$\chi \leadsto_{\zeta} \chi$. This is Item~\ref{signals} from \df{leadsto}.

Consequently, in Item~\ref{others}, which corresponds to Item~\ref{othersCCS} from \df{leadstoCCS},
we can now safely restrict attention to the case $\ell(\zeta) \in\textit{Act}$.
The last part of that item says that if within the scope of a signalling operator an action $v$ occurs,
one escapes from this signalling context, similarly to the cases of choice and recursion.
That would not apply if $v$ is a broadcast discard or signal emission, however.

An interesting case is when $\chi$ is a broadcast receive or discard transition, \ie $\ell(\chi)=b?$ or $b{:}$\,.
We postulate that one can never interfere with such an activity, as each process is always able
to synchronise with a broadcast action, either by receiving or by discarding it.
So we have $\chi\aconc \zeta$ for all $\zeta$ with $\sR=\source(\zeta)=\source(\chi)$.
It could be, however, that in $\chi \leadsto_{\zeta} \chi'$, one has $\ell(\chi)=b?$ and $\ell(\chi')= b{:}$, or vice versa.
Item~\ref{discard base} says that if $\sR=\alpha.P$, with $\zeta$ the $\alpha$-transition to $P$, then $\chi'$
can be any transition labelled $b?$ or $b{:}$ that is enabled in $P$.
The second parts of Items~\ref{dis-sum1} and~\ref{dis-sum2} generalise this idea to discard actions
enabled in a state of the form $P+Q$. Finally, Items~\ref{sum-recv1} and~\ref{sum-recv2} state that
when $\chi$ is a broadcast receive stemming from the left side of $\sR=P+Q$ and $\zeta$ an
action from the right, or vice versa, then $\chi'$ may be any transition labelled $b?$ or $b{:}$
that is enabled in $\target(\zeta)$. In all other cases, successors of $\chi$ are inherited from
successors of their building block, similar to the cases of other transitions.%

\begin{definition}[Successor Relation for \ABCdE]\label{df:leadsto}\rm
  The relation ${\leadsto} \subseteq \nabla\times\nabla\times\nabla$ is the smallest relation satisfying
  \setlength\leftmargini{1.8em}
  \begin{enumerate}
    \item $\ell(\zeta)\in\B{:}\djcup\bar{\Sig}$ and $\source(\zeta)=\source(\chi)$ implies $\chi \leadsto_\zeta \chi$,\label{signals}
    \item $\ell(t)\in\{b?,b{:}\}$ implies $\actsyn{b?}P \leadsto_{\actsyn{b?}P} t$ and $b{:}\alpha.P \leadsto_{\actsyn{\alpha}P} t$,\vspace{1ex}\label{discard base}

    \item $\ell(v)\notin\bar{\Sig} \land t \leadsto_v t'$ implies $t+Q \leadsto_{v+Q} t'$ and $t+u \leadsto_{v+Q} t'$,\label{dis-sum1}
    \item $\ell(w)\notin\bar{\Sig} \land u \leadsto_w u'$ implies $P+u \leadsto_{P+w} u'$ and $t+u \leadsto_{P+w} u'$,\label{dis-sum2}
    \item $\ell(w)\notin\bar{\Sig} \land \ell(t)=b? \land \ell(u')\in\{b?,b{:}\}$ implies $t+Q \leadsto_{P+w} u'$,\label{sum-recv1}
    \item $\ell(v)\notin\bar{\Sig} \land \ell(u)=b? \land \ell(t')\in\{b?,b{:}\}$ implies $P+u \leadsto_{v+Q} t'$,\label{sum-recv2}\vspace{1ex}

    \item $t|Q \leadsto_{P|w} (t|\target(w))$ and $P|u \leadsto_{v|Q} (\target(v)|u)$,
    \item $t \leadsto_v t'$ implies $t|Q \leadsto_{v|Q} t'|Q$, $t|Q \leadsto_{v|w} (t'|\target(w))$, and $t|u \leadsto_{v|Q} t'|u$,
    \item $u \leadsto_w u'$ implies $P|u \leadsto_{P|w} P|u'$, $P|u \leadsto_{v|w} (\target(v)|u')$, and $t|u \leadsto_{P|w} t|u'$,
    \item $t \leadsto_v t' \land u \leadsto_w u'$ implies $t|u \leadsto_{v|w} t'|u'$,\vspace{1ex}

    \item $\ell(v)\in\textit{Act} \land t \leadsto_v t'$ implies
      $t\backslash L \leadsto_{v\hspace{-1pt}\backslash\hspace{-1pt} L} t'\backslash L$, $t[f] \leadsto_{v[f]} t'[f]$,
      $A{:}t \leadsto_{A{:}v} t'$ and $t\signals r \leadsto_{v\hspace{1pt}\signals\hspace{2pt} r} t'$,\label{others}\vspace{1ex}
  \end{enumerate}
  for all $t,t',u,u',v,w\in\nabla$, $P,Q\mathbin\in\cT$ and $\alpha,L,f,A,b,r$  with $\source(t)\mathbin=\source(v)\mathbin=P$,
  $\source(u)\mathbin=\source(w)\mathbin=Q$, $\source(t')\mathbin=\target(v)$ and $\source(u')\mathbin=\target(w)$,
  $\alpha\in\textit{Act}$, $L\subseteq\Ch\djcup\Sig$, $f$ a relabelling, $A\in\A$, $b\in\B$ and
  $r\in\Sig$ -- provided that the composed transitions exist.
\end{definition}

\noindent
Although we have chosen to inductively define the $\leadsto$ relations, in
\arxiv{Appendix~\ref{Synchrons}}{\cite[Appendix~B]{GHW21}}
we follow a different approach in which \df{leadsto} appears as a theorem rather than a definition.
Following \cite{synchrons}, we understand each transition as the synchronisation of a number of
elementary particles called \emph{synchrons}.
Then relations on synchrons are proposed in terms of which
the $\sleadsto$ relation is defined. That this leads to the same result indicates that the above
definition is more principled than arbitrary.

\subsection{Congruence and Other Basic Properties of Ep-bisimilarity}

As mentioned before, the operators $+$ and $|$ are associative and commutative up to strong bisimilarity. We can strengthen this result.
\begin{theorem}
  The operators $+$ and $|$ are associative and commutative up to $\bisep$.
\end{theorem}

\begin{proof} Remember that $\cT$ denotes the set of \ABCdE expressions or processes.

  \vspace{1ex}
  \noindent
  \emph{Commutativity of $+$, \ie $P+Q \bisep Q+P$}:
  The relation
  \[
    \{(I, I, \textit{Id}_I) \mid I\in\cT\} \djcup \{(P+Q, Q+P, R_{\scriptscriptstyle P,Q}) \mid P,Q\in\cT\}
  \]
  is an ep-bisimulation. Here $\textit{Id}_I \coloneqq \{(t,t) \mid t\mathbin\in\nablasub{I}\}$ and
  \[
    \begin{array}{@{}l@{}l@{}}
      R_{\scriptscriptstyle P,Q} \coloneqq & \{(t+Q,Q+t) \mid t\in\nablasub{P} \land \ell(t) \notin \B{:}\} \djcup {} \\
                                           & \{(P+u,u+P) \mid u\in\nablasub{Q} \land \ell(u) \notin \B{:}\} \djcup {} \\
                                           & \{(t+u,u+t) \mid t\in\nablasub{P} \land u\in\nablasub{Q} \land \ell(t)=\ell(u)\in\B{:}\}\,.
    \end{array}
  \]
  $R_{\scriptscriptstyle P,Q}$ relates transitions, \ie derivations of transition triples,
  that are composed of the same sets of {\premise}s, even though their order is reversed.

  \vspace{1ex}
  \noindent
  \emph{Associativity of $+$, \ie $(O+P)+Q \bisep O+(P+Q)$}:
  The relation
  \[
    \{(I, I, \textit{Id}_I) \mid I\in\cT\} \djcup \{((O+P)+Q, O+(P+Q), R_{\scriptscriptstyle O,P,Q}) \mid O,P,Q\in\cT\}
  \]
  is an ep-bisimulation. Here $\textit{Id}_I$ and $R_{\scriptscriptstyle O,P,Q}$ are defined similarly to the previous case.

  \vspace{1ex}
  \noindent
  \emph{Commutativity of $|$, \ie $P|Q \bisep Q|P$}:
  The relation $\{(P|Q, Q|P, R_{\scriptscriptstyle P,Q}) \mid P,Q\in\cT\}$ is an ep-bisimulation. Here
  \[
    \begin{array}{@{}l@{}l@{}}
      R_{\scriptscriptstyle P,Q} = & \{(t|Q,Q|t) \mid t\in\nablasub{P} \land \ell(t) \notin \B! \djcup \B? \djcup \B{:}\} \djcup {} \\
                                   & \{(P|u,u|P) \mid u\in\nablasub{Q} \land \ell(u) \notin \B! \djcup \B? \djcup \B{:}\} \djcup {} \\
                                   & \{(t|u,u|t) \mid t\in\nablasub{P} \land u\in\nablasub{Q} \land \ell(t)=\overline{\ell(u)} \in \Ch \djcup \bar{\Ch} \djcup \Sig \djcup \bar{\Sig}\} \djcup {} \\
                                   & \begin{array}{@{}l@{}l@{}}
                                       \{(t|u,u|t) \mid {} & t\in\nablasub{P} \land u\in\nablasub{Q} \land {} \\
                                                           & \exists\, b\in\B.~ \{\ell(t),\ell(u)\}\in\{\{b!,b?\},\{b!,b{:}\},\{b?\},\{b?,b{:}\},\{b{:}\}\}\}\,.
                                     \end{array}
    \end{array}
  \]

  \vspace{1ex}
  \noindent
  \emph{Associativity of $|$, \ie $(O|P)|Q \bisep O|(P|Q)$}: The relation 
  \[
    \{(O|P)|Q, O|(P|Q), R_{\scriptscriptstyle O,P,Q}) \mid O,P,Q\in\cT\}
  \]
  is an ep-bisimulation, where $R_{\scriptscriptstyle O,P,Q}$ is defined similarly to the previous case.
\end{proof}

Additionally, not only strong bisimilarity should be a congruence for all operators of {\ABCdE}
 -- which follows immediately from the De Simone format -- but also our new ep-bisimilarity.
This means that if two process terms are ep-bisimilar, then they are also ep-bisimilar
in any context. 

\begin{theorem}\label{thm:conc} Ep-bisimilarity is a congruence for all operators of \ABCdE. 
\end{theorem}
We cannot get it directly from the existing meta-theory on structural operational semantics, 
as nobody has studied ep-bisimilarity before.
As is standard, the proof is a case distinction on the type of the operator. 
For example, the case for action prefixing requires\vspace{-1pt}
\[
P \bisep Q\Rightarrow \alpha.P \bisep \alpha.Q  \text{ for } \alpha\in\text{\textit{Act}}\,.
\]
Such properties can be checked by inspecting the syntactic form of the transition rules, 
using structural induction.
While proofs for some statements, such as the one for action prefixing, are 
merely a simple exercise, others require more care, including long and tedious case distinctions.
A detailed proof of \thm{conc} can be found in Appendix~\ref{app:congruence proofs}.

\section{Failed Alternatives for Ep-Bisimulation}\label{sec:discussion}
On an LTSS $(S,\Tr,\source,\target,\ell,\leadsto)$
an ep-bisimulation has the type $S\times S\times\Pow(\Tr\times\Tr)$.
This is different from that of other classical bisimulations, which have the type $S\times S$.
While developing ep-bisimulation we have also explored dozens of other candidates, many of them being of type $(S\times S) \djcup (\Tr\times \Tr)$.
The inclusion of a relation between transitions is necessary to reflect the concept of components
or concurrency in one way or the other.
One such candidate definition declares a relation $\R \subseteq (S\times S)\djcup (\Tr\times \Tr)$
a valid bisimulation iff the set of triples\vspace{-2pt}
      \[\{(p,q,R) \mid (p,q)\in \R \cap (S\times S) \land R = \R \cap (\en(p)\times\en(q))\}\]
is an ep-bisimulation.
However, neither this candidate nor any of the others leads to a transitive notion of bisimilarity.
The problem stems from systems, not hard to model in \ABCdE, with multiple paths $\pi_i$ from states
$p$ to $p'$, such that a triple $(p,q,R)$ in an ep-bisimulation $\R$ forces triples
$(p',q',R'_i)$ to be in $\R$ for multiple relations $R_i \subseteq \en(p') \times \en(q')$,
depending on the chosen path $\pi_i$.

\section{Related Work}

Our LTSSs generalise the concurrent transition systems of \cite{Sta89}.
There $t \leadsto_v u$ is written as $t{\uparrow} v = u$, and $\uparrow$
is a partial function rather than a relation, in the sense that for given $t$ and $v$ there can be at
most one $u$ with $t{\uparrow} v = u$. This condition is not satisfied by broadcast communication, which
is one of the reasons we switched to the notation $t \leadsto_v u$. 
As an example, consider
  $b!|a.(b?+b?)$. The $b!$-transition after the $a$-transition has two variants, namely \plat{$\actsyn{b!}\mathbf{0}|(\actsyn{b?}\mathbf{0}{+}b?)$} and \plat{$\actsyn{b!}\mathbf{0}|(b?{+}\actsyn{b?}\mathbf{0})$}.
Another property of  concurrent transition systems that is not maintained in our framework is the
symmetry of the induced concurrency relation. Finally, \cite{Sta89} requires that
$(v{\uparrow}t){\uparrow}(u{\uparrow}t) = (v{\uparrow}u){\uparrow}(t{\uparrow}u)$,
the \emph{cube axiom}, whereas we have so far not found reasons to restrict attention to processes satisfying this axiom.
We are, however, open to the possibility that for future applications of LTSSs, some closure
properties may be imposed on them.

In \cite{DBLP:journals/fac/BoudolCHK94} a location-based bisimulation is proposed. It also keeps
track of the components participating in transitions. The underlying model is quite different from
ours, which makes it harder to formally argue that this notion of bisimilarity is incomparable to ours.
We do not know yet whether it could be used to reason about justness.\pagebreak[4]

\section{Conclusion and Future Work}\label{sec:conclusion}
In related work, it has been argued that 
fairness assumptions used for verifying liveness properties of distributed systems are too strong or unrealistic~\cite{GH19,TR13,GH15a}. 
As a consequence, justness, a minimal fairness assumption required for the verification 
of liveness properties, has been proposed. 
Unfortunately, all classical semantic equivalences, such as strong bisimilarity, fail to preserve 
justness.

In this paper, we have introduced labelled transition systems augmented by a successor relation, 
and, based on that, the concept of enabling preserving bisimilarity, a finer variant of 
strong bisimilarity. We have proven that this semantic equivalence is a congruence for all 
classical operators.
As it also preserves justness, it is our belief that enabling preserving bisimilarity in combination with justness
can and should be used for verifying liveness properties of large-scale distributed systems.

  Casually speaking, ep-bisimilarity is strong bisimilarity  augmented with 
  the requirement that the relation between enabled transitions is inherited by successor transitions.
  A straightforward question is whether this feature can be combined with other semantic equivalences, 
such as weak bisimilarity or trace equivalence.

We have further shown how process algebras can be mapped into LTSSs. Of course, 
process algebra is only one of many formal frameworks for modelling concurrent systems.
For accurately capturing causalities between event occurrences, models like Petri nets~\cite{Rei13}, event
structures~\cite{Wi87a} or higher dimensional automata~\cite{Pr91a,vG06}
are frequently preferable. Part of future work is therefore to develop a 
formal semantics with respect to LTSSs for these frameworks.

In order to understand the scope of justness in real-world applications, we plan to study systems that depend heavily on liveness. 
As a starting point we plan to verify locks, such as ticket lock.

\newcommand{\noopsort}[1]{}%
\bibliography{refs}

\newpage
\appendix
\section{Congruence Proofs}\label{app:congruence proofs}

Ep-bisimilarity is a congruence for all operators of \ABCdE iff
Propositions~\ref{pr:congruence for action prefixing}--\ref{pr:congruence for signalling} below hold.
We prove them one by one.

\begin{proposition}\label{pr:congruence for action prefixing}\rm
  If $P \bisep Q$ and $\alpha\in\textit{Act}$ then $\alpha.P \bisep \alpha.Q$.
\end{proposition}
\begin{proof}
  A {\transition} enabled in $\alpha.P$ is either $\actsyn{\alpha}P$ or $b{:}\alpha.P$ for some $b\in\B$ with $\alpha\neq b?$\,.
  \\[1ex]
  Let $\R \subseteq \cT\times\cT\times\Pow(\nabla\times\nabla)$ be the smallest relation satisfying
  \begin{enumerate}
    \item if $(P,Q,R)\in\R'$ for some ep-bisimulation $\R'$ then $(P,Q,R)\in\R$,
    \item if $(P,Q,R)\in\R$ and $\alpha\in\textit{Act}$ then $(\alpha.P,\alpha.Q,\alpha.R)\in\R$, where
      \[
        \alpha.R \coloneqq \{(\actsyn{\alpha}P,\actsyn{\alpha}Q)\} \djcup \{(b{:}\alpha.P,b{:}\alpha.Q) \mid b\in\B \land \alpha\neq b?\}\,.
      \]
  \end{enumerate}
  It suffices to show that $\R$ is an ep-bisimulation.
  I.e., all entries in $\R$ satisfy the requirements of \df{ep-bisimilarity}.
  We proceed by structural induction.

  \vspace{1ex}
  \noindent
  \emph{Induction base}:
  Suppose $(P,Q,R)\in\R'$ for some ep-bisimulation $\R'$.
  Since $\R' \subseteq \R$, all requirements of \df{ep-bisimilarity} are satisfied.

  \vspace{1ex}
  \noindent
  \emph{Induction step}:
  Suppose $(P,Q,R)\in\R$ satisfies all requirements of \df{ep-bisimilarity},
  we prove that $(P',Q',R')$, where $P'=\alpha.P$, $Q'=\alpha.Q$, and $R'=\alpha.R$, also satisfies those requirements.
  This follows directly with Definitions~\ref{df:ep-bisimilarity} and~\ref{df:leadsto}.
\end{proof}

\begin{proposition}\label{pr:congruence for choice}\rm
  If $P_\Left \bisep Q_\Left$ and $P_\Right \bisep Q_\Right$ then $P_\Left+P_\Right \bisep Q_\Left+Q_\Right$.
\end{proposition}
\begin{proof}
  A {\transition} enabled in $P+Q$ is either
  \begin{itemize}
    \item $t+Q$ for some $t\in\nablasub{P}$ with $\ell(t)\notin\B{:}$,
    \item $P+u$ for some $u\in\nablasub{Q}$ with $\ell(u)\notin\B{:}$, or
    \item $t+u$ for some $t\in\nablasub{P}$ and $u\in\nablasub{Q}$ with $\ell(t)=\ell(u)\in\B{:}$\,.
  \end{itemize}
  \vspace{1ex}
  \noindent
  Let $\R \subseteq \cT\times\cT\times\Pow(\nabla\times\nabla)$ be the smallest relation satisfying
  \begin{enumerate}
    \item if $(P,Q,R)\in\R'$ for some ep-bisimulation $\R'$ then $(P,Q,R)\in\R$,
    \item if $(P_\Left,Q_\Left,R_\Left),(P_\Right,Q_\Right,R_\Right)\in\R$ then $(P_\Left+P_\Right,Q_\Left+Q_\Right,R_\Left+R_\Right)\in\R$, where
      \[
        \begin{array}{@{}l@{}l@{}}
          R_\Left+R_\Right \coloneqq & \{(t+P_\Right,v+Q_\Right) \mid t \mathrel{R_\Left} v \land \ell(t) \notin \B{:}\} \djcup {} \\
                                     & \{(P_\Left+u,Q_\Left+w) \mid u \mathrel{R_\Right} w \land \ell(u) \notin \B{:}\} \djcup {} \\
                                     & \{(t+u,v+w) \mid t \mathrel{R_\Left} v \land u \mathrel{R_\Right} w \land \ell(t)=\ell(u)\in\B{:}\}\,.
        \end{array}
      \]
  \end{enumerate}
  It suffices to show that $\R$ is an ep-bisimulation.
  I.e., all entries in $\R$ satisfy the requirements of \df{ep-bisimilarity}.
  We proceed by structural induction.

  \vspace{1ex}
  \noindent
  \emph{Induction base}:
  Suppose $(P,Q,R)\in\R'$ for some ep-bisimulation $\R'$.
  Since $\R' \subseteq \R$, all requirements of \df{ep-bisimilarity} are satisfied.

  \vspace{1ex}
  \noindent
  \emph{Induction step}:
  Suppose $(P_\Left,Q_\Left,R_\Left),(P_\Right,Q_\Right,R_\Right)\in\R$ satisfy all requirements of \df{ep-bisimilarity},
  we prove that $(P,Q,R)$, where $P=P_\Left+P_\Right$, $Q=Q_\Left+Q_\Right$ and $R=R_\Left+R_\Right$, also satisfies those requirements.

  \vspace{1ex}
  \noindent
  $R\subseteq\nablasub{P}\times\nablasub{Q}$ follows from $R_\Left\subseteq\nablasub{P_\Left}\times\nablasub{Q_\Left}$ and $R_\Right\subseteq\nablasub{P_\Right}\times\nablasub{Q_\Right}$.

  \vspace{1ex}
  \noindent
  \emph{Requirement 1.a}:
  It suffices to find, for each $\chi\in\nablasub{P}$, a $\zeta\in\nablasub{Q}$ with $\chi \mathrel R \zeta$.
  \begin{enumerate}
    \item Suppose $\chi=t+P_\Right$ for some $t\in\nablasub{P_\Left}$ with $\ell(t)\notin\B{:}$\,. \\
      We obtain $v\in\nablasub{Q_\Left}$ with $t \mathrel{R_\Left} v$, and pick $\zeta=v+Q_\Right$.
    \item Suppose $\chi=P_\Left+u$ for some $u\in\nablasub{P_\Right}$ with $\ell(u)\notin\B{:}$\,. \\
      We obtain $w\in\nablasub{Q_\Right}$ with $u \mathrel{R_\Right} w$, and pick $\zeta=Q_\Left+w$.
    \item Suppose $\chi=t+u$ for some $t\in\nablasub{P_\Left}$ and $u\in\nablasub{P_\Right}$ with $\ell(t)=\ell(u)\in\B{:}$\,. \\
      We obtain $v\in\nablasub{Q_\Left}$ and $w\in\nablasub{Q_\Right}$ with $t \mathrel{R_\Left} v$ and $u \mathrel{R_\Right} w$, and pick $\zeta=v+w$.
  \end{enumerate}
  In all cases, $\zeta\in\nablasub{Q}$ and $\chi \mathrel R \zeta$ hold trivially.

  \vspace{1ex}
  \noindent
  \emph{Requirement 1.b}:
  The proof is similar to that of Requirement 1.(a) and is omitted.

  \vspace{1ex}
  \noindent
  \emph{Requirement 1.c}:
  This follows directly from the observations that
  \begin{itemize}
    \item $\ell(t)=\ell(v) \implies \ell(t+P_\Right)=\ell(v+Q_\Right)$,
    \item $\ell(u)=\ell(w) \implies \ell(P_\Left+u)=\ell(Q_\Left+w)$, and
    \item $\ell(t)=\ell(v) \land \ell(u)=\ell(w) \implies \ell(t+u)=\ell(v+w)$;
  \end{itemize}
  provided that the composed {\transition}s exist.

  \vspace{1ex}
  \noindent
  \emph{Requirement 2}:
  It suffices to find, for arbitrary $\Upsilon,\Upsilon'$ with $\Upsilon \mathrel R \Upsilon'$, an $R'$ with \\
  $(\target(\Upsilon),\target(\Upsilon'),R')\in\R$, such that
  \begin{enumerate}[(a)]
    \item for arbitrary $\chi,\chi'$ with $\chi \mathrel R \chi'$ and $\chi \leadsto_\Upsilon \zeta$, we can find a $\zeta'$ with
      $\chi' \leadsto_{\Upsilon'} \zeta'$ and $\zeta \mathrel{R'} \zeta'$,
    \item for arbitrary $\chi,\chi'$ with $\chi \mathrel R \chi'$ and $\chi' \leadsto_{\Upsilon'} \zeta'$, we can find a $\zeta$ with
      $\chi \leadsto_\Upsilon \zeta$ and $\zeta \mathrel{R'} \zeta'$.
  \end{enumerate}
  Below we focus merely on (a), as (b) will follow by symmetry.

  \vspace{1ex}
  \noindent
  Suppose $\ell(\Upsilon)\in\B{:}\djcup\bar{\Sig}$.
  Pick $R'=R$.
  Then $(\target(\Upsilon),\target(\Upsilon'),R') = (P,Q,R)\in\R$.
  From $\chi \leadsto_\Upsilon \zeta$ we have $\zeta=\chi$.
  Pick $\zeta'=\chi'$.
  Then $\chi' \leadsto_{\Upsilon'} \zeta'$.
  $\zeta \mathrel{R'} \zeta'$ is given by $\chi \mathrel R \chi'$.

  \vspace{1ex}
  \noindent
  We further split the cases when $\ell(\Upsilon)\in\textit{Act}$.
  \begin{enumerate}
    \item Suppose $\Upsilon=v+P_\Right$ and $\Upsilon'=v'+Q_\Right$ with $v \mathrel{R_\Left} v'$.
      We obtain $R_\Left'$ that satisfies Requirement~2 with respect to $v$ and $v'$.
      Pick $R'=R_\Left'$.
      Then $(\target(\Upsilon),\target(\Upsilon'),R')=(\target(v),\target(v'),R_\Left')\in\R$.
      \begin{enumerate}
        \item Suppose $\chi=t+P_\Right$ and $\chi'=t'+Q_\Right$ with $t \mathrel{R_\Left} t'$.
          From $\chi \leadsto_\Upsilon \zeta$ we have $t \leadsto_v \zeta$.
          Then we obtain $x'$ with $t' \leadsto_{v'} x'$ and $\zeta \mathrel{R_\Left'} x'$.
          Pick $\zeta'=x'$.
          Then $\chi' \leadsto_{\Upsilon'} \zeta'$ follows from $t' \leadsto_{v'} x'$.
        \item Suppose $\chi=P_\Left+u$ and $\chi'=Q_\Left+u'$ with $u \mathrel{R_\Right} u'$.
          We obtain $x'$ with $\zeta \mathrel{R_\Left'} x'$ and pick $\zeta'=x'$.
          From $\chi \leadsto_\Upsilon \zeta$ we have $\ell(\chi)=b?$ and $\ell(\zeta)\in\{b?,b{:}\}$ for some $b\in\B$.
          Then $\chi' \leadsto_{\Upsilon'} \zeta'$ follows from $\ell(\chi')=b?$ and $\ell(x')\in\{b?,b{:}\}$.
        \item Suppose $\chi=t+u$ and $\chi'=t'+u'$ with $t \mathrel{R_\Left} t'$ and $u \mathrel{R_\Right} u'$.
          From $\chi \leadsto_\Upsilon \zeta$ we have $t \leadsto_v \zeta$.
          Then we obtain $x'$ with $t' \leadsto_{v'} x'$ and $\zeta \mathrel{R_\Left'} x'$.
          Pick $\zeta'=x'$.
          Then $\chi' \leadsto_{\Upsilon'} \zeta'$ follows from $t' \leadsto_{v'} x'$.
      \end{enumerate}
      In all cases, $\zeta \mathrel{R'} \zeta'$ is given by $\zeta \mathrel{R_\Left'} x'$.
    \item Suppose $\Upsilon=P_\Left+w$ and $\Upsilon'=Q_\Left+w'$ with $w \mathrel{R_\Right} w'$.
      The proof is similar to that of the previous case.
    \qedhere
  \end{enumerate}
\end{proof}

\begin{proposition}\label{pr:congruence for parallel composition}\rm
  If $P_\Left \bisep Q_\Left$ and $P_\Right \bisep Q_\Right$ then $P_\Left|P_\Right \bisep Q_\Left|Q_\Right$.
\end{proposition}
\begin{proof}
  A {\transition} enabled in $P|Q$ is either
  \begin{itemize}
    \item $t|Q$ for some $t\in\nablasub{P}$ with $\ell(t)\notin\B!\djcup\B?\djcup\B{:}$,
    \item $P|u$ for some $u\in\nablasub{Q}$ with $\ell(u)\notin\B!\djcup\B?\djcup\B{:}$,
    \item $t|u$ for some $t\in\nablasub{P}$ and $u\in\nablasub{Q}$ with $\ell(t)=\overline{\ell(u)}\in\Ch\djcup\bar{\Ch}\djcup\Sig\djcup\bar{\Sig}$, or
    \item $t|u$ for some $t\mathbin\in\nablasub{P}$ and $u\mathbin\in\nablasub{Q}$ with $\{\ell(t),\ell(u)\} \mathbin\in \{\{b!,b?\},\{b!,b{:}\},\{b?\},\{b?,b{:}\},\{b{:}\}\}$ for some $b\in\B$.
  \end{itemize}
  \vspace{1ex}
  \noindent
  Let $\R \subseteq \cT\times\cT\times\Pow(\nabla\times\nabla)$ be the smallest relation satisfying
  \begin{enumerate}
    \item if $(P,Q,R)\in\R'$ for some ep-bisimulation $\R'$ then $(P,Q,R)\in\R$,
    \item if $(P_\Left,Q_\Left,R_\Left),(P_\Right,Q_\Right,R_\Right)\in\R$ then $(P_\Left|P_\Right,Q_\Left|Q_\Right,R_\Left|R_\Right)\in\R$, where
      \[
        \begin{array}{@{}l@{}l@{}}
          R_\Left|R_\Right \coloneqq & \{(t|P_\Right,v|Q_\Right) \mid t \mathrel{R_\Left} v \land \ell(t) \notin \B! \djcup \B? \djcup \B{:}\} \djcup {} \\
                            & \{(P_\Left|u,Q_\Left|w) \mid u \mathrel{R_\Right} w \land \ell(u) \notin \B! \djcup \B? \djcup \B{:}\} \djcup {} \\
                            & \{(t|u,v|w) \mid t \mathrel{R_\Left} v \land u \mathrel{R_\Right} w \land \ell(t)=\overline{\ell(u)} \in \Ch \djcup \bar{\Ch} \djcup \Sig \djcup \bar{\Sig}\} \djcup {} \\
                            & \begin{array}{@{}l@{}l@{}}
                                \{(t|u,v|w) \mid {} & t \mathrel{R_\Left} v \land u \mathrel{R_\Right} w \land {} \\
                                                    & \exists\, b\in\B.~ \{\ell(t),\ell(u)\}\in\{\{b!,b?\},\{b!,b{:}\},\{b?\},\{b?,b{:}\},\{b{:}\}\}\}\,.
                              \end{array}
        \end{array}
      \]
  \end{enumerate}
  It suffices to show that $\R$ is an ep-bisimulation.
  I.e., all entries in $\R$ satisfy the requirements of \df{ep-bisimilarity}.
  We proceed by structural induction.

  \vspace{1ex}
  \noindent
  \emph{Induction base}:
  Suppose $(P,Q,R)\in\R'$ for some ep-bisimulation $\R'$.
  Since $\R' \subseteq \R$, all requirements of \df{ep-bisimilarity} are satisfied.

  \vspace{1ex}
  \noindent
  \emph{Induction step}:
  Suppose $(P_\Left,Q_\Left,R_\Left),(P_\Right,Q_\Right,R_\Right)\in\R$ satisfy all requirements of \df{ep-bisimilarity},
  we prove that $(P,Q,R)$, where $P=P_\Left|P_\Right$, $Q=Q_\Left|Q_\Right$ and $R=R_\Left|R_\Right$, also satisfies those requirements.

  \vspace{1ex}
  \noindent
  $R\subseteq\nablasub{P}\times\nablasub{Q}$ follows from $R_\Left\subseteq\nablasub{P_\Left}\times\nablasub{Q_\Left}$ and $R_\Right\subseteq\nablasub{P_\Right}\times\nablasub{Q_\Right}$.

  \vspace{1ex}
  \noindent
  \emph{Requirement 1.a}:
  It suffices to find, for each $\chi\in\nablasub{P}$, a $\zeta\in\nablasub{Q}$ with $\chi \mathrel R \zeta$.
  \begin{enumerate}
    \item Suppose $\chi=t|P_\Right$ for some $t\in\nablasub{P_\Left}$ with $\ell(t)\notin\B!\djcup\B?\djcup\B{:}$\,. \\
      We obtain $v\in\nablasub{Q_\Left}$ with $t \mathrel{R_\Left} v$ and pick $\zeta=v|Q_\Right$.
    \item Suppose $\chi=P_\Left|u$ for some $u\in\nablasub{P_\Right}$ with $\ell(u)\notin\B!\djcup\B?\djcup\B{:}$\,. \\
      We obtain $w\in\nablasub{Q_\Right}$ with $u \mathrel{R_\Right} w$ and pick $\zeta=Q_\Left|w$.
    \item Suppose $\chi=t|u$ for some $t\in\nablasub{P_\Left}$, $u\in\nablasub{P_\Right}$ with $\ell(t)=\overline{\ell(u)} \in \Ch \djcup \bar{\Ch} \djcup \Sig \djcup \bar{\Sig}$. \\
      We obtain $v\in\nablasub{Q_\Left}$ and $w\in\nablasub{Q_\Right}$ with $t \mathrel{R_\Left} v$ and $u \mathrel{R_\Right} w$, and pick $\zeta=v|w$.
    \item Suppose $\chi=t|u$ for some $t\in\nablasub{P_\Left}$ and $u\in\nablasub{P_\Right}$ with \\
      $\{\ell(t),\ell(u)\}\in\{\{b!,b?\},\{b!,b{:}\},\{b?\},\{b?,b{:}\},\{b{:}\}\}$ for some $b\in\B$. \\
      We obtain $v\in\nablasub{Q_\Left}$ and $w\in\nablasub{Q_\Right}$ with $t \mathrel{R_\Left} v$ and $u \mathrel{R_\Right} w$, and pick $\zeta=v|w$.
  \end{enumerate}
  In all cases, $\zeta\in\nablasub{Q}$ and $\chi \mathrel R \zeta$ hold trivially.

  \vspace{1ex}
  \noindent
  \emph{Requirement 1.b}:
  The proof is similar to that of Requirement 1.(a) and is omitted.

  \vspace{1ex}
  \noindent
  \emph{Requirement 1.c}:
  This follows directly from the observation that
  \begin{itemize}
    \item $\ell(t)=\ell(v) \implies \ell(t|P_\Right)=\ell(v|Q_\Right)$,
    \item $\ell(u)=\ell(w) \implies \ell(P_\Left|u)=\ell(Q_\Left|w)$, and
    \item $\ell(t)=\ell(v) \land \ell(u)=\ell(w) \implies \ell(t|u)=\ell(v|w)$;
  \end{itemize}
  provided that the composed {\transition}s exist.

  \vspace{1ex}
  \noindent
  \emph{Requirement 2}:
  It suffices to find, for arbitrary $\Upsilon,\Upsilon'$ with $\Upsilon \mathrel R \Upsilon'$, an $R'$ with \\
  $(\target(\Upsilon),\target(\Upsilon'),R')\in\R$, such that
  \begin{enumerate}[(a)]
    \item for arbitrary $\chi,\chi'$ with $\chi \mathrel R \chi'$ and $\chi \leadsto_\Upsilon \zeta$, we can find a $\zeta'$ with
      $\chi' \leadsto_{\Upsilon'} \zeta'$ and $\zeta \mathrel{R'} \zeta'$,
    \item for arbitrary $\chi,\chi'$ with $\chi \mathrel R \chi'$ and $\chi' \leadsto_{\Upsilon'} \zeta'$, we can find a $\zeta$ with
      $\chi \leadsto_\Upsilon \zeta$ and $\zeta \mathrel{R'} \zeta'$.
  \end{enumerate}
  Below we focus merely on (a), as (b) will follow by symmetry.
  \begin{enumerate}
    \item Suppose $\Upsilon=v|P_\Right$ and $\Upsilon'=v'|Q_\Right$ with $v \mathrel{R_\Left} v'$.
      We obtain $R_\Left'$ that satisfies Requirement~2 with respect to $v$ and $v'$.
      Pick $R'=R_\Left'|R_\Right$.
      Then $(\target(\Upsilon),\target(\Upsilon'),R')=(\target(v)|P_\Right,\target(v')|Q_\Right,R_\Left'|R_\Right)\in\R$.
      \begin{enumerate}
        \item Suppose $\chi=t|P_\Right$ and $\chi'=t'|Q_\Right$ with $t \mathrel{R_\Left} t'$.
          From $\chi \leadsto_\Upsilon \zeta$ we have $\zeta=x|P_\Right$ for some $x$ with $t \leadsto_v x$.
          Then we obtain $x'$ with $t' \leadsto_{v'} x'$ and $x \mathrel{R_\Left'} x'$.
          Pick $\zeta'=x'|Q_\Right$.
          Then $\chi' \leadsto_{\Upsilon'} \zeta'$ follows from $t' \leadsto_{v'} x'$;
          $\zeta \mathrel{R'} \zeta'$ is given by $x \mathrel{R_\Left'} x'$.
        \item Suppose $\chi=P_\Left|u$ and $\chi'=Q_\Left|u'$ with $u \mathrel{R_\Right} u'$.
          From $\chi \leadsto_\Upsilon \zeta$ we have $\zeta=\target(v)|u$.
          Pick $\zeta'=\target(v')|u'$.
          Then $\chi' \leadsto_{\Upsilon'} \zeta'$ follows directly;
          $\zeta \mathrel{R'} \zeta'$ is given by $u \mathrel{R_\Right} u'$.
        \item Suppose $\chi=t|u$ and $\chi'=t'|u'$ with $t \mathrel{R_\Left} t'$ and $u \mathrel{R_\Right} u'$.
          From $\chi \leadsto_\Upsilon \zeta$ we have $\zeta=x|u$ for some $x$ with $t \leadsto_v x$.
          Then we obtain $x'$ with $t' \leadsto_{v'} x'$ and $x \mathrel{R_\Left'} x'$.
          Pick $\zeta'=x'|u'$.
          Then $\chi' \leadsto_{\Upsilon'} \zeta'$ follows from $t' \leadsto_{v'} x'$;
          $\zeta \mathrel{R'} \zeta'$ is given by $x \mathrel{R_\Left'} x'$ and $u \mathrel{R_\Right} u'$.
      \end{enumerate}
    \item Suppose $\Upsilon=P_\Left|w$ and $\Upsilon'=Q_\Left|w'$ with $w \mathrel{R_\Right} w'$.
      The proof is similar to that of the previous case.
    \item Suppose $\Upsilon=v|w$ and $\Upsilon'=v'|w'$ with $v \mathrel{R_\Left} v'$ and $w \mathrel{R_\Right} w'$.
      We obtain $R_\Left'$ that satisfies Requirement~2 with respect to $v$ and $v'$, and $R_\Right'$ that satisfies Requirement~2 with respect to $w$ and $w'$.
      Pick $R'=R_\Left'|R_\Right'$.
      Then $(\target(\Upsilon),\target(\Upsilon'),R')=(\target(v)|\target(w),\target(v')|\target(w'),R_\Left'|R_\Right')\in\R$.
      \begin{enumerate}
        \item Suppose $\chi=t|P_\Right$ and $\chi'=t'|Q_\Right$ with $t \mathrel{R_\Left} t'$.
          From $\chi \leadsto_\Upsilon \zeta$ we have $\zeta=x|\target(w)$ for some $x$ with $t \leadsto_v x$.
          Then we obtain $x'$ with $t' \leadsto_{v'} x'$ and $x \mathrel{R_\Left'} x'$.
          Pick $\zeta'=x'|\target(w')$.
          Then $\chi' \leadsto_{\Upsilon'} \zeta'$ follows from $t' \leadsto_{v'} x'$;
          $\zeta \mathrel{R'} \zeta'$ is given by $x \mathrel{R_\Left'} x'$.
        \item Suppose $\chi=P_\Left|u$ and $\chi'=Q_\Left|u'$ with $u \mathrel{R_\Right} u'$.
          From $\chi \leadsto_\Upsilon \zeta$ we have $\zeta=\target(v)|y$ for some $y$ with $u \leadsto_w y$.
          Then we obtain $y'$ with $u' \leadsto_{w'} y'$ and $y \mathrel{R_\Right'} y'$.
          Pick $\zeta'=\target(v')|y'$.
          Then $\chi' \leadsto_{\Upsilon'} \zeta'$ follows from $u' \leadsto_{w'} y'$;
          $\zeta \mathrel{R'} \zeta'$ is given by $y \mathrel{R_\Right'} y'$.
        \item Suppose $\chi=t|u$ and $\chi'=t'|u'$ with $t \mathrel{R_\Left} t'$ and $u \mathrel{R_\Right} u'$.
          From $\chi \leadsto_\Upsilon \zeta$ we have $\zeta\mathbin=x|y$ for some $x,y$ with $t \leadsto_v x$ and $u \leadsto_w y$.
          Then we obtain $x',y'$ with $t' \leadsto_{v'} x'$, $u' \leadsto_{w'} y'$, $x \mathrel{R_\Left'} x'$, and $y \mathrel{R_\Right'} y'$.
          Pick $\zeta'=x'|y'$.
          Then $\chi' \leadsto_{\Upsilon'} \zeta'$ follows from $t' \leadsto_{v'} x'$ and $u' \leadsto_{w'} y'$;
          $\zeta \mathrel{R'} \zeta'$ is given by $x \mathrel{R_\Left'} x'$ and $y \mathrel{R_\Right'} y'$.
        \qedhere
      \end{enumerate}
  \end{enumerate}
\end{proof}

\begin{proposition}\label{pr:congruence for restriction}\rm
  If $P \bisep Q$ and $L\subseteq\Ch\djcup\Sig$ then $P\backslash L \bisep Q\backslash L$.
\end{proposition}
\begin{proof}
  A {\transition} enabled in $P\backslash L$ is $t\backslash L$ for some $t\in\nablasub{P}$ with $\ell(t)\notin L\djcup\overline{L}$.
  \\[1ex]
  Let $\R \subseteq \cT\times\cT\times\Pow(\nabla\times\nabla)$ be the smallest relation satisfying
  \begin{enumerate}
    \item if $(P,Q,R)\in\R'$ for some ep-bisimulation $\R'$ then $(P,Q,R)\in\R$,
    \item if $(P,Q,R)\in\R$ and $L\subseteq\Ch\djcup\Sig$ then $(P\backslash L,Q\backslash L,R\backslash L)\in\R$, where
      \[
        R\backslash L \coloneqq \{(t\backslash L,v\backslash L) \mid t \mathrel R v \land
        \ell(t)\notin L \djcup\overline{L}\}\,.
      \]
  \end{enumerate}
  It suffices to show that $\R$ is an ep-bisimulation.
  I.e., all entries in $\R$ satisfy the requirements of \df{ep-bisimilarity}.
  We proceed by structural induction.

  \vspace{1ex}
  \noindent
  \emph{Induction base}:
  Suppose $(P,Q,R)\in\R'$ for some ep-bisimulation $\R'$.
  Since $\R' \subseteq \R$, all requirements of \df{ep-bisimilarity} are satisfied.

  \vspace{1ex}
  \noindent
  \emph{Induction step}:
  Suppose $(P,Q,R)\in\R$ satisfies all requirements of \df{ep-bisimilarity},
  we prove that $(P',Q',R')$, where $P'=P\backslash L$, $Q'=Q\backslash L$, and $R'=R\backslash L$, also satisfies those requirements.
  This follows directly with Definitions~\ref{df:ep-bisimilarity} and~\ref{df:leadsto}.
\end{proof}

\begin{proposition}\label{pr:congruence for relabelling}\rm
  If $P \bisep Q$ and $f$ is a relabelling then $P[f] \bisep Q[f]$.
\end{proposition}
\begin{proof}
  An easy structural induction on the structure of $P$;
  similar to the proof for restriction (\pr{congruence for restriction}).
\end{proof}

\begin{proposition}\label{pr:congruence for signalling}\rm
  If $P \bisep Q$ and $s\in\Sig$ $P\signals s \bisep Q\signals s$.
\end{proposition}
\begin{proof}
  An easy structural induction on the structure of $P$;
  similar to the proof for restriction (\pr{congruence for restriction}).
\end{proof}

\arxiv{
\newpage
\section{Synchrons -- an Alternative Interpretation of \ABCdE as an LTSS}\label{Synchrons}

In \cite{synchrons}, each transition is seen as the synchronisation of a number of so-called synchrons.
Each synchron of a transition $t\in\Tr$ represents a path in the proof tree $t$ from its root to a leaf.
If $\source(t)=P \in \cT$ a synchron of $t$ is also a synchron of $P$,
  and can be seen as a path in the parse tree of $P$ to an unguarded subexpression
  ${\bf 0}$, $\alpha.Q$ or $Q\signals s$ of $P$
  -- except that recursion \plat{$A \defis P$} gets unfolded in the construction of such a path.
Here a subexpression of $P$ occurs \emph{unguarded} if it does not lay within a subexpression $\beta.O$ of $P$.

\begin{definition}[Synchrons \cite{synchrons}]\label{df:synchron}\rm
  Let the set $\textit{Arg}$ of \emph{arguments} be
    $\{{+_\Left},{+_\Right},{|_\Left},{|_\Right},{\backslash L},[f],A{:},{{}\signals r}
    \mid L\subseteq\Ch\djcup\Sig \land f \text{ is a relabelling} \land A\in\A \land r\in\Sig\}$.
  A \emph{synchron} is an expression $\sigma(\actsyn{\alpha}P)$, $\sigma(b{:})$ or $\sigma(P\sigsyn{s})$
    with $\sigma\in\textit{Arg}^*$, $\alpha\in\textit{Act}$, $b\in\B$, $s\in\Sig$ and $P\in\cT$.
  An arguments $\iota\in\textit{Arg}$ is applied componentwise to a set $\Sigma$ of synchrons:
    $\iota(\Sigma) \coloneqq \{\iota\varsigma \mid \varsigma\in\Sigma\}$.
  The set $\varsigma(P)$ of synchrons of an \ABCdE process $P$ is inductively defined by \vspace{-1ex}
  \[\begin{array}{l@{~=~}l@{\hspace{1cm}}l@{~=~}l}
    \varsigma({\bf 0})       & \{(b{:}) \mid b\in\B\} &
    \varsigma(\alpha.P)      & \{(\actsyn{\alpha}P)\} \djcup \{(b{:}) \mid b\in\B \land b?\neq\alpha\} \\
    \varsigma(P{+}Q)         & {+_\Left}\varsigma(P) \djcup {+_\Right}\varsigma(Q) &
    \varsigma(P|Q)           & {|_\Left}\varsigma(P) \djcup {|_\Right}\varsigma(Q) \\
    \varsigma(P\backslash L) & {\backslash L}\,\varsigma(P) &
    \varsigma(P[f])          & [f]\varsigma(P) \\
    \varsigma(A)             & A{:}\varsigma(P) ~\text{when}~ \plat{$A \defis P$} &
    \varsigma(P\signals s)   & \{(P\sigsyn{s})\} \djcup {{}\signals s}\,\varsigma(P) \,.
  \end{array}\]
  The set $\varsigma(t)$ of synchrons of a {\transition} $t$ is inductively defined by \vspace{-1ex}
  \[\begin{array}{l@{~=~}l@{\hspace{1cm}}l@{~=~}l@{\hspace{1cm}}l@{~=~}l}
    \varsigma(b{:}{\bf 0})      & \{(b{:})\} &
    \varsigma(\actsyn{\alpha}P) & \{(\actsyn{\alpha}P)\} &
    \varsigma(b{:}\alpha.P)     & \{(b{:})\} \\
    \varsigma(t{+}Q)            & {+_\Left}\varsigma(t) &
    \varsigma(t{+}u)            & {+_\Left}\varsigma(t) \djcup {+_\Right}\varsigma(u) &
    \varsigma(P{+}u)            & {+_\Right}\varsigma(u) \\
    \varsigma(t|Q)              & {|_\Left}\varsigma(t) &
    \varsigma(t|u)              & {|_\Left}\varsigma(t) \djcup {|_\Right}\varsigma(u) &
    \varsigma(P|u)              & {|_\Right}\varsigma(u) \\
    \varsigma(t\backslash L)    & {\backslash L}\,\varsigma(t) &
    \varsigma(t[f])             & [f]\varsigma(t) &
    \varsigma(A{:}t)            & A{:}\varsigma(t) \\
    \varsigma(P\sigsyn{s})      & \{(P\sigsyn{s})\} &
    \varsigma(t\signals r)      & {{}\signals r}\,\varsigma(t) \,.
  \end{array}\]
\end{definition}
Note that we use the symbol $\varsigma$ as a variable ranging over synchrons, and as the name of two
  functions -- disambiguated by context.
Also $\upsilon$, $\nu$ and $\xi$ will range over synchrons.
Here $\upsilon$ -- the Greek letter \emph{upsilon} -- should not be confused with the Latin letter $v$,
  denoting~a~transition.

\begin{lemma}[\cite{synchrons}]\label{lem:synchrons of source}\rm
  If $t\in\nabla$ and $P=\source(t)$ then $\varsigma(t)\subseteq\varsigma(P)$.
\end{lemma}
\begin{proof}
  A structural induction on $t$.
\end{proof}

\begin{lemma}\label{lem:transition uniquely identified by source and synchrons}\rm
  A {\transition} $t$ is uniquely identified by $source(t)$ and $\varsigma(t)$.
\end{lemma}
\begin{proof}
  We support this lemma with a partial function $\textit{retrieve}$ to retrieve the {\transition} from its source and set of synchrons
  -- provided that the corresponding {\transition} exists.
  $\textit{retrieve}(P,\Sigma)$ where $P\in\cT$ and $\Sigma\subseteq\varsigma(P)$ is inductively defined by
  \vspace{-1ex}
  \[\begin{array}[b]{@{\textit{retrieve}}l@{,~}l@{~=~}ll}
    ({\bf 0}       & \{(b{:})\})                                           & b{:}{\bf 0}                                                          & \\
    (\alpha.P      & \{(\plat{$\actsyn{\alpha}P$})\})                      & \plat{$\actsyn{\alpha}P$}                                            & \\
    (\alpha.P      & \{(b{:})\})                                           & b{:}\alpha.P                                                         & \\
    (P+Q           & {+_\Left}\Sigma_\Left)                                & \textit{retrieve}(P,\Sigma_\Left)+Q                                  & \\
    (P+Q           & {+_\Right}\Sigma_\Right)                              & P+\textit{retrieve}(Q,\Sigma_\Right)                                 & \\
    (P+Q           & {+_\Left}\Sigma_\Left \djcup {+_\Right}\Sigma_\Right) & \textit{retrieve}(P,\Sigma_\Left)+\textit{retrieve}(Q,\Sigma_\Right) & \text{when}~ \Sigma_\Left,\Sigma_\Right\neq\emptyset \\
    (P|Q           & {|_\Left}\Sigma_\Left)                                & \textit{retrieve}(P,\Sigma_\Left)|Q                                  & \\
    (P|Q           & {|_\Right}\Sigma_\Right)                              & P|\textit{retrieve}(Q,\Sigma_\Right)                                 & \\
    (P|Q           & {|_\Left}\Sigma_\Left \djcup {|_\Right}\Sigma_\Right) & \textit{retrieve}(P,\Sigma_\Left)|\textit{retrieve}(Q,\Sigma_\Right) & \text{when}~ \Sigma_\Left,\Sigma_\Right\neq\emptyset \\
    (P\backslash L & {\backslash L}\,\Sigma)                               & \textit{retrieve}(P,\Sigma)\backslash L                              & \\
    (P[f]          & [f]\Sigma)                                            & \textit{retrieve}(P,\Sigma)[f]                                       & \\
    (A             & A{:}\Sigma)                                           & A{:}\textit{retrieve}(P,\Sigma)                                      & \text{when}~ \plat{$A \defis P$} \\
    (P\signals s   & \{(P\sigsyn{s})\})                                    & P\sigsyn{s}                                                          & \\
    (P\signals r   & {{}\signals r}\,\Sigma)                               & \textit{retrieve}(P,\Sigma)\signals r                                & .
  \end{array}\qedhere\]
\end{proof}

\subsection{Concurrency Relations between Synchrons and Transitions}

\begin{definition}[Active and Necessary Synchrons]\label{df:active and necessary synchrons}\rm
  All synchrons of the form \plat{$\sigma(\actsyn{\alpha}P)$} are \emph{active}; synchrons
  $\sigma(P\sigsyn{s})$ and $\sigma(b{:})$ are \emph{passive}. Moreover,
  all synchrons that are not of the form \plat{$\sigma(\actsyn{b?}P)$} or \plat{$\sigma(b{:})$}
  are \emph{necessary}.
  Let $a\varsigma(t)$ denote the set of active synchrons of the {\transition} $t$,
  and $n\varsigma(t)$ its set of necessary synchrons.
\end{definition}

\noindent
Intuitively, the execution of an active synchron \plat{$\sigma(\actsyn{\alpha}P)$}
  causes a transition \plat{$\alpha.P \mathbin{\goto{\alpha}} P$} in the relevant component of the represented system,
  whereas the execution of a passive synchron does not cause any state change.
Moreover, since for any $b\mathop{\in}\B$ a translation labelled $b?$ or $b{:}$ is enabled by any \ABCdE process in any state,
  whether a broadcast action $b!$ can take place depends only on the presence of a b!-synchron;
  its synchronisation partners of the form \plat{$\sigma(\actsyn{b?}P)$} or \plat{$\sigma(b{:})$}
  are not necessary for a $b!$-transition to occur.
In this paper we choose to see a transition labelled $b?$ or $b{:}$ merely as
  the potential for receiving or discarding a broadcast $b!$;
  none of the synchrons of such a transition is necessary.

\begin{definition}[Concurrency \cite{synchrons}]\label{df:sconc}\rm
  Two synchrons $\varsigma$ and $\upsilon$ are \emph{concurrent}, notation $\varsigma \sconc_d \upsilon$,
    if $\varsigma=\sigma{|_\D}\varsigma'$ and $\upsilon=\sigma{|_\E}\upsilon'$
    with $\sigma\in\textit{Arg}^*$, $\{\mathrm{D},\mathrm{E}\}=\{\mathrm{L},\mathrm{R}\}$, and $\varsigma',\upsilon'$ synchrons.

  A synchron $\varsigma$ is \emph{unaffected} by a {\transition} $u$, notation $\varsigma \saconc_d u$,
    if $\forall \upsilon\in a\varsigma(u).~ \varsigma \sconc_d \upsilon$.

  A {\transition} $t$ is \emph{unaffected} by a {\transition} $u$, notation $t \ssaconc u$,
    iff $\source(t)=\source(u)$ and $\forall \varsigma\in n\varsigma(t).~ \varsigma \saconc_d u$.
\end{definition}

\begin{example}
  Let $P = (b!.{\bf 0})\signals s \mid (s.{\bf 0} + b?.{\bf 0}) \in \cT$,
    $t = (\actsyn{b!}{\bf 0})\signals s \mid (s.{\bf 0} + (\actsyn{b?}{\bf 0})) \in \Tr$
    and $u = (b!.{\bf 0})\sigsyn{s} \mid ((\actsyn{s}{\bf 0}) + b?.{\bf 0}) \in \Tr$.
  Then $\ell(t)=b!$, $\ell(u)=\tau$, $\source(t)=\source(u)=P$, $\target(t)={\bf 0}|{\bf 0}$ and
    $\target(u)=(b!.{\bf 0})\signals s \mid {\bf 0}$.
  Transition $t$ has exactly two synchrons, namely \plat{$\varsigma \coloneqq {|_\Left}{{}\signals s}\,(\actsyn{b!}{\bf 0})$},
    and \plat{$\upsilon \coloneqq {|_\Right}{+_\Right}(\actsyn{b?}{\bf 0})$}.
  Also $u$ has two synchrons, namely $\nu \coloneqq {|_\Left} (b!.{\bf 0})\sigsyn{s}$ and
    \plat{$\xi \coloneqq {|_\Right}{+_\Left}(\actsyn{s}{\bf 0})$}.
  Of these synchrons, $\varsigma$, $\upsilon$ and $\xi$ are active, whereas $\varsigma$, $\nu$ and $\xi$ are necessary.

  We have $\varsigma \sconc_d \upsilon$, since these two synchrons stem from opposite sides of a parallel composition.
  Likewise $\varsigma \sconc_d \xi$, $\upsilon \sconc_d \nu$ and $\nu \sconc_d \xi$,
    yet $\varsigma \nsconc_d \nu$ and $\upsilon \nsconc_d \xi$.

  Transition $u$ is affected by $t$, $u \naconc' t$, since its necessary synchron $\nu$ is affected by $t$.
  I.e., $\nu \nsaconc_d t$. This is because $\varsigma$ is an active synchron of $t$.
  This verdict is consistent with the observation that after doing $t$ it is no longer possible to perform a $\tau$-action.

  Yet $t \ssaconc u$.
  I.e., $t$ is not affected by $u$.
  Namely $\varsigma$, the only necessary synchron of $t$, is not affected by $u$; \ie $\varsigma \saconc_d u$.
  This holds because $\varsigma \sconc_d \xi$.
  The fact that $\varsigma \nsconc_d \nu$ does not stand in the way of $\varsigma \saconc_d u$,
    because the synchron $\nu$ is not active.
  The fact that $\upsilon \nsconc_d \xi$ and hence $\upsilon \naconc_d u$ does not stand in the way of $t \ssaconc u$,
    because the synchron $\upsilon$ is not necessary.
  This verdict is consistent with the observation that after doing $u$ it is still possible to perform a $b!$-action.
  \hfill$\lrcorner$
\end{example}

\noindent
Let $\Tr^{s\bullet} \subseteq \Tr$ be the set of transitions $t$ with $\ell(t)\notin\B?\djcup\B{:}$\,.
In \cite{synchrons}, a concurrency relation ${\aconc} \subseteq \Tr^{s\bullet}\times\Tr$
  that also relates transitions that do not share the same source, was defined.
Restricted to pairs of transitions that do share the same source it coincides with the relation $\ssaconc$ of \df{sconc},
  as was remarked in \cite[Section~11]{synchrons}.
In \cite{synchrons} the relation $\sconc_d$ between synchrons was called \emph{direct} concurrency, which explains the subscript $d$.
There it was used to define a more liberal relation $\sconc$ between synchrons,
  which in turn was used to define the relation ${\aconc} \subseteq \Tr^{s\bullet}\times\Tr$.
Since in this paper we study concurrency relations merely between transitions that share the same source,
  we may use $\sconc_d$ instead of $\sconc$.

Later on, we will show that the relation $\ssaconc$ coincides with the relation $\aconc$ from the present paper,
  which is derived from \df{leadsto}.

\subsection{Classifying the Sets of Synchrons of Transitions}

Define the \emph{label} of a synchron by $\ell(\sigma(\actsyn{\alpha}P))=\alpha$, $\ell(\sigma(b{:}))=b{:}$ and
  $\ell(\sigma(P\sigsyn{s}))=\bar{s}$.

\begin{lemma}\label{lem:synchrons of transition}\rm
  Let $t \in \Tr$.
  \begin{itemize}
    \item Either $t$ has exactly one synchron $\varsigma$, with
      $\ell(\varsigma) = \ell(t) \in \Ch \djcup \bar{\Ch} \djcup \{\tau\} \djcup \Sig \djcup \bar{\Sig}$,
    \item or $\varsigma(t) = \{\varsigma,\upsilon\}$, with
      $\ell(\varsigma) \in \Ch\djcup\bar{\Ch}\djcup\Sig\djcup\bar{\Sig}$, $\ell(\upsilon)=\overline{\ell(\varsigma)}$,
      $\varsigma \sconc_d \upsilon$ and $\ell(t)=\tau$,
    \item or $\varsigma(t)\neq\emptyset$, there is a $b\in\B$ such that
      $\ell(\varsigma)\in\{b!,b?,b{:}\}$ for all $\varsigma\in\varsigma(t)$, and
      \begin{itemize}
        \item $\ell(t)=b!$ and $\ell(\varsigma)=b!$ for exactly one $\varsigma\in\varsigma(t) $, or
        \item $\ell(t)=b?$ and $b? \in \{\ell(\varsigma) \mid \varsigma\in\varsigma(t)\} \subseteq \{b?,b{:}\}$, or
        \item $\ell(t)=b{:}$ and $\ell(\varsigma)=b{:}$ for all $\varsigma\in\varsigma(t)$.
      \end{itemize}
  \end{itemize}
\end{lemma}
\begin{proof}
  A trivial induction on the definition of $\varsigma(t)$.
\end{proof}

\noindent
The next result classifies exactly which sets of synchrons $\Sigma\subseteq\varsigma(P)$ arise as $\Sigma=\varsigma(t)$
  for some transition $t\in\en(P)$.

\begin{definition}\label{df:P-complete}\rm
  A set of synchrons $\Sigma\subseteq \varsigma(P)$ is \emph{$P$-complete} if
  \begin{enumerate}
    \item either $\Sigma=\{\varsigma\}$ with
      $\ell(\varsigma) \in \Ch \djcup \bar{\Ch} \djcup \{\tau\} \djcup \Sig \djcup \bar{\Sig}$,
    \item or $\Sigma = \{\varsigma,\upsilon\}$, with
      $\ell(\varsigma) \in \Ch\djcup\bar{\Ch}\djcup\Sig\djcup\bar{\Sig}$, $\ell(\upsilon)=\overline{\ell(\varsigma)}$,
      and $\varsigma \sconc_d \upsilon$,
    \item or there is a $b\in\B$ such that\label{third}
      \begin{itemize}
        \item $\ell(\varsigma)\in\{b!,b?,b{:}\}$ for all $\varsigma\in\Sigma$,
        \item $\ell(\varsigma)=b!$ for at most one $\varsigma\in\Sigma$,
        \item $\varsigma \sconc_d \upsilon$ for each pair $\varsigma,\upsilon \in \Sigma$
          with $\varsigma\neq\upsilon$ and either $\ell(\varsigma)\neq b{:}$ or $\ell(\upsilon)\neq b{:}$,
        \item if $\varsigma \mathop{\in} \varsigma(P)$ with $\ell(\varsigma) \mathbin{=} b?$ and
          $\varsigma \sconc_d \upsilon$ for each active synchron $\upsilon \mathop{\in} \Sigma{\setminus}\{\varsigma\}$,
          then $\varsigma \mathop{\in} \Sigma$,~and
        \item if $\varsigma \mathop{\in} \varsigma(P)$ with $\ell(\varsigma) \mathbin{=} b{:}\,$ and
          $\varsigma \sconc_d \upsilon$ for each active synchron $\upsilon \mathop{\in} \Sigma$,
          then $\varsigma \mathop{\in} \Sigma$.
          \hspace*{33.6pt}
      \end{itemize}
  \end{enumerate}
\end{definition}

\begin{theorem}\label{thm:P-complete}\rm
  A set of synchrons $\Sigma\subseteq\varsigma(P)$ is $P$-complete iff $\varsigma(t)=\Sigma$ for some $t\in\en(P)$.
\end{theorem}
\begin{proof}
  `If' follows by a straightforward induction on the definition of $\varsigma(t)$, using \lem{synchrons of transition}.

  For `only if', the first case of \df{P-complete} follows by structural induction on $\varsigma$.
  Using this, the second case of \df{P-complete} follows by structural induction on the common prefix $\sigma$ of the two synchrons in $\Sigma$.
  The third case of \df{P-complete} proceeds by induction on the well-founded order $<$ on sets $\Sigma$ of synchrons,
    defined by $\Sigma'<\Sigma$ iff $\iota\Sigma' \subseteq \Sigma$ for some $\iota\in\textit{Arg}$.
  All steps are straightforward.
\end{proof}

\subsection{From Source to Target}\label{sec:after}

We now explore how the sets of synchrons of processes evolve with transitions
  -- \ie how $\varsigma(P)$ becomes $\varsigma(Q)$ when a transition $w$ with $\source(w)=P$ and $\target(w)=Q$ is taken.

If a process $P$ has a synchron $\varsigma$ that is not affected by a transition $w\in\en(P)$,
  then after executing $w$ a variant of $\varsigma$ ought to remain a synchron of the target state $\target(w)$.
We will denote this variant as $\varsigma@w$.

\begin{example}\label{ex:after}
  Let $P \mathbin{=} (a.{\bf 0}|d.{\bf 0}){+}c.{\bf 0}$, \plat{$\varsigma \mathbin{=} {+_\Left}{|_\Right}(\actsyn{d}{\bf 0}) \mathop{\in}\hspace{-.7pt} \varsigma(P)$} and
    $w \mathbin{=} ((\actsyn{a}{\bf 0})|d.{\bf 0}){+}c.{\bf 0} \mathop{\in}\hspace{-.7pt} \en(P)$.
  Then $\varsigma \saconc_d w$ and \plat{$\varsigma@w = {|_\Right}(\actsyn{d}{\bf 0}) \in \en(Q)$}, where $Q = \target(w) = {\bf 0}|d.{\bf 0}$.
  \hfill$\lrcorner$
\end{example}

\begin{definition}[Dynamic and Static Arguments~\cite{synchrons}]\label{df:dynamic and static arguments}\rm
  The arguments ${+_\Left}$, ${+_\Right}$, $A{:}$ and $\signals r$ are called \emph{dynamic}; the others are \emph{static}.
  For $\sigma\in\textit{Arg}^*$, let $\textit{static}(\sigma)$ be the result of removing all dynamic arguments from $\sigma$.
  Moreover, for $\varsigma = \sigma\upsilon$ with \plat{$\upsilon \in \{(\actsyn{\alpha}P),(P\sigsyn{s}),(b{:})\}$},
    define $\textit{static}(\varsigma) \coloneqq \textit{static}(\sigma)\upsilon$.
\end{definition}

\begin{definition}[After Function~\cite{synchrons}]\label{df:@}\rm
  Let $\varsigma,\upsilon$ be synchrons with $\varsigma \sconc_d \upsilon$,
    \ie $\varsigma = \sigma{|_\D}\varsigma'$ and $\upsilon = \sigma{|_\E}\upsilon'$
    with $\sigma\in\textit{Arg}^*$, $\{\mathrm{D},\mathrm{E}\} = \{\mathrm{L},\mathrm{R}\}$, and $\varsigma',\upsilon'$ synchrons.
  Define $\varsigma@\upsilon$, where $@$ is pronounced `after', to be $\textit{static}(\sigma){|_\D}\varsigma'$.

  For $w\in\Tr$ with $\ell(w)\in\textit{Act}$ and $\varsigma \saconc_d w$,
    \ie $\forall \upsilon\in a\varsigma(w).~ \varsigma \sconc_d \upsilon$,
    let $\varsigma@w \coloneqq \varsigma@\upsilon$ for the $\upsilon\in a\varsigma(w)$ that is `closest' to $\varsigma$,
    in the sense that it has the largest prefix in common with $\varsigma$.
  Note that $\ell(w)\in\textit{Act}$ implies $a\varsigma(w)\neq\emptyset$.

  For $w\in\Tr$ with $\ell(w)\in\B{:}\djcup\bar{\Sig}$, let $\varsigma@w \coloneqq \varsigma$.
  In this case, $\varsigma \saconc_d w$ trivially holds.

  For a set $\Sigma$ of synchrons, define $\Sigma@w \coloneqq \{\varsigma@w \mid \varsigma\in\Sigma \land \varsigma \saconc_d w\}$.
\end{definition}

\noindent
The next lemma says that the concurrency relation between synchrons is not affected by the execution of a concurrent transition.

\begin{lemma}\label{lem:@ sconc_d @}\rm
  Let $w\in\Tr$ and $\varsigma,\upsilon$ be synchrons with $\varsigma \saconc_d w$ and $\upsilon \saconc_d w$.
  Then $\varsigma \sconc_d \upsilon \Leftrightarrow \varsigma@w \sconc_d \upsilon@w$.
\end{lemma}
\begin{proof}
  The case $\ell(w)\in\B{:}\djcup\bar{\Sig}$ is trivial, so suppose $\ell(w)\in\textit{Act}$.
  We obtain $\nu\in a\varsigma(w)$ that is `closest' to $\varsigma$ and $\xi\in a\varsigma(w)$ that is `closest' to $\upsilon$.
  Then $\varsigma@w=\varsigma@\nu$ and $\upsilon@w=\upsilon@\xi$.

  \vspace{1ex}
  \noindent
  Suppose $\nu=\xi$.

  If the common prefix shared by $\varsigma$ and $\nu$ is the same as that shared by $\upsilon$ and $\nu$,
    then $\varsigma=\sigma{|_\D}\varsigma'$, $\upsilon=\sigma{|_\D}\upsilon'$ and $\nu=\sigma{|_\E}\nu'$
    with $\sigma\in\textit{Arg}^*$, $\{\mathrm{D},\mathrm{E}\}=\{\mathrm{L},\mathrm{R}\}$, and $\varsigma',\upsilon',\nu'$ synchrons.
  Thus $\varsigma@w \mathbin{=} \textit{\static}(\sigma){|_\D}\varsigma'$ and
    $\upsilon@w \mathbin{=} \textit{\static}(\sigma){|_\D}\upsilon'$\!.
  Hence $\varsigma \mathbin{\sconc_d} \upsilon \Leftrightarrow \varsigma' \mathbin{\sconc_d} \upsilon' \Leftrightarrow \varsigma@w \mathbin{\sconc_d} \upsilon@w$.

  If the common prefix shared by $\varsigma$ and $\nu$ is shorter than that shared by $\upsilon$ and $\nu$,
    then $\varsigma=\sigma{|_\D}\varsigma'$, $\upsilon=\sigma{|_\E}\sigma'{|_{\D'}}\upsilon'$ and $\nu=\sigma{|_\E}\sigma'{|_{\E'}}\nu'$
    with $\sigma,\sigma'\in\textit{Arg}^*$, $\{\mathrm{D},\mathrm{E}\} = \{\mathrm{D}',\mathrm{E}'\} = \{\mathrm{L},\mathrm{R}\}$, and $\varsigma',\upsilon',\nu'$ synchrons.
  Hence $\varsigma \sconc_d \upsilon$ and $\varsigma@w \sconc_d \upsilon@w$ always hold.

  If the common prefix shared by $\varsigma$ and $\nu$ is longer than that shared by $\upsilon$ and $\nu$,
    the proof is similar to that of the previous case.

  \vspace{1ex}
  \noindent
  Suppose $\nu\neq\xi$ and the first argument where $\nu$ and $\xi$ differ is a ${+_\Left}$ versus ${+_\Right}$.
  We have $\nu=\sigma{+_\D}\nu'$ and $\xi=\sigma{+_\E}\xi'$
    with $\sigma\in\textit{Arg}^*$, $\{\mathrm{D},\mathrm{E}\}=\{\mathrm{L},\mathrm{R}\}$, and $\nu',\xi'$ synchrons.

  If either the common prefix shared by $\varsigma$ and $\nu$ or that shared by $\upsilon$ and $\xi$ is shorter than $\sigma$,
    then the proof is similar to that for the case where $\nu=\xi$.
  Alternatively, in this case we could have chosen $\nu=\xi$, thereby reducing this case to the previous one.

  Otherwise, $\varsigma$ has a prefix $\sigma{+_\D}$ and $\upsilon$ has a prefix $\sigma{+_\E}$.
  Then $\varsigma \nsconc_d \xi$ and $\upsilon \nsconc_d \nu$, contradicting the assumptions $\varsigma \saconc_d w$ and $\upsilon \saconc_d w$.

  \vspace{1ex}
  \noindent
  Suppose $\nu\neq\xi$ and the first argument where $\nu$ and $\xi$ differ is a ${|_\Left}$ versus ${|_\Right}$.
  We have $\nu=\sigma{|_\D}\nu'$ and $\xi=\sigma{|_\E}\xi'$
    with $\sigma\in\textit{Arg}^*$, $\{\mathrm{D},\mathrm{E}\}=\{\mathrm{L},\mathrm{R}\}$, and $\nu',\xi'$ synchrons.

  Again, if either the common prefix shared by $\varsigma$ and $\nu$ or that shared by $\upsilon$ and $\xi$ is shorter than $\sigma$,
    then the proof is similar to that for the case where $\nu=\xi$.

  Otherwise, $\varsigma$ has a prefix $\sigma{|_\D}$, $\upsilon$ has a prefix $\sigma{|_\E}$,
    $\varsigma@w$ has a prefix $\textit{static}(\sigma){|_\D}$ and $\upsilon@w$ has a prefix $\textit{static}(\sigma){|_\E}$.
  Hence $\varsigma \sconc_d \upsilon$ and $\varsigma@w \sconc_d \upsilon@w$ always hold.
\end{proof}

\noindent
The following lemma says that, for $w\in\Tr$, the operation $@w$ on synchrons of the same process $P$ is injective.

\begin{lemma}\label{lem:injective}\rm
  If $\varsigma,\upsilon\in\varsigma(P)$ with $\varsigma \mathbin{\neq} \upsilon$
  and $\varsigma \mathbin{\saconc_d} w$, $\upsilon \mathbin{\saconc_d} w$ for some $w \mathbin{\in} \Tr$,
  then $\varsigma@w \mathbin{\neq} \upsilon@w$.
\end{lemma}
\begin{proof}
  If $\ell(w)\in\B{:}\djcup\bar{\Sig}$ then $\varsigma@w = \varsigma \neq \upsilon = \upsilon@w$.
  If $\ell(w)\in\textit{Act}$, then let $\nu\in a\varsigma(w)$ be the synchron `closest' to
  $\varsigma$, and $\xi\in a\varsigma(w)$ be the one `closest' to $\upsilon$,
    so that $\varsigma@w=\varsigma@\nu$ and $\upsilon@w=\upsilon@\xi$.
  Suppose $\varsigma@w=\upsilon@w$.
  Since $\varsigma$ and $\upsilon$ are synchrons of the same process, the first argument where they differ must be either (a) a ${+_\Left}$ versus ${+_\Right}$ or (b) a ${|_\Left}$ versus ${|_\Right}$.
  $\varsigma$ and $\upsilon$ cannot become the same after the $@w$ operation in the case of (b), because $@w$ can only remove dynamic arguments.
  The only possibility is (a) and $@w$ removes at least one of these arguments.
  Let $\varsigma=\sigma{+_\D}\varsigma'$ and $\upsilon=\sigma{+_\E}\upsilon'$
    with $\sigma\in\textit{Arg}^*$, $\{\mathrm{D},\mathrm{E}\}=\{\mathrm{L},\mathrm{R}\}$, and $\varsigma',\upsilon'$ synchrons.
  Assume that ${+_\D}$ is removed from $\varsigma$ by $@w$; the other case will follow by symmetry.
  Then $\nu=\sigma{+_\D}\nu'$ for some active synchron $\nu'$.
  Hence $\upsilon \nsconc_d \nu \in a\varsigma(w)$, contradicting the assumption $\upsilon \saconc_d w$.
  Thus $\varsigma@w \neq \upsilon@w$.
\end{proof}

\noindent
The synchrons affected by a transition are `destroyed' after that transition is taken.
At the same time, some new synchrons are introduced.
We denotes the set of new synchrons introduced by taking $w\in\Tr$ as $\textit{new}(w)$.

\begin{example}
  Let $P=\alpha.Q$ and \plat{$w = (\actsyn{\alpha}Q) \in \en(P)$}.
  Taking the transition $w$ `activates' the active synchron \plat{$(\actsyn{\alpha}Q)$},
    which removes the `guard' $\alpha$ and `reveals' the synchrons in $Q$.
  Then $\textit{new}(w) = \varsigma(Q) \subseteq \varsigma(\target(w))$.
  \hfill$\lrcorner$
\end{example}

\begin{definition}[New Function]\label{df:new}\rm
  Let $\upsilon$ be an active synchron.
  I.e., \plat{$\upsilon=\sigma(\actsyn{\alpha}P)$} for some $\sigma\in\textit{Arg}^*$, $\alpha\in\textit{Act}$ and $P\in\cT$.
  Define $\textit{new}(\upsilon) \coloneqq \{\textit{static}(\sigma)\varsigma \mid \varsigma\in\varsigma(P)\}$.

  For $w\in\Tr$, let $\textit{new}(w) \coloneqq \bigdjcup\{\textit{new}(\upsilon) \mid \upsilon\in a\varsigma(w)\}$.
\end{definition}

\begin{lemma}\label{lem:new sconc_d new}\rm
  Let $\nu,\xi$ be two active synchrons, $\varsigma\in\textit{new}(\nu)$ and $\upsilon\in\textit{new}(\xi)$.
  If $\nu \sconc_d \xi$ then $\varsigma \sconc_d \upsilon$.
\end{lemma}
\begin{proof}
  From $\nu \sconc_d \xi$, we have $\nu=\sigma{|_\D}\nu'$ and $\xi=\sigma{|_\E}\xi'$
    with $\sigma\in\textit{Arg}^*$, $\{\mathrm{D},\mathrm{E}\}=\{\mathrm{L},\mathrm{R}\}$, and $\nu',\xi'$ synchrons.
  By \df{new}, $\varsigma$ has a prefix $\textit{static}(\sigma){|_\D}$ and $\upsilon$ has a prefix $\textit{static}(\sigma){|_\E}$.
  Thus $\varsigma \sconc_d \upsilon$.
\end{proof}

\noindent
The opposite direction of \lem{new sconc_d new} does not hold.

\begin{example}
  Let \plat{$\nu = \xi = (\actsyn{a}(c.{\bf 0} \mid d.{\bf 0}))$}, \plat{$\varsigma = {|_\Left}(\actsyn{c}{\bf 0})$} and \plat{$\upsilon = {|_\Right}(\actsyn{d}{\bf 0})$}.
  Then $\nu,\xi$ are active synchrons, $\varsigma\in\textit{new}(\nu)$, $\upsilon\in\textit{new}(\xi)$ and $\varsigma \sconc_d \upsilon$.
  However, $\nu \nsconc_d \xi$.
  \hfill$\lrcorner$
\end{example}

\noindent
The next result says that the inherited synchrons after a transition are always concurrent with the new ones.

\begin{lemma}\label{lem:new sconc_d @}\rm
  Let $w\in\en(P)$.
  If $\varsigma\in\textit{new}(w)$ and $\upsilon\in\varsigma(P)@w$ then $\varsigma \sconc_d \upsilon$.
\end{lemma}
\begin{proof}
  From $\varsigma\in\textit{new}(w)$ we obtain $\nu\in a\varsigma(w)$ such that $\varsigma\in\textit{new}(\nu)$.
  From $\upsilon\in\varsigma(P)@w$ we obtain $\xi\in\varsigma(P)$ with $\xi \saconc_d w$ such that $\upsilon=\xi@w$.
  From $\xi \saconc_d w$ we have $\xi \sconc_d \nu$.
  Thus $\nu=\sigma{|_\D}\nu'$ and $\xi=\sigma{|_\E}\xi'$
    with $\sigma\in\textit{Arg}^*$, $\{\mathrm{D},\mathrm{E}\}=\{\mathrm{L},\mathrm{R}\}$, and $\nu',\xi'$ synchrons.
  From $\varsigma\in\textit{new}(\nu)$ we have $\varsigma=\textit{static}(\sigma){|_\D}\varsigma'$ for some synchron $\varsigma'$.
  From $\upsilon=\xi@w$ and the fact that $\nu$ shares the prefix $\sigma$ with $\xi$, we have $\upsilon=\textit{static}(\sigma){|_\E}\upsilon'$ for some synchron $\upsilon'$.
  Thus $\varsigma \sconc_d \upsilon$.
\end{proof}

\begin{corollary}\label{cor:@ and new are disjoint}\rm
  Let $w\in\en(P)$.
  $\varsigma(P)@w$ and $\textit{new}(w)$ are disjoint.
\end{corollary}
\begin{proof}
  This follows from \lem{new sconc_d @} and the fact that $\sconc_d$ is irreflexive.
\end{proof}

\noindent
We conclude this section with an algorithm to calculate the set of synchrons of a state $Q$
  reached after a doing a transition $w \mathbin{\in} \en{P}$ as a function of $w$ and the set of synchrons of $P$\!.

\begin{theorem}\label{thm:synchrons of target}\rm
  If $w \mathop{\in} \Tr$, $P \mathbin{=} \source(w)$ and $Q \mathbin{=} \target(w)$
    then $\varsigma(Q) \mathbin{=} \varsigma(P)@w \djcup \textit{new}(w)$.
\end{theorem}
\begin{proof}
  First suppose $\ell(w)\in\B{:}\djcup\bar{\Sig}$.
  Then $a\varsigma(w)=\emptyset$ and $Q=P$.
  We have $\varsigma(P)@w=\varsigma(P)$ and $\textit{new}(w)=\emptyset$.
  Hence $\varsigma(Q) = \varsigma(P) = \varsigma(P)@w \djcup \textit{new}(w)$.

  Now suppose $\ell(w)\in\textit{Act}$.
  We proceed by structural induction on $\varsigma(w)$,
    using the same well-founded order $<$ on sets of synchrons as in the proof of \thm{P-complete}.

  \vspace{1ex}
  \noindent
  \emph{Induction base}:
  Suppose \plat{$\varsigma(w)=\{(\actsyn{\alpha}Q')\}$} for some $\alpha\in\textit{Act}$ and $Q'\in\cT$.
  Then $P=\alpha.Q'$ and $Q=Q'$.
  We have $\varsigma(P)@w=\emptyset$ and $\textit{new}(w)=\varsigma(Q')$.
  Hence $\varsigma(Q) = \varsigma(Q') = \varsigma(P)@w \djcup \textit{new}(w)$.

  \vspace{1ex}
  \noindent
  \emph{Induction step}:
  Given $w \mathop{\in} \Tr$, $P \mathbin{=} \source(w)$ and $Q \mathbin{=} \target(w)$,
    we assume that $\varsigma(Q') \mathbin{=} \varsigma(P')@w' \djcup \textit{new}(w')$ for all
    $w' \mathop{\in} \Tr$, $P' \mathbin{=} \source(w')$ and $Q' \mathbin{=} \target(w')$ such that $\varsigma(w')<\varsigma(w)$.

  \vspace{1ex}
  \noindent
  Suppose $\varsigma(w)={+_\Left}\Sigma_\Left$ for some $\Sigma_\Left$.
  Then $P=P_\Left{+}P_\Right$ and $w=w_\Left{+}P_\Right$ for some $P_\Left,P_\Right\in\cT$ and $w_\Left\in\en(P_\Left)$
    with $\varsigma(w_\Left)=\Sigma_\Left$.
  Thus $Q=Q_\Left$ for some $Q_\Left\in\cT$ with $Q_\Left=\target(w_\Left)$.
  Using that ${+_\Left}\varsigma \,@\, {+_\Left}\upsilon = \varsigma@\upsilon$ whenever $\varsigma \sconc_d \upsilon$,
    we have $\varsigma(P)@w = \varsigma(P_\Left)@w_\Left$ and $\textit{new}(w) = \textit{new}(w_\Left)$.
  By (IH), $\varsigma(Q_\Left) = \varsigma(P_\Left)@w_\Left \djcup \textit{new}(w_\Left)$.
  Hence $\varsigma(Q) = \varsigma(Q_\Left) = \varsigma(P)@w \djcup \textit{new}(w)$.

  \vspace{1ex}
  \noindent
  Suppose $\varsigma(w)={+_\Right}\Sigma_\Right$ for some $\Sigma_\Right$.
  The proof is similar to that of the previous case.

  \vspace{1ex}
  \noindent
  Suppose $\varsigma(w)={|_\Left}\Sigma_\Left$ for some $\Sigma_\Left$.
  Then $P=P_\Left|P_\Right$ and $w=w_\Left|P_\Right$ for some $P_\Left,P_\Right\in\cT$ and $w_\Left\in\en(P_\Left)$
    with $\varsigma(w_\Left)=\Sigma_\Left$.
  Thus $Q=Q_\Left|P_\Right$ for some $Q_\Left\in\cT$ with $Q_\Left=\target(w_\Left)$.
  Using that ${|_\Left}\varsigma \,@\, {|_\Left}\upsilon = {|_\Left}\,\varsigma@\upsilon$ whenever $\varsigma \sconc_d \upsilon$,
    we have $\varsigma(P)@w = {|_\Left}\,\varsigma(P_\Left)@w_\Left \djcup {|_\Right}\varsigma(P_\Right)$ and $\textit{new}(w) = {|_\Left}\textit{new}(w_\Left)$.
  By (IH), $\varsigma(Q_\Left) = \varsigma(P_\Left)@w_\Left \djcup \textit{new}(w_\Left)$.
  Hence $\varsigma(Q) = {|_\Left}\varsigma(Q_\Left) \djcup {|_\Right}\varsigma(P_\Right) = \varsigma(P)@w \djcup \textit{new}(w)$.

  \vspace{1ex}
  \noindent
  Suppose $\varsigma(w)={|_\Right}\Sigma_\Right$ for some $\Sigma_\Right$.
  The proof is similar to that of the previous case.

  \vspace{1ex}
  \noindent
  Suppose $\varsigma(w) = {|_\Left}\Sigma_\Left \djcup {|_\Right}\Sigma_\Right$ for some non-empty $\Sigma_\Left,\Sigma_\Right$.
  Then $P=P_\Left|P_\Right$ and $w=w_\Left|w_\Right$ for some $P_\Left,P_\Right\in\cT$, $w_\Left\in\en(P_\Left)$ and $w_\Right\in\en(P_\Right)$
    with $\varsigma(w_\Left)=\Sigma_\Left$ and $\varsigma(w_\Right)=\Sigma_\Right$.
  Thus $Q=Q_\Left|Q_\Right$ for some $Q_\Left,Q_\Right\in\cT$ with $Q_\Left=\target(w_\Left)$ and $Q_\Right=\target(w_\Right)$.
  Using that ${|_\D}\varsigma \,@\, {|_\D}\upsilon = {|_\D}\,\varsigma@\upsilon$ for $\mathrm{D}\in\{\mathrm{L},\mathrm{R}\}$ whenever $\varsigma \sconc_d \upsilon$,
    we have $\varsigma(P)@w = {|_\Left}\,\varsigma(P_\Left)@w_\Left \djcup {|_\Right}\,\varsigma(P_\Right)@w_\Right$ and
    $\textit{new}(w) = {|_\Left}\textit{new}(w_\Left) \djcup {|_\Right}\textit{new}(w_\Right)$.
  By (IH), $\varsigma(Q_\Left) = \varsigma(P_\Left)@w_\Left \djcup \textit{new}(w_\Left)$ and
    $\varsigma(Q_\Right) = \varsigma(P_\Right)@w_\Right \djcup \textit{new}(w_\Right)$.
  Hence $\varsigma(Q) = {|_\Left}\varsigma(Q_\Left) \djcup {|_\Right}\varsigma(Q_\Right) = \varsigma(P)@w \djcup \textit{new}(w)$.

  \vspace{1ex}
  \noindent
  Suppose $\varsigma(w)={\backslash L}\,\Sigma'$ for some $L\subseteq\Ch\djcup\Sig$ and $\Sigma'$.
  Then $P=P'\backslash L$ and $w=w'\backslash L$ for some $P'\in\cT$ and $w'\in\en(P')$ with $\varsigma(w')=\Sigma'$.
  Thus $Q=Q'\backslash L$ for some $Q'\in\cT$ with $Q'=\target(w')$.
  Using that ${\backslash L}\varsigma \,@\, {\backslash L}\upsilon = {\backslash L}\,\varsigma@\upsilon$ whenever $\varsigma \sconc_d \upsilon$,
    we have $\varsigma(P)@w = {\backslash L}\,\varsigma(P')@w'$ and $\textit{new}(w) = {\backslash L}\,\textit{new}(w')$.
  By (IH), $\varsigma(Q') = \varsigma(P')@w' \djcup \textit{new}(w')$.
  Hence $\varsigma(Q) = {\backslash L}\,\varsigma(Q') = \varsigma(P)@w \djcup \textit{new}(w)$.

  \vspace{1ex}
  \noindent
  Suppose $\varsigma(w)=[f]\Sigma'$ for some relabelling $f$ and $\Sigma'$.
  The proof is similar to that of the previous case.

  \vspace{1ex}
  \noindent
  Suppose $\varsigma(w)=A{:}\Sigma'$ for some $A\in\A$ and $\Sigma'$.
  Then $P=A$, \plat{$A \defis P'$} and $w=A{:}w'$ for some $P'\in\cT$ and $w'\in\en(P')$ with $\varsigma(w')=\Sigma'$.
  Thus $Q=Q'$ for some $Q'\in\cT$ with $Q'=\target(w')$.
  Using that $A{:}\varsigma \,@\, A{:}\upsilon = \varsigma@\upsilon$ whenever $\varsigma \sconc_d \upsilon$,
    we have $\varsigma(P)@w = \varsigma(P')@w'$ and $\textit{new}(w) = \textit{new}(w')$.
  By (IH), $\varsigma(Q') = \varsigma(P')@w' \djcup \textit{new}(w')$.
  Hence $\varsigma(Q) = \varsigma(Q') = \varsigma(P)@w \djcup \textit{new}(w)$.

  \vspace{1ex}
  \noindent
  Suppose $\varsigma(w)={{}\signals s}\,\Sigma'$ for some $r\in\Sig$ and $\Sigma'$.
  The proof is similar to that of the previous case.
\end{proof}

\subsection{What Happens to a Transition after a Concurrent Transition?}

If a process $P$ enables a transition $t$ that is not affected by a transition $w\in\en(P)$,
  then a variant $t'$ of $t$ ought to remain enabled in the state $\target(w)$ reached after executing $w$.
We denote this by $t \ssleadsto_w t'$.
Unlike the situation for synchrons in \Sec{after}, the transition $t'$ is not fully determined by $t$ and $w$.

\begin{example}\label{ex:discard-receive}
  Let $P = a.(b?.P_1 + b?.P_2) \mid b!.{\bf 0}$, \plat{$t = b{:}a.(b?.P_1 + b?.P_2) \mid (\actsyn{b!}{\bf 0})$},
    and let $w$ be the transition $(\actsyn{a}(b?.P_1 + b?.P_2)) \mid b!.{\bf 0}$.
  Since the $b{:}$-synchron of $t$ is not necessary, $t$ is not affected by $w$, \ie $t \ssaconc w$.
  So there must be a transition $t'$ with $t \ssleadsto_w t'$.
  However, the process $\target(w)$ enables exactly two $b!$-transitions, and they are equally plausible candidates for $t'$.

  The same example applies when leaving out the parallel component $b!.{\bf 0}$.
  In that case $\ell(t)=b{:}$ and $\ell(t')=b?$\,.
  \hfill$\lrcorner$
\end{example}

\begin{definition}[Possible Successor Transition]\label{df:ssleadsto}\rm
  A {\transition} $t'$ is a \emph{possible successor} of a {\transition} $t$ after a {\transition} $w$,
    notation $t \ssleadsto_w t'$, iff
  \begin{enumerate}
    \item $\source(t)=\source(w)$,
    \item $\source(t')=\target(w)$,
    \item $t \ssaconc w$\label{ssaconc},
    \item for all $\varsigma\in\varsigma(t)$ with $\varsigma \saconc_d w$, one has $\varsigma@w\in\varsigma(t')$,
      \label{fourth} and
    \item $\ell(t)=\ell(t') \lor \exists\, b\in\B.~ \{\ell(t),\ell(t')\}=\{b?,b{:}\}$.\label{labels}%
      \footnote{This allows a broadcast-receive {\transition} to be a possible successor of a broadcast-discard {\transition},
        and vice versa, as required for \ex{discard-receive}b.}
    \vspace{2ex}
  \end{enumerate}
\end{definition}
The first two conditions are obvious, and the third states that a successor $t'$ of $t$ after $w$ exists
  only if $t$ is unaffected by $w$.
The fourth condition requires that all synchrons $\varsigma$ of $t$ that are unaffected by $w$,
  return in the form $\varsigma@w$ as a synchron of $t'$.
The last condition states that the label of $t'$ should be the same as that of $t$,
  with the exception of Footnote~\thefootnote.
We would have imposed one more condition, but by the following lemma it is implied.

\begin{lemma}\label{lem:label condition}\rm
  If $t \ssleadsto_w t'$ then $|n\varsigma(t)| = |n\varsigma(t')|$ and $@w$ is a bijection between the necessary synchrons of $t$ and $t'$.
\end{lemma}
\begin{proof}
  By \lem{injective} $|n\varsigma(t)| \leq |n\varsigma(t')|$ and $@w$ is an injection from $n\varsigma(t)$ to $n\varsigma(t')$.
  It remains to show that $|n\varsigma(t)| = |n\varsigma(t')|$.
  Note that $\ell(\varsigma@w)=\ell(\varsigma)$ for all synchrons $\varsigma$ with $\varsigma \saconc_d w$,
    and thus for all $\varsigma\in n\varsigma(t)$.

  By \lem{synchrons of transition}, either (i) $n\varsigma(t')=\emptyset$ and $\ell(t')\in\B?\djcup\B{:}$,
    or (ii) $n\varsigma(t')=\{\upsilon\}$ for a synchron $\upsilon$ with $\ell(\upsilon)=\ell(t')\notin\B?\djcup\B{:}$,
    or (iii) $|n\varsigma(t')| = 2$ and $\ell(t')=\tau$, yet $\ell(\upsilon)\neq\tau$ for all $\upsilon\in n\varsigma(t')$.
  A similar classification applies to $n\varsigma(t)$.

  In Case (i) $|n\varsigma(t)| = |n\varsigma(t')|$ follows from $|n\varsigma(t)| \leq |n\varsigma(t')|$, and we are done.

  In Case (ii) $|n\varsigma(t)| < |n\varsigma(t')|$ would imply that $n\varsigma(t)=\emptyset$ and $\ell(t)\in\B?\djcup\B{:}$,
    contradicting Requirement~\ref{labels} of \df{ssleadsto}.

  In Case (iii) $|n\varsigma(t)| < |n\varsigma(t')|$ would imply that either $n\varsigma(t)=\emptyset$,
    leading to a contradiction just as above, or $n\varsigma(t)=\{\varsigma\}$ for some synchron $\varsigma$.
  In the latter case $\varsigma@w\in n\varsigma(t')$, so $\ell(t) = \ell(\varsigma) = \ell(\varsigma@w) \neq \tau = \ell(t')$,
    again a contradiction.
\end{proof}

\noindent
The main result of this appendix, delivered in \Sec{coincide}, says that the relation $\ssleadsto$ defined above
  coincides with the relation $\leadsto$ of \df{leadsto}.
This justifies the clauses of \df{leadsto}.

The concurrency relation $\aconc$ in \Sec{ltss} is defined as the projection of $\leadsto$
  on its first two arguments.
The next theorem says that likewise $\ssaconc$ is the projection of $\ssleadsto$ on its first two arguments.
As a result, ${\ssleadsto}={\leadsto}$ will imply ${\ssaconc}={\aconc}$.

\begin{theorem}\label{thm:relationship between ssaconc ssleadsto}\rm
  $t \ssaconc w$ iff $\exists\, t'.~ t \ssleadsto_w t'$.
\end{theorem}
\begin{proof}
  By Condition 3 of \df{ssleadsto}, $\exists\, t'.~ t \ssleadsto_w t' \implies t \ssaconc w$ trivially holds.
  It remains to prove $t \ssaconc w \implies \exists\, t'.~ t \ssleadsto_w t'$.
  So assume $t \ssaconc w$ and let $P = \source(t) = \source(w)$.
  By \thm{P-complete}, the set $\varsigma(t)$ of synchrons is $P$-complete,
    so one of the three possibilities listed in \df{P-complete} applies.

  In the first case, $\varsigma(t) = n\varsigma(t) = \{\varsigma\}$ with
    $\ell(\varsigma) \in \Ch \djcup \bar{\Ch} \djcup \{\tau\} \djcup \Sig \djcup \bar{\Sig}$, and $\varsigma \saconc_d w$.
  By \thm{synchrons of target}, $\varsigma@w\in\varsigma(\target(w))$.
  Now $\Sigma \coloneqq \{\varsigma@w\}$ is $P$-complete by \df{P-complete}, using that $\ell(\varsigma@w)=\ell(\varsigma)$.
  So by \thm{P-complete} there is a transition $t'$ with $\source(t')=\target(w)$ and $\varsigma(t')=\Sigma$.
  Now $t \ssleadsto_w t'$, as all conditions of \df{ssleadsto} are trivially satisfied.
  In particular, $\ell(t) = \ell(\varsigma) = \ell(\varsigma@w) = \ell(t')$, by \lem{synchrons of transition}.

  In the second case, $\varsigma(t) = n\varsigma(t) = \{\varsigma,\upsilon\}$ with
    $\ell(\varsigma) \in \Ch \djcup \bar{\Ch} \djcup \Sig \djcup \bar{\Sig}$, $\ell(\upsilon)=\overline{\ell(\varsigma)}$, and
    $\varsigma \sconc_d \upsilon$.
  Moreover, $\varsigma \saconc_d w$ and $\upsilon \saconc_d w$.
  By \thm{synchrons of target}, $\varsigma@w,\upsilon@w \in \varsigma(\target(w))$.
  From $\varsigma \sconc_d \upsilon$ we have $\varsigma@w \sconc_d \upsilon@w$ by \lem{@ sconc_d @}.
  So $\Sigma \coloneqq \{\varsigma@w,\upsilon@w\}$ is $P$-complete by \df{P-complete}.
  By \thm{P-complete} there is a transition $t'$ with $\source(t')=\target(w)$ and $\varsigma(t')=\Sigma$.
  Now $t \ssleadsto_w t'$, as all conditions of \df{ssleadsto} are trivially satisfied.
  In particular, $\ell(t) = \tau = \ell(t')$, by \lem{synchrons of transition}.

  Finally suppose the third case applies, for some $b \mathop{\in} \B$.
  Let $\Sigma=\varsigma(t)@w$.
  Then $\Sigma \subseteq \varsigma(\target(w))$ by \thm{synchrons of target}.
  Note that $\Sigma$ satisfies the first two requirements of \df{P-complete}.\ref{third}, since $\varsigma(t)$ satisfies them.
  By \lem{@ sconc_d @} the third requirement also holds for $\Sigma$.

  Now we extend $\Sigma$ to $\Sigma'$ to by adding one synchron at a time.
  Namely, if there exists any $\varsigma\in\varsigma(\target(w))$ with $\ell(\varsigma)=b?$ and $\varsigma \sconc_d \upsilon$
    for each active synchron $\upsilon$ currently in $\Sigma$, then we add $\varsigma$ to $\Sigma$.
  The set $\Sigma'$ has been obtained when no further such synchrons $\varsigma$ can be found.
  By construction, $\Sigma'$ satisfies the first four requirements of \df{P-complete}.\ref{third}.

  In particular, a newly added $b?$-synchron will not affect any $b{:}$-synchron already existing in $\Sigma$.
  Note that according to \thm{synchrons of target}, this $b?$-synchron must be
    either (a) a member of $\textit{new}(w)$ or (b) a member of $(\varsigma(P){\setminus}\varsigma(t))@w$.
  In Case (a), it will be concurrent with all the existing $b{:}$-synchrons in $\Sigma$ thanks to \lem{new sconc_d @}.
  In Case (b), it will be $\varsigma@w$ for some $\varsigma \in \varsigma(P){\setminus}\varsigma(t)$ with $\varsigma \saconc_d w$.
  Since $\varsigma\notin\varsigma(t)$, by \df{P-complete} and \thm{P-complete} we can obtain $\nu\in a\varsigma(t)$ such that $\varsigma \nsconc_d \nu$.
  If $\nu \saconc_d w$, then $\nu@w\in\Sigma$ and, by \lem{@ sconc_d @}, $\varsigma@w \nsconc_d \nu@w$,
    contradicting the requirement that $\varsigma@w$ is not affected by any active synchron currently in $\Sigma$.
  So $\nu \nsaconc_d w$. We obtain $\xi\in a\varsigma(w)$ with $\nu \nsconc_d \xi$.
  \begin{enumerate}
    \item If the common prefix shared by $\varsigma$ and $\nu$ is shorter than that shared by $\nu$ and $\xi$,
          we have $\varsigma=\sigma{+_\D}\varsigma'$ $\nu=\sigma{+_\E}\sigma'{+_{\D'}}\nu'$ and $\xi=\sigma{+_\E}\sigma'{+_{\E'}}\xi'$
          with $\sigma,\sigma'\in\textit{Arg}^*$, $\{\mathrm{D},\mathrm{E}\} = \{\mathrm{D}',\mathrm{E}'\} = \{\mathrm{L},\mathrm{R}\}$,
          and $\varsigma',\nu',\xi'$ synchrons.
      Then $\varsigma \nsconc_d \xi$, contradicting $\varsigma \saconc_d w$.
    \item If the common prefix shared by $\varsigma$ and $\nu$ is the same as that shared by $\nu$ and $\xi$,
        we have $\varsigma=\sigma{+_\D}\varsigma'$ $\nu=\sigma{+_\E}\nu'$ and $\xi=\sigma{+_\D}\xi'$
        with $\sigma\in\textit{Arg}^*$, $\{\mathrm{D},\mathrm{E}\}=\{\mathrm{L},\mathrm{R}\}$, and $\varsigma',\nu',\xi'$ synchrons.
      A $b{:}$-synchrons in $\Sigma$ must be $\upsilon@w$ for some $\upsilon\in\varsigma(t)$ with $\ell(\upsilon)=b{:}$ and $\upsilon \saconc_d w$.
      From $\upsilon \saconc_d w$ we have $\upsilon \sconc_d \xi$; by \df{P-complete} and \thm{P-complete}, $\upsilon \sconc_d \nu$.
      Thus $\upsilon$ has a prefix $\sigma'{|_{\D'}}$ and $\nu$ has a prefix $\sigma'{|_{\E'}}$
        with $\{\mathrm{D}',\mathrm{E}'\}=\{\mathrm{L},\mathrm{R}\}$.
      Moreover, $\sigma'$ must be a proper prefix of $\sigma$, as otherwise $\upsilon \nsconc_d \xi$.
      Hence $\varsigma \sconc_d \upsilon$.
      By \lem{@ sconc_d @} $\varsigma@w \sconc_d \upsilon@w$.
    \item If the common prefix shared by $\varsigma$ and $\nu$ is longer than that shared by $\nu$ and $\xi$,
          we have $\varsigma=\sigma{+_\D}\sigma'{+_{\D'}}\varsigma'$ $\nu=\sigma{+_\D}\sigma'{+_{\E'}}\nu'$ and $\xi=\sigma{+_\E}\xi'$
          with $\sigma,\sigma'\in\textit{Arg}^*$, $\{\mathrm{D},\mathrm{E}\} = \{\mathrm{D}',\mathrm{E}'\} = \{\mathrm{L},\mathrm{R}\}$,
          and $\varsigma',\nu',\xi'$ synchrons.
      Then $\varsigma \nsconc_d \xi$, contradicting $\varsigma \saconc_d w$.
  \end{enumerate}

  Next we extend $\Sigma'$ to $\Sigma''$ by adding one synchron at a time.
  Namely, if there exists any $\varsigma\in\varsigma(\target(w))$ with $\ell(\varsigma)=b{:}$ and $\varsigma \sconc_d \upsilon$
    for each active synchron $\upsilon$ currently in $\Sigma'$, then we add $\varsigma$ to $\Sigma'$.
  The set $\Sigma''$ has been obtained when no further such synchrons $\varsigma$ can be found.
  By construction, $\Sigma''$ satisfies all requirements of \df{P-complete}.\ref{third}, and thus is $P$-complete.

  By \thm{P-complete} there is a transition $t'$ with $\source(t')=\target(w)$ and $\varsigma(t')=\Sigma''$.
  Now $t \ssleadsto_w t'$, as all conditions of \df{ssleadsto} are satisfied.
  Regarding the last condition, if $n\varsigma(t)\neq\emptyset$ then $n\varsigma(t')\neq\emptyset$ and $\ell(t) = b! = \ell(t')$.
  If $n\varsigma(t)=\emptyset$ then $n\varsigma(t')=\emptyset$ and $\ell(t),\ell(t') \in \{b?,b{:}\}$, by \lem{synchrons of transition}.
  Either way, the last condition holds.
\end{proof}

\subsection*{Digression: A Successor Relation between Synchrons}

In this section we provide an alternative form of Condition~\ref{fourth} of \df{ssleadsto}, which does not use the function $@$.
We use it to relate our $\ssleadsto$-relation with the $\leadsto$-relation from \cite{synchrons}.

\begin{definition}[Possible Successor Synchron~\cite{synchrons}]\label{df:sleadsto}\rm
  A synchron $\varsigma'$ is a \emph{possible successor} of a synchron $\varsigma$,
  notation $\varsigma \sleadsto \varsigma'$,
  if either $\varsigma'=\varsigma$ or
  $\varsigma=\sigma{|_\D}\varsigma''$ with $\sigma\in\textit{Arg}^*$, $\mathrm{D}\in\{\mathrm{L},\mathrm{R}\}$, and $\varsigma''$ a synchron,
  and $\varsigma'=\textit{static}(\sigma){|_\D}\varsigma''$.
\end{definition}

\begin{lemma}\label{lem:@}\rm
  Let $w\in\Tr$, $P=\source(w)$, $Q=\target(w)$, $\varsigma\in\varsigma(P)$ and $\varsigma'\in\varsigma(Q)$,
    such that $\varsigma \saconc_d w$.
  Then $\varsigma \sleadsto \varsigma'$ iff $\varsigma'=\varsigma@w$.
\end{lemma}
\begin{proof}
  By \df{@} $\varsigma \sleadsto \varsigma@w$, so `only if' is trivial.

  For `if', suppose $\varsigma \sleadsto \varsigma'$.
  Then $\varsigma'$ and $\varsigma@w$ are both obtained from $\varsigma$ by deleting a
  prefix of their dynamic arguments.
  So $\varsigma' = \sigma_1\sigma'_2\varsigma_3$ and $\varsigma@w = \sigma_1\sigma_2\varsigma_3$
    with either $\sigma'_2 = \textit{static}(\sigma_2)$ or $\sigma_2 = \textit{static}(\sigma'_2)$.
  Moreover, if $\sigma_2$ is non-empty, then $\varsigma_3$ must start with an argument ${|_\Left}$ or ${|_\Right}$.
  Furthermore, since $\varsigma@w \in \varsigma(Q)$ by \thm{synchrons of target},
    both $\varsigma'$ and $\varsigma@w$ describe paths in the (unfolded) parse tree of $Q$.
  So the first argument on which they differ, if any, must be $\textit{op}_\Left$ versus $\textit{op}_\Right$
    for the same binary \ABCdE operator $\textit{op}\in\{+,|\}$.
  Together, these conditions imply that $\varsigma'=\varsigma@w$.
\end{proof} 

\noindent
Consequently, Condition~\ref{fourth} of \df{ssleadsto} can be restated as
  \begin{enumerate}
    \item[\ref{fourth}.] for all $\varsigma\in\varsigma(t)$ with $\varsigma \saconc_d w$,
      there exists a $\varsigma'\in\varsigma(t')$ with $\varsigma \sleadsto \varsigma'$.
    \vspace{2ex}
  \end{enumerate}

\noindent
Recall that $\Tr^{s\bullet} \subseteq \Tr$ is the set of transitions $t$ with $\ell(t)\notin\B?\djcup\B{:}$\,.
In \cite{synchrons} a binary relation ${\leadsto} \subseteq \Tr^{s\bullet} \times \Tr^{s\bullet}$ was defined by
  $t \leadsto t'$ iff $|n\varsigma(t)| = |n\varsigma(t')|$ and
  $\forall \varsigma' \mathbin{\in} n\varsigma(t').~ \exists \varsigma \mathbin{\in} n\varsigma(t).~ \varsigma \sleadsto \varsigma'$.
By Lemmas~\ref{lem:label condition} and~\ref{lem:@}
  it follows that $t \ssleadsto_w t'$ for $t,t' \in \Tr^{s\bullet}$ and $w\in\Tr$ implies $t \leadsto t'$.

\subsection{The two Successor Relations between Transitions Coincide}\label{sec:coincide}

Theorems~\ref{thm:leadsto implies ssleadsto} and~\ref{thm:ssleadsto implies leadsto} below
  yield the main result of this appendix, that ${\leadsto}$ equals ${\ssleadsto}$.

\begin{theorem}\label{thm:leadsto implies ssleadsto}\rm
  ${\leadsto} \subseteq {\ssleadsto}$.
  I.e., if $\chi \leadsto_\zeta \chi'$ then $\chi \ssleadsto_\zeta \chi'$.
\end{theorem}
\begin{proof}
  Since the $\leadsto$ relation is defined inductively,
    $\chi \leadsto_\zeta \chi'$ must stem from some rules in \df{leadsto},
    which means at least one statement below holds
    for some $t,t',u,u',v,w\in\Tr$, $P,Q\in\cT$, $\alpha,L,f,A,b,r$ with $\source(t)=\source(v)=P$,
    $\source(u)=\source(w)=Q$, $\source(t')=\target(v)$, $\source(u')=\target(w)$,
    $\alpha\in\textit{Act}$, $L\subseteq\Ch\djcup\Sig$, $f$ a relabelling, $A\in\A$, $b\in\B$ and $r\in\Sig$.
  \setlength\leftmargini{2.5em}
  \begin{enumerate}
    \item [1.a] $\ell(\zeta)\in\B{:}\djcup\bar{\Sig}$, $\source(\zeta)=\source(\chi)$ and $\chi'=\chi$;
    \item [2.a] \plat{$\chi=\actsyn{b?}P$}, \plat{$\zeta=\actsyn{b?}P$} and $\chi'=t$ with $\ell(t)\in\{b?,b{:}\}$;
    \item [2.b] $\chi=b{:}\alpha.P$, \plat{$\zeta=\actsyn{\alpha}P$} and $\chi'=t$ with $\ell(t)\in\{b?,b{:}\}$;
    \item [3.a] $\chi=t+Q$, $\zeta=v+Q$ and $\chi'=t'$ with $\ell(v)\notin\bar{\Sig}$ and $t \leadsto_v t'$;
    \item [3.b] $\chi=t+u$, $\zeta=v+Q$ and $\chi'=t'$ with $\ell(v)\notin\bar{\Sig}$ and $t \leadsto_v t'$;
    \item [4.a] $\chi=P+u$, $\zeta=P+w$ and $\chi'=u'$ with $\ell(w)\notin\bar{\Sig}$ and $u \leadsto_w u'$;
    \item [4.b] $\chi=t+u$, $\zeta=P+w$ and $\chi'=u'$ with $\ell(w)\notin\bar{\Sig}$ and $u \leadsto_w u'$;
    \item [5.a] $\chi=t+Q$, $\zeta=P+w$ and $\chi'=u'$ with $\ell(w)\notin\bar{\Sig}$, $\ell(t)=b?$ and $\ell(u')\in\{b?,b{:}\}$;
    \item [6.a] $\chi=P+u$, $\zeta=v+Q$ and $\chi'=t'$ with $\ell(v)\notin\bar{\Sig}$, $\ell(u)=b?$ and $\ell(t')\in\{b?,b{:}\}$;
    \item [7.a] $\chi=t|Q$, $\zeta=P|w$ and $\chi'=t|\target(w)$;
    \item [7.b] $\chi=P|u$, $\zeta=v|Q$ and $\chi'=\target(v)|u$;
    \item [8.a] $\chi=t|Q$, $\zeta=v|Q$ and $\chi'=t'|Q$ with $t \leadsto_v t'$;
    \item [8.b] $\chi=t|Q$, $\zeta=v|w$ and $\chi'=t'|\target(w)$ with $t \leadsto_v t'$;
    \item [8.c] $\chi=t|u$, $\zeta=v|Q$ and $\chi'=t'|u$ with $t \leadsto_v t'$;
    \item [9.a] $\chi=P|u$, $\zeta=P|w$ and $\chi'=P|u'$ with $u \leadsto_w u'$;
    \item [9.b] $\chi=P|u$, $\zeta=v|w$ and $\chi'=\target(v)|u'$ with $u \leadsto_w u'$;
    \item [9.c] $\chi=t|u$, $\zeta=P|w$ and $\chi'=t|u'$ with $u \leadsto_w u'$;
    \item [10.a] $\chi=t|u$, $\zeta=v|w$ and $\chi'=t'|u'$ with $t \leadsto_v t'$ and $u \leadsto_w u'$;
    \item [11.a] $\chi=t\backslash L$, $\zeta=v\backslash L$ and $\chi'=t'\backslash L$ with $\ell(v)\in\textit{Act}$ and $t \leadsto_v t'$;
    \item [11.b] $\chi=t[f]$, $\zeta=v[f]$ and $\chi'=t'[f]$ with $\ell(v)\in\textit{Act}$ and $t \leadsto_v t'$;
    \item [11.c] $\chi=A{:}t$, $\zeta=A{:}v$ and $\chi'=t'$ with $\ell(v)\in\textit{Act}$ and $t \leadsto_v t'$;
    \item [11.d] $\chi=t\signals r$, $\zeta=v\signals r$ and $\chi'=t'$ with $\ell(v)\in\textit{Act}$ and $t \leadsto_v t'$.
  \end{enumerate}
  The numbers of the above statements correspond to the numbers of the rules in \df{leadsto}.
  For example, Statement 2.a or 2.b holds if $\chi \leadsto_\zeta \chi'$ stems from Rule~\ref{discard base}.

  \vspace{1ex}
  \noindent
  It suffices to prove that each statement above implies $\chi \ssleadsto_\zeta \chi'$
    under the assumptions $t \leadsto_v t' \implies t \ssleadsto_v t'$ and $u \leadsto_w u' \implies u \ssleadsto_w u'$.
  This gives us a structural induction proof of \thm{leadsto implies ssleadsto}.
  Namely, cases for Statements 1.a, 2.a, 2.b, 5.a, 6.a, 7.a and 7.b -- where the two assumptions will not be used -- form the induction base,
    while the other cases represent the induction step with the two assumptions being the induction hypothesis.

  \vspace{1ex}
  \noindent
  By \df{ssleadsto}, in each case we need to prove
  \begin{enumerate}
    \item [(a)] $\source(\chi)=\source(\zeta)$,
    \item [(b)] $\source(\chi')=\target(\zeta)$,
    \item [(c)] $\chi \ssaconc \zeta$,
    \item [(d)] for all $\varsigma\in\varsigma(\chi)$ with $\varsigma \saconc_d \zeta$,
      one has $\varsigma@\zeta\in\varsigma(\chi')$, and
    \item [(e)] $\ell(\chi)=\ell(\chi') \lor \exists\, b\in\B.~ \{\ell(\chi),\ell(\chi')\}=\{b?,b{:}\}$.
  \end{enumerate}

  \vspace{1ex}
  \noindent
  Below we use the number of statements for cases.

  \vspace{1ex}
  \noindent
  \emph{Case 1.a}:
  \begin{itemize}
    \item (a) is given in the statement.
    \item From $\ell(\zeta)\in\B{:}\djcup\bar{\Sig}$ we have $\target(\zeta)=\source(\zeta)$.
      With $\source(\zeta)=\source(\chi)$ and $\chi'=\chi$ we have (b).
    \item From $\ell(\zeta)\in\B{:}\djcup\bar{\Sig}$ we have $\alpha\varsigma(\zeta)=\emptyset$.
      Thus (c) trivially holds.
    \item Suppose $\varsigma\in\varsigma(\chi)$.
      From $\ell(\zeta)\in\B{:}\djcup\bar{\Sig}$ we have $\varsigma@\zeta=\varsigma$.
      Using $\chi'=\chi$ we have (d).
    \item From $\chi'=\chi$ we have $\ell(\chi)=\ell(\chi')$.
      Thus (e) holds.
  \end{itemize}

  \vspace{1ex}
  \noindent
  \emph{Case 2.a}:
  \begin{itemize}
    \item $\source(\chi)=b?.P=\source(\zeta)$.
    \item $\source(\chi')=P=\target(\zeta)$.
    \item We have $\ell(\chi)=b?$\,.
      Then $n\varsigma(\chi)=\emptyset$.
      Thus (c) trivially holds.\vspace{1pt}
    \item We have \plat{$(\actsyn{b?}P) \in a\varsigma(\zeta)$}
        and \plat{$\forall \varsigma\in\varsigma(\chi).~ \varsigma \nsconc_d (\actsyn{b?}P)$}.
      Thus (d) trivially holds.
    \item We have $\ell(\chi)=b?$ and $\ell(\chi')\in\{b?,b{:}\}$.
      Thus (e) holds.
  \end{itemize}

  \vspace{1ex}
  \noindent
  \emph{Case 2.b}:
  \begin{itemize}
    \item $\source(\chi)=\alpha.P=\source(\zeta)$.
    \item $\source(\chi')=P=\target(\zeta)$.
    \item We have $\ell(\chi)=b{:}$\,.
      Then $n\varsigma(\chi)=\emptyset$.
      Thus (c) trivially holds.
    \item We have \plat{$(\actsyn{\alpha}P) \in a\varsigma(\zeta)$}
        and \plat{$\forall \varsigma\in\varsigma(\chi).~ \varsigma \nsconc_d (\actsyn{\alpha}P)$}.
      Thus (d) trivially holds.
    \item We have $\ell(\chi)=b{:}$ and $\ell(\chi')\in\{b?,b{:}\}$.
      Thus (e) holds.
  \end{itemize}

  \vspace{1ex}
  \noindent
  \emph{Case 3.a}:
  \begin{itemize}
    \item $\source(\chi)=P+Q=\source(\zeta)$.
    \item $\source(\chi')=\target(v)=\target(\zeta)$.
    \item Suppose $\varsigma\in n\varsigma(\chi)$ and $\upsilon\in a\varsigma(\zeta)$.
      Then $\varsigma={+_\Left}\varsigma_t$ and $\upsilon={+_\Left}\upsilon_v$
        for some $\varsigma_t\in n\varsigma(t)$ and $\upsilon_v\in a\varsigma(v)$.
      From $t \ssleadsto_v t'$ we have $\varsigma_t \sconc_d \upsilon_v$.
      Then $\varsigma \sconc_d \upsilon$.
      Thus (c) holds.
    \item Suppose $\varsigma\in\varsigma(\chi)$ and $\varsigma \saconc_d \zeta$.
      Then $\varsigma={+_\Left}\varsigma_t$ for some $\varsigma_t\in\varsigma(t)$ with $\varsigma_t \saconc_d v$.
      From $t \ssleadsto_v t'$ we have $\varsigma_t@v\in\varsigma(t')$.
      Using that ${+_\Left}\varsigma' \,@\, {+_\Left}\upsilon' = \varsigma'@\upsilon'$ whenever $\varsigma' \sconc_d \upsilon'$,
        we have $\varsigma@\zeta=\varsigma_t@v$.
      Thus (d) holds.
    \item From $t \ssleadsto_v t'$ we have $\ell(t)=\ell(t') \lor \exists\, b\in\B.~ \{\ell(t),\ell(t')\}=\{b?,b{:}\}$.
      Also, we have $\ell(\chi)=\ell(t)$ and $\ell(\chi')=\ell(t')$.
      Thus (e) holds.
  \end{itemize}

  \vspace{1ex}
  \noindent
  \emph{Case 3.b}:
  \begin{itemize}
    \item $\source(\chi)=P+Q=\source(\zeta)$.
    \item $\source(\chi')=\target(v)=\target(\zeta)$.
    \item We have $\ell(\chi)\in\B{:}$\,.
      Then $n\varsigma(\chi)=\emptyset$.
      Thus (c) trivially holds.
    \item Suppose $\varsigma\in\varsigma(\chi)$ and $\varsigma \saconc_d \zeta$.
      From $\zeta=v+Q$ and $\ell(v)\notin\bar{\Sig}$ we have $\ell(v)\in\textit{Act}$.
      Thus $a\varsigma(\zeta)\neq\emptyset$.
      We have that $\varsigma$ cannot be ${+_\Right}$-prefixed
        or it will be affected by some ${+_\Left}$-prefixed synchrons in $a\varsigma(\zeta)$.
      Thus $\varsigma={+_\Left}\varsigma_t$ for some $\varsigma_t\in\varsigma(t)$ with $\varsigma_t \saconc_d v$.
      From $t \ssleadsto_v t'$ we have $\varsigma_t@v\in\varsigma(t')$.
      Using that ${+_\Left}\varsigma' \,@\, {+_\Left}\upsilon' = \varsigma'@\upsilon'$ whenever $\varsigma' \sconc_d \upsilon'$,
        we have $\varsigma@\zeta=\varsigma_t@v$.
      Thus (d) holds.
    \item From $t \ssleadsto_v t'$ we have $\ell(t)=\ell(t') \lor \exists\, b\in\B.~ \{\ell(t),\ell(t')\}=\{b?,b{:}\}$.
      Also, we have $\ell(\chi)=\ell(t)$ and $\ell(\chi')=\ell(t')$.
      Thus (e) holds.
  \end{itemize}

  \vspace{1ex}
  \noindent
  \emph{Case 4.a}:
  The proof is similar to that for Case 3.a and is omitted.

  \vspace{1ex}
  \noindent
  \emph{Case 4.b}:
  The proof is similar to that for Case 3.b and is omitted.

  \vspace{1ex}
  \noindent
  \emph{Case 5.a}:
  \begin{itemize}
    \item $\source(\chi)=P+Q=\source(\zeta)$.
    \item $\source(\chi')=\target(w)=\target(\zeta)$.
    \item We have $\ell(\chi)=b?$\,.
      Then $n\varsigma(\chi)=\emptyset$.
      Thus (c) trivially holds.
    \item Suppose $\varsigma\in\varsigma(\chi)$.
      Then $\varsigma$ is ${+_\Left}$-prefixed.
      From $\zeta=P+w$ and $\ell(w)\notin\bar{\Sig}$ we have $\ell(w)\in\textit{Act}$.
      Thus $a\varsigma(\zeta)\neq\emptyset$.
      We have that $\varsigma$ will be affected by some ${+_\Right}$-prefixed synchrons in $a\varsigma(\zeta)$.
      Thus (d) trivially holds.
    \item We have $\ell(\chi)=b?$ and $\ell(\chi')\in\{b?,b{:}\}$.
      Thus (e) holds.
  \end{itemize}

  \vspace{1ex}
  \noindent
  \emph{Case 6.a}:
  The proof is similar to that for Case 5.a and is omitted.

  \vspace{1ex}
  \noindent
  \emph{Case 7.a}:
  \begin{itemize}
    \item $\source(\chi)=P|Q=\source(\zeta)$.
    \item $\source(\chi')=P|\target(w)=\target(\zeta)$.
    \item We have that all synchrons in $n\varsigma(\chi)$ are ${|_\Left}$-prefixed and
        all synchrons in $a\varsigma(\zeta)$, if any, are ${|_\Right}$-prefixed.
      Thus (c) holds.
    \item Suppose $\varsigma\in\varsigma(\chi)$ and $\varsigma \saconc_d \zeta$.
      Then $\varsigma={|_\Left}\varsigma_t$ for some $\varsigma_t\in\varsigma(t)$.
      If $\ell(\zeta)\in\bar{\Sig}$, then $\varsigma@\zeta = \varsigma = {|_\Left}\varsigma_t \in \varsigma(t|\target(w))$.
      Otherwise, we have that all synchrons in $a\varsigma(\zeta)$ are ${|_\Right}$-prefixed.
      Using that ${|_\Left}\varsigma' \,@\, {|_\Right}\upsilon' = {|_\Left}\varsigma'$,
        again we have $\varsigma@\zeta = \varsigma = {|_\Left}\varsigma_t \in \varsigma(t|\target(w))$.
      Thus (d) holds.
    \item We have $\ell(\chi)=\ell(t)=\ell(\chi')$.
      Thus (e) holds.
  \end{itemize}

  \vspace{1ex}
  \noindent
  \emph{Case 7.b}:
  The proof is similar to that for Case 7.a and is omitted.

  \vspace{1ex}
  \noindent
  \emph{Case 8.a}:
  \begin{itemize}
    \item $\source(\chi)=P|Q=\source(\zeta)$.
    \item $\source(\chi')=\target(v)|Q=\target(\zeta)$.
    \item Suppose $\varsigma\in n\varsigma(\chi)$ and $\upsilon\in a\varsigma(\zeta)$.
      Then $\varsigma={|_\Left}\varsigma_t$ and $\upsilon={|_\Left}\upsilon_v$
        for some $\varsigma_t\in n\varsigma(t)$ and $\upsilon_v\in a\varsigma(v)$.
      From $t \ssleadsto_v t'$ we have $\varsigma_t \sconc_d \upsilon_v$.
      Then $\varsigma \sconc_d \upsilon$.
      Thus (c) holds.
    \item Suppose $\varsigma\in\varsigma(\chi)$ and $\varsigma \saconc_d \zeta$.
      Then $\varsigma={|_\Left}\varsigma_t$ for some $\varsigma_t\in\varsigma(t)$ with $\varsigma_t \saconc_d v$.

      If $\ell(\zeta) = \ell(v) \in \B{:}\djcup\bar{\Sig}$, then $\varsigma@\zeta = \varsigma = {|_\Left}\varsigma_t \in \varsigma(t|Q)$.
      Also, from $t \leadsto_v t'$ we have $t'=t$.
      Thus (d) holds.

      Otherwise, using that ${|_\Left}\varsigma' \,@\, {|_\Left}\upsilon' = {|_\Left}\,\varsigma'@\upsilon'$ whenever $\varsigma' \sconc_d \upsilon'$,
        we have $\varsigma@\zeta={|_\Left}\,\varsigma_t@v$.
      From $t \ssleadsto_v t'$ we have $\varsigma_t@v\in\varsigma(t')$.
      Thus (d) holds.
    \item From $t \ssleadsto_v t'$ we have $\ell(t)=\ell(t') \lor \exists\, b\in\B.~ \{\ell(t),\ell(t')\}=\{b?,b{:}\}$.
      Also, $\ell(\chi)=\ell(t)$ and $\ell(\chi')=\ell(t')$.
      Thus (e) holds.
  \end{itemize}

  \vspace{1ex}
  \noindent
  \emph{Case 8.b}:
  \begin{itemize}
    \item $\source(\chi)=P|Q=\source(\zeta)$.
    \item $\source(\chi')=\target(v)|\target(w)=\target(\zeta)$.
    \item Suppose $\varsigma\in n\varsigma(\chi)$ and $\upsilon\in a\varsigma(\zeta)$.
      Then $\varsigma={|_\Left}\varsigma_t$ for some $\varsigma_t\in n\varsigma(t)$ and
        either $\upsilon={|_\Left}\upsilon_v$ for some $\upsilon_v\in a\varsigma(v)$
        or $\upsilon={|_\Right}\upsilon_w$ for some $\upsilon_w\in a\varsigma(w)$.
      In the first case we have $\varsigma_t \sconc_d \upsilon_v$ from $t \ssleadsto_v t'$.
      Then $\varsigma \sconc_d \upsilon$.
      In the second case $\varsigma \sconc_d \upsilon$ holds by definition.
      Thus (c) holds.
    \item Suppose $\varsigma\in\varsigma(\chi)$ and $\varsigma \saconc_d \zeta$.
      Then $\varsigma={|_\Left}\varsigma_t$ for some $\varsigma_t\in\varsigma(t)$ with $\varsigma_t \saconc_d v$.

      If $\ell(v)\in\B{:}\djcup\bar{\Sig}$, then
        either $\ell(\zeta)\in\B{:}\djcup\bar{\Sig}$
        or all synchrons in $a\varsigma(\zeta)$ are ${|_\Right}$-prefixed,
        so $\varsigma@\zeta = \varsigma = {|_\Left}\varsigma_t \in \varsigma(t|\target(w))$.
      Also, from $t \leadsto_v t'$ we have $t'=t$.
      Thus~(d) holds.

      Otherwise, using that ${|_\Left}\varsigma' \,@\, {|_\Left}\upsilon' = {|_\Left}\,\varsigma'@\upsilon'$ whenever $\varsigma' \sconc_d \upsilon'$,
        we have $\varsigma@\zeta={|_\Left}\,\varsigma_t@v$.
      From $t \ssleadsto_v t'$ we have $\varsigma_t@v\in\varsigma(t')$.
      Thus (d) holds.
    \item From $t \ssleadsto_v t'$ we have $\ell(t)=\ell(t') \lor \exists\, b\in\B.~ \{\ell(t),\ell(t')\}=\{b?,b{:}\}$.
      Also, $\ell(\chi)=\ell(t)$ and $\ell(\chi')=\ell(t')$.
      Thus (e) holds.
  \end{itemize}

  \vspace{1ex}
  \noindent
  \emph{Case 8.c}:
  The proof is similar to that for Cases 7.b and 8.a/b, and is omitted.

  \vspace{1ex}
  \noindent
  \emph{Case 9.a}:
  The proof is similar to that for Case 8.a and is omitted.

  \vspace{1ex}
  \noindent
  \emph{Case 9.b}:
  The proof is similar to that for Case 8.b and is omitted.

  \vspace{1ex}
  \noindent
  \emph{Case 9.c}:
  The proof is similar to that for Case 8.c and is omitted.

  \vspace{1ex}
  \noindent
  \emph{Case 10.a}:
  The proof is similar to that for Cases 8.b and 9.b, and is omitted.

  \vspace{1ex}
  \noindent
  \emph{Case 11.a}:
  The proof is similar to that for Case 8.a and is omitted.

  \vspace{1ex}
  \noindent
  \emph{Case 11.b}:
  The proof is similar to that for Case 11.a and is omitted.

  \vspace{1ex}
  \noindent
  \emph{Case 11.c}:
  The proof is similar to that for Case 3.a and is omitted.

  \vspace{1ex}
  \noindent
  \emph{Case 11.d}:
  The proof is similar to that for Case 11.c and is omitted.
\end{proof}

\begin{theorem}\label{thm:ssleadsto implies leadsto}\rm
  ${\ssleadsto} \subseteq {\leadsto}$.
  I.e., if $\chi \ssleadsto_\zeta \chi'$ then $\chi \leadsto_\zeta \chi'$.
\end{theorem}
\begin{proof}
  First suppose $\ell(\zeta)\in\B{:}\djcup\bar{\Sig}$.
  From $\chi \ssleadsto_\zeta \chi'$ we have $\source(\chi') = \target(\zeta) = \source(\zeta) = \source(\chi)$.
  By \df{leadsto}(1) and \lem{transition uniquely identified by source and synchrons}, it suffices to show $\varsigma(\chi')=\varsigma(\chi)$.

  By \df{ssleadsto} and \df{@}, $\varsigma(\chi') \supseteq \varsigma(\chi)@\zeta = \varsigma(\chi)$.
  Suppose $\varsigma \in \varsigma(\chi'){\setminus}\varsigma(\chi)$,
    then $\ell(\chi')\in\B?\djcup\B{:}$ and $\ell(\varsigma)\in\B?\djcup\B{:}$ by \lem{label condition}.
  Then $\varsigma(\chi)$ is not $P$-complete by \df{P-complete}; we have a contradiction with \thm{P-complete}.
  So $\varsigma(\chi')=\varsigma(\chi)$.

  \vspace{1ex}
  \noindent
  Now suppose $\ell(\zeta)\in\textit{Act}$.
  We proceed by structural induction on $\varsigma(\zeta)$,
    using the same well-founded order $<$ on sets of synchrons as in the proof of \thm{P-complete}.

  \vspace{1ex}
  \noindent
  \emph{Induction base}:
  Suppose \plat{$\varsigma(\zeta)=\{(\actsyn{\alpha}Q)\}$} for some $\alpha\in\textit{Act}$ and $Q\in\cT$.
  Then \plat{$\zeta=(\actsyn{\alpha}Q)$} and $\source(\chi) = \source(\zeta) = \alpha.Q$.
  By \df{ssleadsto}(\ref{ssaconc}) we have $\chi \ssaconc \zeta$.
  Since all synchrons in $\chi$ will be affected by $\zeta$, $n\varsigma(\chi)$ must be empty; \ie $\ell(\chi)\in\B?\djcup\B{:}$\,.
  So either (1) $\alpha=b?$ and \plat{$\chi=(\actsyn{\alpha}P)$} or (2) $\alpha\neq b?$ and $\chi=b{:}\alpha.P$, with $b\in\B$.
  In both cases, by \df{ssleadsto}(\ref{labels}), we have $\ell(\chi')\in\{b?,b{:}\}$.
  Thus $\chi \leadsto_\zeta \chi'$ by \df{leadsto}(2).

  \vspace{1ex}
  \noindent
  \emph{Induction step}:
  Given $\zeta\in\Tr$, we assume that (IH) $\Upsilon \ssleadsto_{\zeta'} \Upsilon' \implies \Upsilon \leadsto_{\zeta'} \Upsilon'$
    for all $\Upsilon,\zeta',\Upsilon'\in\Tr$ such that $\varsigma(\zeta')<\varsigma(\zeta)$.

  \vspace{1ex}
  \noindent
  Suppose $\varsigma(\zeta)={+_\Left}\Sigma_\Left$ for some $\Sigma_\Left$.
  Then $\zeta=v+Q$ and $\source(\chi) = \source(\zeta) = P+Q$
    for some $P,Q\in\cT$ and $v\in\en(P)$ with $\varsigma(v)=\Sigma_\Left$.
  So either (1) $\chi=t+Q$, (2) $\chi=P+u$ or (3) $\chi=t+u$, with $t\in\en(P)$ and $u\in\en(Q)$.

  In Cases (1) and (3), by \df{leadsto}(3) and (IH), it suffices to prove $t \ssleadsto_v \chi'$.
  \begin{itemize}
    \item $\source(t)=\source(v)$ follows from $\source(\chi)=\source(\zeta)$.
    \item $\source(\chi')=\target(v)$ follows from $\source(\chi')=\target(\zeta)$ and $\ell(\zeta)\in\textit{Act}$.
    \item $t \ssaconc v$ follows from $\chi \ssaconc \zeta$.
    \item From $\ell(\zeta)\in\textit{Act}$, we have $a\varsigma(v) \neq \emptyset \neq a\varsigma(\zeta)$.
      Using that ${+_\Left}\varsigma \,@\, {+_\Left}\upsilon = \varsigma@\upsilon$ whenever $\varsigma \sconc_d \upsilon$,
        we have $\varsigma(t)@v = \varsigma(\chi)@\zeta \subseteq \varsigma(\chi')$.
    \item $\ell(t)=\ell(\chi') \lor \exists\, b\in\B.~ \{\ell(t),\ell(\chi')\}=\{b?,b{:}\}$ follows from
      $\ell(\chi)=\ell(\chi') \lor \exists\, b\in\B.~ \{\ell(\chi),\ell(\chi')\}=\{b?,b{:}\}$ and $\ell(t)=\ell(\chi)$.
  \end{itemize}

  In Case (2), since all synchrons in $\chi$ will be affected by $\zeta$, $n\varsigma(\chi)$ must be empty; \ie $\ell(\chi)=b?$ for some $b\in\B$.
  Then we have $\chi'\in\{b?,b{:}\}$ by \df{ssleadsto}(\ref{labels}).
  Thus $\chi \leadsto_\zeta \chi'$ by \df{leadsto}(6).

  \vspace{1ex}
  \noindent
  Suppose $\varsigma(\zeta)={+_\Right}\Sigma_\Right$ for some $\Sigma_\Right$.
  The proof is similar to that of the previous case.

  \vspace{1ex}
  \noindent
  Suppose $\varsigma(\zeta)={|_\Left}\Sigma_\Left$ for some $\Sigma_\Left$.
  Then $\zeta=v|Q$ and $\source(\chi) = \source(\zeta) = P|Q$
    for some $P,Q\in\cT$ and $v\in\en(P)$ with $\varsigma(v)=\Sigma_\Left$.
  So either (1) $\chi=t|Q$, (2) $\chi=P|u$ or (3) $\chi=t|u$, with $t\in\en(P)$ and $u\in\en(Q)$.

  In Case (1), we have $\ell(\chi') = \ell(\chi) \in \Ch\djcup\bar{\Ch}\djcup\{\tau\}\djcup\Sig\djcup\bar{\Sig}$.
  By \lem{synchrons of transition} $\varsigma(\chi)=n\varsigma(\chi)$ and $\varsigma(\chi')=n\varsigma(\chi')$.
  Moreover $\varsigma(\chi') \supseteq \varsigma(\chi)@\zeta$ since $\chi \ssleadsto_\zeta \chi'$.
  Thus $\varsigma(\chi')=\varsigma(\chi)@\zeta$ by \lem{label condition}.
  Using that ${|_\Left}\varsigma \,@\, {|_\Left}\upsilon = {|_\Left}\,\varsigma@\upsilon$ whenever $\varsigma \sconc_d \upsilon$,
    we have that $\chi'=t'|Q$ for some $t'\in\en(\target(v))$ with $\varsigma(t')=\varsigma(t)@v$.
  Now $t \ssleadsto_v t'$, as all conditions of \df{ssleadsto} are trivially satisfied.
  Thus $\chi \leadsto_\zeta \chi'$ by (IH) and \df{leadsto}(8).

  In Case (2), as in case (1), using that ${|_\Right}\varsigma \,@\, {|_\Left}\upsilon = {|_\Right}\varsigma$,
    we have that $\chi'=\target(v)|u'$ for some $u'\in\en(Q)$ with $\varsigma(u')=\varsigma(u)$.
  By \lem{transition uniquely identified by source and synchrons}, we have $u'=u$.
  Thus $\chi \leadsto_\zeta \chi'$ by \df{leadsto}(7).

  In Case (3), we have
    $\varsigma(\chi)@\zeta = {|_\Left}\,\varsigma(t)@v \djcup {|_\Right}\varsigma(u) \subseteq \varsigma(\chi')$
    by \df{ssleadsto}(\ref{fourth}).
  Then $\zeta'=t'|u'$ for some $t'\in\en(\target(v))$ and $u'\in\en(Q)$ with $\varsigma(t') \supseteq \varsigma(t)@v$ and $\varsigma(u') \supseteq \varsigma(u)$.
  Note that $\varsigma(t')$ won't be empty regardless of whether $\ell(t)\in\B?\djcup\B{:}$\,.
  If $\varsigma(u') \supset \varsigma(u)$, then
    either, when $\ell(u) \in \Ch\djcup\bar{\Ch}\djcup\Sig\djcup\bar{\Sig}$, we have $\ell(u')=\tau$ by \lem{synchrons of transition}, preventing $t',u'$ from synchronising,
    or, when $\ell(u) \in \B!\djcup\B?\djcup\B{:}$, \df{P-complete} and \thm{P-complete} are violated.
  Hence we have $\varsigma(u')=\varsigma(u)$ and thus, by \lem{transition uniquely identified by source and synchrons}, $u'=u$.
  By \df{leadsto}(8) and (IH), it suffices to prove $t \ssleadsto_v t'$.
  All conditions except Item~\ref{labels} of \df{ssleadsto} are trivially satisfied.
  We now show that Condition~\ref{labels} also holds for $t,v,t'$.
  \begin{itemize}
    \item If $\ell(t) \in \Ch\djcup\bar{\Ch}\djcup\Sig\djcup\bar{\Sig}\djcup\B!$,
        then $n\varsigma(t)=\{\varsigma\}$ for some synchron $\varsigma$.
      By \df{ssleadsto}(\ref{ssaconc},\ref{fourth}), $\varsigma@v \in \varsigma(t)@v \subseteq \varsigma(t')$.
      If $n\varsigma(t') \supset \{\varsigma@v\}$, then
        either, when $\ell(t) \in \Ch\djcup\bar{\Ch}\djcup\Sig\djcup\bar{\Sig}$, we have $\ell(t')=\tau$ by \lem{synchrons of transition}, preventing $t',u'$ from synchronising,
        or, when $\ell(t) \in \B!$, \df{P-complete} and \thm{P-complete} are violated.
      Thus $n\varsigma(t')=\{\varsigma@v\}$ and, using \lem{synchrons of transition}, $\ell(t') = \ell(\varsigma@v) = \ell(\varsigma) = \ell(t)$.
    \item If $\ell(t)\in\{b?,b{:}\}$ and $\ell(u)=b!$ for some $b\in\B$, then $\ell(u')=b!$,
      so $\ell(t')\in\{b?,b{:}\}$, or $t',u'$ won't be able to synchronise.
    \item If $\ell(t)\in\{b?,b{:}\}$ and $\ell(u)\in\{b?,b{:}\}$ for some $b\in\B$, then $\ell(\chi)\in\{b?,b{:}\}$.
      By \df{ssleadsto}(\ref{labels}) we have $\ell(\chi')\in\{b?,b{:}\}$.
      Thus $\ell(t')\in\{b?,b{:}\}$ by \lem{synchrons of transition}.
  \end{itemize}

  \vspace{1ex}
  \noindent
  Suppose $\varsigma(\zeta)={|_\Right}\Sigma_\Right$ for some $\Sigma_\Right$.
  The proof is similar to that of the previous case.

  \vspace{1ex}
  \noindent
  Suppose $\varsigma(\zeta) = {|_\Left}\Sigma_\Left \djcup {|_\Right}\Sigma_\Right$ for some non-empty $\Sigma_\Left,\Sigma_\Right$.
  The proof is similar to that of the previous two cases.
  
  \vspace{1ex}
  \noindent
  Suppose $\varsigma(\zeta)={\backslash L}\,\Sigma'$ for some $L\subseteq\Ch\djcup\Sig$ and $\Sigma'$.
  The proof is similar to that for the case where $\zeta=v|Q$ and $\chi=t|Q$.

  \vspace{1ex}
  \noindent
  Suppose $\varsigma(\zeta)=[f]\Sigma'$ for some relabelling $f$ and $\Sigma'$.
  The proof is similar to that of the previous case.

  \vspace{1ex}
  \noindent
  Suppose $\varsigma(\zeta)=A{:}\Sigma'$ for some $A\in\A$ and $\Sigma'$.
  The proof is similar to that for the case where $\zeta=v+Q$ and $\chi=t+Q$.

  \vspace{1ex}
  \noindent
  Suppose $\varsigma(w) \mathbin{=} {{}\signals s}\,\Sigma'$ for some $r \mathop{\in} \Sig$ and $\Sigma'$.
  The proof is similar to that of the previous~case.
\end{proof}
}{}
\end{document}